\documentclass[plain,biber]{nowfnt} %
\usepackage[utf8]{inputenc}

\usepackage{pgfplots}
\pgfplotsset{compat=newest}
\usepackage{tikz}
\usetikzlibrary{arrows,matrix,positioning}
\usepackage{caption,subcaption}
\DeclareFieldFormat{sentencecase}{\MakeSentenceCase*{#1}}

\renewbibmacro*{title}{%
  \ifthenelse{\iffieldundef{title}\AND\iffieldundef{subtitle}}
    {}
    {\ifthenelse{\ifentrytype{article}\OR\ifentrytype{inbook}%
      \OR\ifentrytype{incollection}\OR\ifentrytype{inproceedings}%
      \OR\ifentrytype{inreference}}
      {\printtext[title]{%
        \printfield[sentencecase]{title}%
        \setunit{\subtitlepunct}%
        \printfield[sentencecase]{subtitle}}}%
      {\printtext[title]{%
        \printfield[titlecase]{title}%
        \setunit{\subtitlepunct}%
        \printfield[titlecase]{subtitle}}}%
     \newunit}%
  \printfield{titleaddon}}

\usepackage{color}
\usepackage{tabularx}
\usepackage{footnote}
\usepackage{array}
\usepackage{multirow}
\usepackage{arydshln}
\usetikzlibrary{patterns,calc}
\usepackage{scalerel}

\usepackage{amsmath,amsthm,amssymb,mathtools}
\usepackage[capitalise,noabbrev]{cleveref}
\crefname{Corollary}{corollary}{corollaries}
\newtheorem{problem}[theorem]{Problem}
\newtheorem{construction}{Construction}[section]

\pgfplotsset{compat=1.8}
\pgfplotsset{
mystyle/.style={
    scale only axis,
    width=0.7\columnwidth,
    height=0.55\columnwidth,
    label style={inner sep=0, font=\normalsize},
    tick label style={font=\scriptsize},
    legend style={font=\scriptsize},
    mark size=3,
    major grid style={dashed},
    line width=0.8pt,
    axis line style = thin}
}
\tikzset{
	>=stealth',
	mycircle/.style={circle, draw=gray, very thick, text width=.1em, minimum height=1.5em, text centered},
	mycircle_small/.style={circle,draw=awgray_dark,fill = awgray_dark, inner sep=0,minimum size=.6em},
	mycircle_small_black/.style={circle,draw=black,fill = black, inner sep=0,minimum size=.6em},
	mybox/.style={rectangle,rounded corners,draw=black, thick,text width=1em,minimum height=4em,minimum width=4em,text centered},
	mybox_small/.style={rectangle,rounded corners,draw=black, thick,text width=1em,minimum height=2em,minimum width=2em,text centered},
	mybox_vec/.style={rectangle,rounded corners,draw=black, thick,text width=1em,minimum height=0.7em, minimum width=4em,text centered},
	mybox_vec_short/.style={rectangle,rounded corners,draw=black, thick,text width=1em,minimum height=0.7em, minimum width=2em,text centered},
	pil/.style={->, thick, shorten <=2pt, shorten >=2pt,},
}
\usepackage[linesnumbered,ruled,vlined,titlenumbered]{algorithm2e}

\newcommand{\printalgoIEEE}[1]
{{\centering
\scalebox{0.97}{
\begin{algorithm}[H]
 \begin{small}
 #1
 \end{small}
\end{algorithm}
}
}
}

\newcommand{\Mooremat}[3]{\ensuremath{\mathcal{M}_{#1,#2}\left( #3 \right)}}
\newcommand{\MoormatExplicit}[3]{
	\begin{pmatrix}
		#1_{1} & #1_{2} & \dots& #1_{#3}\\
		#1_{1}^{q} & #1_{2}^{q} & \dots& #1_{#3}^{q}\\
		\vdots &\vdots&\ddots& \vdots\\
		#1_{1}^{q^{#2-1}} & #1_{2}^{q^{#2-1}} & \dots& #1_{#3}^{q^{#2-1}}\\
	\end{pmatrix}}

\newcommand{\F}{\ensuremath{\mathbb{F}}}
\newcommand{\Fq}{\ensuremath{\mathbb{F}_q}}
\newcommand{\Fqm}{\ensuremath{\mathbb{F}_{q^m}}}
\newcommand{\Fqmu}{\ensuremath{\mathbb{F}_{q^{mu}}}}
\newcommand{\Gal}{\ensuremath{\mathrm{Gal}}}

\DeclareMathOperator{\extsmallfield}{ext}
\DeclareMathOperator{\rank}{rk}
\DeclareMathOperator{\Tr}{Tr}
\DeclareMathOperator{\supp}{supp}
\newcommand{\nkhalffrac}{\left\lfloor \frac{n-k}{2}\right\rfloor}

\newcommand{\myspace}[1]{\ensuremath{\mathcal{#1}}}

\newcommand{\dimZ}{\ensuremath{\zeta}}

\providecommand{\Mat}[1]{\ensuremath{\mathbf{#1}}}

\newcommand{\ceil}[1]{\ensuremath{\left\lceil{#1}\right\rceil}}

\newcommand{\myspan}[1]{\left\langle #1 \right\rangle}
\newcommand{\dhalffrac}{\left\lfloor \frac{d-1}{2}\right\rfloor}

\newcommand{\Gaussbinom}[2]{\ensuremath{
{#1
\brack
#2}_q
}}
\newcommand{\eps}{\varepsilon}

\newcommand{\mycode}[1]{\ensuremath{\mathcal{#1}}}
\newcommand{\fontmetric}[1]{\mathsf{#1}}
\newcommand{\codeRank}[1]{\ensuremath{(#1)_{q^m}^\fontmetric{R}}}
\newcommand{\codelinearRank}[1]{\ensuremath{[#1]_{q^m}^\fontmetric{R}}}
\newcommand{\idealCode}[1]{\mathcal{IM}(#1)}

\newcommand{\Gabcode}[2]{\ensuremath{\mathcal{G}(#1,#2)}}
\newcommand{\IntGabcode}[1]{\ensuremath{\mathcal{IG}(#1)}}
\newcommand{\IntGab}{\ensuremath{\mathcal{IG}}}

\newcommand{\Gabmat}{\vec{G}_{\mathsf{G}}}
\newcommand{\Gabvec}{\vec{g}_{\mathsf{G}}}

\newcommand{\Randcode}{\mathcal{C}_{\mathsf{R}}}
\newcommand{\Randmat}{\vec{H}_{\mathsf{R}}}
\newcommand{\Randvec}{\vec{h}_{\mathsf{R}}}

\newcommand{\wR}{w_{\mathsf{R}}}

\newcommand{\numbRowErasures}{\varrho}
\newcommand{\numbColErasures}{\gamma}

\newcommand{\y}{\vec{y}}
\renewcommand{\u}{\vec{u}}
\renewcommand{\a}{\vec{a}}
\renewcommand{\v}{\vec{v}}
\renewcommand{\r}{\vec{r}}
\newcommand{\w}{\vec{w}}

\newcommand{\h}{\vec{h}}
\newcommand{\m}{\vec{m}}
\renewcommand{\c}{\vec{c}}
\newcommand{\e}{\vec{e}}

\newcommand{\I}{\vec{I}}
\newcommand{\G}{\vec{G}}
\renewcommand{\P}{\vec{P}}
\newcommand{\s}{\vec{s}}
\newcommand{\z}{\vec{z}}
\newcommand{\x}{\vec{x}}
\newcommand{\ZZ}{\mathbb{Z}}
\newcommand{\A}{\vec{A}}
\newcommand{\B}{\vec{B}}
\renewcommand{\H}{\vec{H}}
\newcommand{\X}{\vec{X}}
\newcommand{\Y}{\vec{Y}}

\newcommand{\cN}{\mathcal{N}}
\newcommand{\cP}{\mathcal{P}}
\newcommand{\cB}{\mathcal{B}}

\newcommand{\wtR}{\rank_{q}}

\newcommand{\setup}[1]{\ensuremath{\textsf{Setup}\left(#1\right)}}
\newcommand{\param}{\ensuremath{\textsf{param}}}
\newcommand{\keyGen}[1]{\ensuremath{\textsf{KeyGen}\left(#1\right)}}
\newcommand{\enc}[1]{\ensuremath{\textsf{Enc}\left(#1\right)}}
\newcommand{\dec}[1]{\ensuremath{\textsf{Dec}\left(#1\right)}}
\newcommand{\pk}{\ensuremath{\textsf{pk}}}
\newcommand{\secLevel}{\ensuremath{\textsf{SL}}}
\newcommand{\sk}{\ensuremath{\textsf{sk}}}
\newcommand{\ct}{\ensuremath{\textsf{ct}}}
\newcommand{\plainText}{\ensuremath{\vec{m}}}
\newcommand{\cipherText}{\ensuremath{\vec{c}}}
\newcommand{\distTrans}[1]{\ensuremath{\mathfrak{D}\left(#1\right)}}
\newcommand{\GL}[2]{\ensuremath{\mathsf{GL}_{#1}\left(#2\right)}}
\newcommand{\tdist}{\ensuremath{t_{X}}}

\newcommand{\decode}[2]{\ensuremath{\mathsf{Decode}_{#1}\left(#2\right)}}

\newcommand{\spannedBy}[1]{\ensuremath{\langle #1\rangle_q}}
\newcommand{\assignDet}{\ensuremath{\leftarrow}}
\newcommand{\assignRand}{\ensuremath{\xleftarrow{\$}}}

\newcommand{\qImg}[2]{\ensuremath{\text{Im}_{#1}(#2)}}

\newcommand{\evpolys}{\mathcal{P}^{n,k}_{\tVec,\hVec,\etaVec}}
\newcommand{\TRS}[3]{\mathcal{TG}_{#1}^{#3}[#2]}
\newcommand{\NN}{\mathbb{N}}
\newcommand{\alphaVec}{\boldsymbol{\alpha}}

\newcommand{\tVec}{{\boldsymbol{t}}}
\newcommand{\hVec}{\ensuremath{\boldsymbol{h}}}
\newcommand{\etaVec}{\ensuremath{\boldsymbol{\eta}}}

\newcommand{\qpow}[1]{^{[#1]}}
\newcommand{\twisted}{$[k,\tVec,\hVec,\etaVec]$-twisted }
\newcommand{\twistedC}{$[\alphaVec,\tVec,\hVec,\etaVec]$-twisted }
\newcommand{\Cmult}{\TRS{\alphaVec,\tVec,\hVec,\etaVec}{n,k}{}}
\newcommand{\numTwists}{\ell}

\newcommand{\code}{\ensuremath{\mathcal{C}}}

\newcommand{\tpub}{\ensuremath{t_{\mathsf{pub}}}}
\newcommand{\kpub}{\ensuremath{\mathbf{k}_{\mathsf{pub}}}}
\newcommand{\tikznode}[2]{%
	\ifmmode%
	\tikz[remember picture,baseline=(#1.base),inner sep=0pt] \node (#1) {$#2$};
	\else
	\tikz[remember picture,baseline=(#1.base),inner sep=0pt] \node (#1) {#2};%
	\fi}
\tikzset{%
	mybox_block/.style={rectangle,rounded corners,draw=black, thick,text width=1em,minimum height=2em,minimum width=4.75em,text centered},
	[highlight/.style={rectangle,rounded corners,fill=#1!15,draw,fill opacity=0.5,thick,inner sep=0pt},
	highlight/.default=gray],
	plot1/.style = {thick,
		dotted,
		mark=+}
}
\newif\ifcomment
\commentfalse

\newcommand{\lh}[1]{\ifcomment {\color{teal}[lh: #1]} \fi}

\newcommand{\set}[1]{\ensuremath{\mathcal{#1}}}
\renewcommand{\vec}[1]{\ensuremath{\mathbf{#1}}}
\newcommand{\ve}[1]{\vec{#1}}
\newcommand{\0}{\ensuremath{\ve{0}}}
\newcommand{\1}{\ensuremath{\ve{1}}}

\newcommand{\cA}{\mathcal{A}}

\newcommand{\cE}{\mathcal{E}}

\newcommand{\cI}{\mathcal{I}}
\newcommand{\cW}{\mathcal{W}}

\newcommand{\intervallincl}[2]{\ensuremath{[#1,#2]}}
\newcommand{\dimAmbSpace}{\ensuremath{N}}
\newcommand{\myref}[2]{\hyperref[#2]{#1~\ref{#2}}}
\DeclareMathOperator{\rk}{rk}

\newcommand{\ProjspaceAny}[1]{\mathbb{P}_q(#1)}

\definecolor{TUMdgray}{RGB}{088,088,090}
\definecolor{TUMmgray}{RGB}{156,157,159}
\definecolor{TUMlgray}{RGB}{215,217,218}
\definecolor{TUMyellow}{RGB}{255,180,000}
\definecolor{TUMorange}{RGB}{255,128,000}
\definecolor{TUMblue}{RGB}{000,101,189}
\definecolor{TUMgreen}{RGB}{0,124,48}
\definecolor{TUMred}{RGB}{196,7,27}
\definecolor{TUMgold}{RGB}{255,200,0}

\definecolor{blueskydeep}{RGB}{0,191,255}
\definecolor{lightcoral}{RGB}{240,128,128}
\definecolor{limegreen}{RGB}{50,205,50}

\definecolor{mycyan}{rgb}{0.30100,0.74500,0.93300}%
\definecolor{mygreen}{rgb}{0.46600,0.67400,0.18800}%

\newcommand{\myspaceDual}[1]{\ensuremath{#1}^{\perp}}
\newcommand{\Rowspace}[1]{\ensuremath{\left\langle #1\right\rangle_{\!q}}}

\newcommand{\RowspaceHuge}[1]{\scaleleftright[3ex]{\Biggl\langle}{#1}{\Biggr\rangle}\!\!\!\raisebox{-27pt}{$q$}} %

\newcommand{\Colspace}[1]{\Rowspace{#1}^c}%
\newcommand{\Grassm}[1]{\mathbb{G}_q(#1)}

\newcommand{\nTransmit}{\ensuremath{n_t}}
\newcommand{\nReceive}{\ensuremath{n_r}}
\newcommand{\delOp}[1]{\ensuremath{\mathfrak{H}_{#1}}}
\newcommand{\insertions}{\ensuremath{\gamma}}
\newcommand{\deletions}{\ensuremath{\delta}}

\newcommand{\Rankdist}[1]{d_r(#1)}
\newcommand{\RankdistNoInput}{d_r}
\newcommand{\Subspacedist}[1]{d_{\mathsf{S}}(#1)}
\newcommand{\Injectiondist}[1]{d_I(#1)}
\newcommand{\SubspacedistNoInput}{d_s}

\newcommand{\normSubspacedist}{\ensuremath{\eta}}

\newcommand{\volSubBall}[1]{\ensuremath{V_S(#1)}}
\newcommand{\volSubBallVardy}[1]{\ensuremath{\tilde{V}_S(#1)}}

\newcommand{\List}{\set{L}}

\newcommand{\avgListSizeCS}[1]{\ensuremath{\bar{\List}(#1)}}
\newcommand{\avgListSizeCSstar}[1]{\ensuremath{\bar{\List}'(#1)}}
\newcommand{\avgListSizeIS}[1]{\ensuremath{\bar{\List}_\text{I}(#1)}}

\newcommand{\avgListSizeFS}[1]{\ensuremath{\bar{\List}_\text{F}(#1)}}

\newcommand{\colVec}[2]{\ensuremath{#1^{#2\times1}}}

\newcommand{\Gab}[1]{\ensuremath{\mathcal{G}(#1)}}
\newcommand{\IntSub}[1]{\ensuremath{\mathcal{I}\mathcal{S}[#1]}}
\newcommand{\FGab}[1]{\ensuremath{\mathcal{F}\mathcal{G}[#1]}}
\newcommand{\FSub}[1]{\ensuremath{\mathcal{F}\mathcal{S}[#1]}}
\newcommand{\CSub}{\ensuremath{\mathcal{C}_s}}
\newcommand{\CSubComp}{\ensuremath{\mathcal{C}_s^{\perp}}}
\newcommand{\CRank}{\ensuremath{\mathcal{C}_r}}

\newcommand{\Basis}{\ensuremath{\myset{B}}}

\newcommand{\NormbasisOrdered}{\boldsymbol{\Normelement}}

\newcommand{\Normelement}{\beta}

\newcommand{\mymap}[1]{\textup{{#1}}}

\newcommand{\extsmallfieldinput}[1]{\ensuremath{\mymap{ext}_{\NormbasisOrdered}\left(#1\right)}}
\newcommand{\extsmallfieldinputInverse}[1]{\ensuremath{\mymap{ext}^{-1}_{\NormbasisOrdered}\left(#1\right)}}

\newcommand{\lifting}{\Pi}
\newcommand{\liftingMap}[2]{\ensuremath{\lifting_{#1}\!\left(#2\right)}}

\newcommand{\vecFq}[1]{\ensuremath{\underline{#1}}}
\newcommand{\vecalpha}{\ensuremath{\boldsymbol{\alpha}}}

\newcommand{\foldedVec}{\ensuremath{\vecalpha_{F}}}
\newcommand{\intOrder}{\ensuremath{L}}

\newcommand{\foldPar}{\ensuremath{h}}
\newcommand{\intDim}{\ensuremath{s}}

\newcommand{\ambSpace}{\ensuremath{W_s}}
\newcommand{\pe}{\ensuremath{\alpha}}
\newcommand{\lenFG}{\ensuremath{g}}

\newcommand{\Linpolyring}{\mathbb{L}_{q^m}[x]}

\newcommand{\LinpolyringK}{\mathbb{L}_{q^m}[x]_{<k}}

\newcommand{\quadbinom}[2]{\left[\genfrac{}{}{0pt}{}{#1\vphantom{N_N}}{#2\vphantom{N}}\right]}
\newcommand{\AntiCode}[1]{\ensuremath{\mathcal{A}\left(#1\right)}}

\newcommand{\myset}[1]{\mathcal{#1}}
\newcommand{\vecelements}[1]{\ensuremath{(#1_0 \ #1_1 \ \dots \ #1_{n-1})}}
\newcommand{\vecelementsm}[1]{\ensuremath{(#1_0 \ #1_1 \ \dots \ #1_{m-1})}}

\newcommand{\vecelementsArb}[2]{\ensuremath{(#1_0 \ #1_1 \ \dots \ #1_{#2-1})}}
\newcommand{\MoorematNoInput}{\mymap{qvan}}
\newcommand{\MinSubspacePoly}[1]{M_{#1}(x)}

\newcommand{\Fqs}{\ensuremath{\mathbb F_{q^s}}}

\newcommand{\Ball}[2]{\mathcal{B}_\sfR^{(#1)}(#2)}
\newcommand{\Sphere}[2]{\mathcal{S}_\sfR^{(#1)}(#2)}

\newcommand{\maxCardinalityRank}[1]{\mycode{A}_{q^m}^\sfR\left(#1\right)}
\newcommand{\Uniquecorrcap}{\ensuremath{\tau_0}}
\newcommand{\MRD}[1]{\ensuremath{\mycode{MRD}(#1)}}
\newcommand{\MRDlinear}[1]{\ensuremath{\mycode{MRD}[#1]}}
\newcommand{\veceval}[2]{\ensuremath{\left(#1(#2_0) \ #1(#2_1) \ \dots \ #1(#2_{n-1})\right)}}
\newcommand{\vecevalm}[2]{\ensuremath{\left(#1(#2_0) \ #1(#2_1) \ \dots \ #1(#2_{m-1})\right)}}
\newcommand{\intervallexcl}[2]{\ensuremath{[#1,#2-1]}}
\newcommand{\setelements}[2]{\ensuremath{\{#1_0,#1_1,\dots, #1_{#2-1}\}}}

\DeclareMathOperator{\wt}{wt}

\newcommand{\sfH}{\mathsf{H}}
\newcommand{\sfR}{\mathsf{R}}
\newcommand{\sfS}{\mathsf{S}}

\newcommand{\dminHam}{d^{(\sfH)}}
\newcommand{\dminRank}{d^{(\sfR)}}

\newcommand{\wtHam}{\wt_{\sfH}}

\newcommand{\Hlocal}{\H^{\mathsf{(local)}}}
\newcommand{\Hglobal}{\H^{\mathsf{(global)}}}

\newcommand{\eglob}{s} %

\newcommand{\bbE}{\mathbb{E}}
\newcommand{\bbElocal}{\mathbb{E}^{\mathsf{(local)}}}

\title{Rank-Metric Codes and Their Applications}

\maintitleauthorlist{
  Hannes Bartz\\
  German Aerospace Center (DLR) \\
  hannes.bartz@dlr.de
  \and
  Lukas Holzbaur \\
  Technical University of Munich \\
  lukas.holzbaur@tum.de
  \and
  Hedongliang Liu \\
  Technical University of Munich \\
  lia.liu@tum.de
  \and
  Sven Puchinger\\
  Hensoldt Sensors GmbH\\
  mail@svenpuchinger.de
  \and
  Julian Renner\\
  Technical University of Munich \\
  julian.renner@tum.de
  \and
  Antonia Wachter-Zeh\\
  Technical University of Munich \\
  antonia.wachter-zeh@tum.de
}

\issuesetup
{%
 copyrightowner={A.~Heezemans and M.~Casey},
 volume        = xx,
 issue         = xx,
 pubyear       = 2021,
 isbn          = xxx-x-xxxxx-xxx-x,
 eisbn         = xxx-x-xxxxx-xxx-x,
 doi           = 10.1561/XXXXXXXXX,
 firstpage     = 1, %
 lastpage      = 18
 }

\author[1]{Hannes Bartz}
\author[2]{Lukas Holzbaur}
\author[2]{Hedongliang Liu}
\author[2]{Sven Puchinger}
\author[2]{Julian Renner}
\author[2]{Antonia Wachter-Zeh}

\affil[1]{German Aerospace Center (DLR); hannes.bartz@dlr.de}
\affil[2]{Technical University of Munich; lukas.holzbaur@tum.de}
\affil[3]{Technical University of Munich; lia.liu@tum.de}
\affil[4]{Hensoldt Sensors GmbH; mail@svenpuchinger.de}
\affil[5]{Technical University of Munich; julian.renner@tum.de}
\affil[6]{Technical University of Munich; antonia.wachter-zeh@tum.de}

\articledatabox{\nowfntstandardcitation}

\begin{document}

\makeabstracttitle
\begin{abstract}
  The rank metric measures the distance between two matrices by the rank of their difference.
  Codes designed for the rank metric have attracted considerable attention in recent years, reinforced by network coding and further motivated by a variety of applications.
  In code-based cryptography, the hardness of the corresponding generic decoding problem can lead to systems with reduced public-key size. In distributed data storage, codes in the rank metric have been used repeatedly to construct codes with locality, and in coded caching, they have been employed for the placement of coded symbols.
  This survey gives a general introduction to rank-metric codes, explains their most important applications, and highlights their relevance to these areas of research.
\end{abstract}

\chapter{Introduction}
Codes composed of matrices are a natural generalization of codes composed of vectors.
Codes in the rank metric of length $n \leq m$ can be considered as a set of $m \times n$ matrices over a finite field $\Fq$ or equivalently as a set of vectors of length $n$ over the extension field $\Fqm$.
The rank weight of each codeword vector is the rank of its matrix representation and the rank distance between two matrices is the rank of their difference. These definitions rely on the fact that the rank distance is indeed a metric.
Several code constructions and basic properties of the rank metric show strong similarities to codes in the Hamming metric. However, there are also notable differences, e.g., in the list decoding properties.

Error-correcting codes in the rank metric were first considered by \citet{Delsarte_1978}, who proved a Singleton-like upper bound on the cardinality of rank-metric codes and constructed a class of codes achieving this bound\footnote{In analogy to MDS codes, such codes are called Maximum Rank Distance (MRD) codes.}.
This class of codes was reintroduced by \citet{Gabidulin_TheoryOfCodes_1985} in his fundamental paper \emph{``Theory of Codes with Maximum Rank Distance''}. Further, in his paper several properties of codes in the rank metric and an efficient decoding algorithm based on an equivalent of the Euclidean algorithm were shown. Since Gabidulin's publication contributed significantly to the development of error-correcting codes in the rank metric,
the most famous class of codes in the rank metric --- the equivalents of Reed--Solomon codes --- are nowadays called \emph{Gabidulin codes}. These codes can be defined by evaluating non-commutative {linearized polynomials}, proposed by Ore \citep{Ore_OnASpecialClassOfPolynomials_1933,Ore_TheoryOfNonCommutativePolynomials_1933}.
Independently of the previous work, \citet{Roth_RankCodes_1991} discovered in 1991 codes in the rank metric and applied them for correcting crisscross error patterns.

The goal of this survey is to provide an overview of the known properties of rank-metric codes and their application to problems in different areas of coding theory and cryptography.

Chapter~\ref{ch:introRankMetric} provides a brief introduction to rank-metric codes, their properties and their decoding. After providing basic notations for finite fields and linearized polynomials, we consider codes in the rank metric. We first define the rank metric and give basic properties and bounds on the cardinality of codes in the rank metric (namely, equivalents of the Singleton, sphere-packing, and the Gilbert–Varshamov bounds). Then, we define Gabidulin codes, show that they attain the Singleton-like upper bound on the cardinality and give their generator and parity-check matrices.
We describe their decoding up to half the minimum rank-distance by syndrome-based decoding. A summary of how to accomplish this efficiently is given and the problem of error-erasure correction is considered. We also give an overview on list decoding of Gabidulin codes and consider interleaved and folded Gabidulin codes. Finally, further classes of rank-metric codes such as twisted Gabidulin codes are briefly discussed. %

Rank-metric codes have several applications in communications and security, including public-key code-based cryptography.
In 1978, Rivest, Shamir and Adleman (RSA)~\citep{rivest1978method} proposed the first public-key cryptosystem in order to guarantee secure communication in an asymmetric manner. Since then, public-key cryptography is essential to protect data via encryption, to enable secure key exchange for symmetric encryption, and to protect the authenticity and integrity of data via digital signature schemes.
Only one year after the RSA cryptosystem {was introduced}, whose security relies on the hardness of the integer factorization problem, \citet{mceliece1978public} proposed the first public key cryptosystem based on error-correcting codes.
In his pioneering work McEliece showed that hard problems in coding theory can be used to derive public-key cryptosystems.
A crucial drawback of the McEliece cryptosystem compared to other public-key cryptosystems, such as RSA or elliptic curve cryptosystems (ECC), is its large public-key size.
The recent developments in quantum computing rendered all of the currently used public-key cryptosystems whose security relies on the integer factorization or the discrete logarithm problem insecure.
In particular, Shor's algorithm~\citep{shor1999polynomial} allows to solve both, the integer factorization problem and the discrete logarithm problem in polynomial time, which in turn allows to break the corresponding public-key cryptosystems completely, given a sufficiently large quantum computer.
Since code-based public-key cryptosystems are resilient against all known attacks on quantum computers, including Shor's algorithm, they are considered to be \emph{quantum-resistant} (or post-quantum secure) cryptosystems.
Quantum-resistant cryptography is an important research area to ensure the long-term security of transmitted and stored data.
Therefore, the National Institute of Standards and Technology (NIST) opened a standardization call, which meanwhile has reached its final round~\citep{NIST2017post}.
In order to reduce the public-key size, many new McEliece variants based on several codes were proposed, both before and independent of the NIST competition and also as submissions to the NIST competition. This includes a long history of variants based on codes in the rank metric.
The first McEliece variant in the rank metric was proposed by \citet{gabidulin1991ideals} and is therefore known as the GPT cryptosystem.
Although no rank-metric based schemes are among the finalists, rank-metric based schemes are considered as potential candidates for future standards~\citep{nist2020status}.

Chapter~\ref{chap:crypto} gives an overview of rank-metric code-based quantum-resistant encryption and authentication schemes.
First, hard problems which can be used to design rank-metric code-based cryptosystems are considered.
Then, a general framework to define most GPT variants is given, and the particular variants are described.
Finally, an overview on non-GPT-like cryptosystems, including the NIST submission Rank Quasi Cyclic (RQC), and rank-metric code-based signature schemes is given.

Rank-metric codes find applications not only in the cryptographic protection of data, but also in ensuring its integrity.
The increase in the amount of data that is stored by distributed storage systems has motivated a transition from replication of the data to the use of more involved storage codes, most commonly Maximum Distance Separable (MDS) codes. By storing one symbol of a codeword on each node, a node failure then corresponds to a symbol erasure and the Hamming distance of the storage code provides a guarantee on the number of failures the system can tolerate before data loss occurs. However, as the number of nodes in these systems grows, not only the number of tolerable node failures, but also the efficiency of the node repair process becomes a concern. Codes with locality \citep{Huang2007, chen2007maximally,gopalan2012locality} address this issue by reducing the number of nodes required for repair in the more likely event of a single or small number of node failures. While these codes are designed for the Hamming metric, codes for the rank metric, in particular, Gabidulin codes have repeatedly been used to construct these codes, especially for the stronger notion of maximally recoverable (MR) locally repairable/recoverable codes (LRCs)\footnote{MR LRCs are also referred to as partial MDS (PMDS) codes.}.
Further, rank-metric codes have also been used in another area related to distributed storage, referred to as coded caching.  Caching is a commonly used strategy to reduce the traffic rate during the peak hours. The communication procedure consists of two phases: placement and delivery. The seminal work by~\citet{MAN14} has shown that applying coding merely in the delivery phase can reduce the traffic rate. As a further improved scheme~\citep{YMA18} has been shown to be order-optimal under uncoded placement, schemes with coded placement~\citep{CFL16,GV18} become of interest in order to further reduce the traffic rate during the delivery phase. Rank-metric codes have been utilized in the scheme with coded placement by \citet{TC18}, which has been shown to outperform the optimal scheme with uncoded placement~\citep{YMA18} in the regime of small cache size.

In Chapter~\ref{chap:storage}, the application of rank-metric codes to distributed data storage is outlined.
First, we explore the connection between codes with locality and rank-metric codes by providing a high-level description of the property exploited by many constructions of (MR) LRCs.
Second, we present the application of Maximum Rank Distance (MRD) %
codes in the coded caching scheme by~\citet{TC18}. %

Network coding has been attracting attention since the fundamental works by \citet{Ahlswede_NetworkInformationFlow_2000} and \citet{LiYeungCai-LinearNetworkCoding_2003} showed that the capacity of multicast networks can be achieved by performing linear combinations of packets instead of just forwarding them. Rank-metric codes have been used in network coding solutions~\citep{etzion2018vector} and error correction in coherent networks~\citep{SilvaKschischang-MetricsErrorCorrectionNetworkCoding_2009}. For random networks, rank-metric codes are used to correct errors by the lifting construction~\citep{silva_rank_metric_approach}. In addition, the subspace metric~\citep{koetter_kschischang} was introduced for error control, as this metric perfectly captures the type of errors that occur in (random) linear network coding. Due to the close relation between the rank metric and the subspace metric, rank-metric codes are a natural choice to construct subspace codes for error control in random network coding.

Chapter~\ref{chap:network_coding} introduces constructions of network codes based on MRD codes, constructions of subspace codes by lifting rank-metric codes, bounds on the cardinality, and the list decoding capability of subspace codes.
We first present constructions based on MRD codes for a class of
deterministic multicast networks, which guarantee that all the receivers decode all the messages. Two error models commonly considered in networks are described. We introduce subspace codes, with a focus on constructions based on lifting rank-metric codes and provide upper bounds on the size of subspace codes. Further, an analysis of list decoding subspace codes is provided.

Finally, Chapter~\ref{chap:concl} concludes this survey and shortly mentions further applications of rank-metric codes.

\chapter{Basics on Rank-Metric Codes}\label{ch:introRankMetric}

Codes in the rank metric have drawn increasing interest in recent years due to their application to cryptography, distributed storage and network coding. The survey by~\citet{gorla2018codes} mainly focuses on combinatorial properties of rank-metric codes, while the one by~\citet{sheekey201913} considers MRD codes and their properties. In particular,~\citet{gorla2018codes, gorla2019rankmetric} treat (amongst others): isometries, anticodes, duality, MacWilliams identities, generalized weights in the rank metric. We will therefore not considers these topics in this survey.  An English version of the textbook by~\citet{GabidulinVladmir2021} has been published recently which contains a collection of Gabidulin's results in the area of rank-metric codes. Our survey deals shortly with properties of rank-metric codes, but the main focus is on their decoding and their applications to code-based cryptography, storage, and network coding.

In this chapter, we give an introduction to rank-metric codes, their properties and decoding.
In Section~\ref{sec:notation}, we provide the basic notation used in this survey.
Section~\ref{sec:lin-poly} defines linearized polynomials and recalls some basic properties.
In Section~\ref{sec:rank-metric}, we define the rank metric and state bounds on the cardinality of rank-metric codes, i.e., sphere-packing, Gilbert--Varshamov, and Singleton-like bounds.
Section~\ref{sec:weightDist}, \ref{sec:constantRank} and \ref{sec:coveringProperty} provide the weight distribution of rank metric codes, constant-rank codes and covering property of rank-metric codes, respectively.
Section~\ref{sec:gabidulin} defines Gabidulin codes, proves their minimum rank distance and gives generator and parity-check matrices.
In Section~\ref{sec:decoding_gabidulin_codes}, we describe syndrome-based decoding of Gabidulin codes, i.e., we prove the key equation, show how to solve it and how to reconstruct the final error. We also outline error-erasure decoding and summarize known fast decoders. In Section~\ref{sec:list-dec-gab}, we discuss results on list decoding of Gabidulin codes. Further, we introduce interleaved Gabidulin codes (Section~\ref{sec:interleaved}), folded Gabidulin codes (Section~\ref{subsec:FoldedGabidulin}) and summarize further classes of MRD codes (Section~\ref{sec:furtherclasses}).

\section{Notation}\label{sec:notation}
This section introduces notation that is used throughout the survey.

We denote a finite field of size $q$ by $\F_q$ %
and the set of all row vectors and matrices over this field by $\F_q^{n}$ and $\F_q^{m\times n}$, respectively. The ring of integers is given by $\mathbb{Z}$ and the non-negative integers by $\mathbb{Z}_{\geq 0}$. The set of integers $\{i \ | \ a \leq i \leq b\}$ is denoted by $[a,b]$ or by $[b]$ if $a=1$. %
To simplify the notation, we denote $\alpha^{q^i} = \alpha^{[i]}$.

For a linear code of length $n$, dimension $k$ over $\Fqm$, and $\Fq$-rank distance $d$ we write $\codelinearRank{n,k,d}$ and for a non-linear code of cardinality $M$ we write $\codeRank{n,M,d}$. Similarly, a code of Hamming distance $d$ over $\Fqm$ is denoted $[n,k,d]_{q^m}^\sfH$. The rank and Hamming distance of a code $\code$ are given by $d_\sfR(\code)$ and $d_\sfH(\code)$, respectively.

The entries of a matrix $\A \in \F_q^{m\times n}$ are given by $A_{i,j}$ for $i\in [1,m],j\in [1,n]$ and the entries of a vector $\a\in \F^{n}$ by $a_i$ for $i\in [1,n]$. Denote by $\spannedBy{\A}$ the \emph{row-space} of $\A$, i.e., the $\Fq$-linear vector space spanned by the rows of ${\A}$. Similarly, $\Colspace{\A}$ denotes its column space.
We define a mapping from $\Fqm^n$ to $\F_q^{m \times n}$ by
   \begin{align*}
\extsmallfield_{\NormbasisOrdered}:\Fqm^{n} &\mapsto \Fq^{m \times n}\label{eq:mapping_smallfield}\\
  \vec{a} = (a_1,\hdots,a_n) &\mapsto \vec{A} =
                               \begin{pmatrix}
                                 A_{1,1} & \hdots & A_{1,n} \\
                                 \vdots & \ddots & \vdots \\
                                 A_{m,1} &  \hdots & A_{m,n} \\
                               \end{pmatrix},
\end{align*}
where $\NormbasisOrdered = (\beta_1,\beta_2,\dots,\beta_{m})$ is a basis of $\Fqm$ over $\Fq$ and
\begin{equation*}
\quad a_j = \sum_{i=1}^{m} A_{i,j} \beta_i, \quad \forall j \in [1,n].
\end{equation*}
The weight of a vector $\a \in \Fqm^{n}$ in the rank metric is its $\F_q$-rank, which is defined as $\rank_q(\a) = \rank(\extsmallfield_{\NormbasisOrdered}(\a))$. %

The trace operator is given by
\begin{align*}
  \Tr_{q^m/q}:\Fqm^{n} &\rightarrow \Fq^{n}\\
\vec{a} = (a_1,\hdots,a_n) &\mapsto \Bigg( \sum_{i=0}^{m-1} a_1^{[i]}, \hdots, \sum_{i=0}^{m-1} a_n^{[i]} \Bigg).
\end{align*}

The Gaussian binomial coefficient, i.e., the number of $r$-dimensional subspaces of the vectorspace $\F_q^s$, is given by
\begin{equation*}
  \Gaussbinom{s}{r}\coloneqq
  \begin{cases}
    \frac{ (1-q^s)(1-q^{s-1})\hdots(1-q^{s-r+1})}{(1-q)(1-q^2)\hdots(1-q^{r})} &\text{ for $r \leq s$} \\
    0 &\text{ for $r> s$},
    \end{cases}
\end{equation*}
where $s$ and $r$ are non-negative integers.
The collection of these subspaces of dimension $\dimAmbSpace$ of the ambient space is denoted $\ProjspaceAny{\dimAmbSpace}$ and
$\Grassm{\dimAmbSpace,i}$ denotes the set of subspaces of dimension $i$ in $\ProjspaceAny{\dimAmbSpace}$.
Hence, $\ProjspaceAny{\dimAmbSpace} = \cup_{i=1}^N \Grassm{\dimAmbSpace,i}$.

\section{Linearized Polynomials}\label{sec:lin-poly}
Linearized polynomials constitute a \emph{non-commutative} ring and will later provide the definition of \emph{Gabidulin codes}. %
Apart from their application to coding theory, linearized polynomials are used, e.g., in root-finding of \emph{usual} polynomials and as permutation polynomials in cryptography.

They are also called $q$-polynomials and were introduced in 1933 by
\citet{Ore_OnASpecialClassOfPolynomials_1933} as a special case of \emph{skew polynomials} \citep{Ore_TheoryOfNonCommutativePolynomials_1933}.
The theory of skew polynomials is quite rich and widely investigated \citep{Jacobson-TheTheoryOfRings_1943,Giesbrecht-FactoringSkewPolynomials,Jacobson-FiniteDimensionalDivisonAlgebrasOverFields_2010}. It is possible
to construct error-correcting codes based on skew polynomials \citep{BoucherGeiselmannUlmer-SkewCyclicCodes_2007,BoucherUlmer-CodingWithSkewPolynomialRings_2009,BoucherUlmer-CodesAsModulesoverSkewPolynomialRings,ChaussadeLoidreauUlmer-SkewCodesPrescribedDistanceRank_2009,BoucherUlmer-LinearCodesSkewPolynomialsAutomorphisms}.

Linearized polynomials are defined as follows.

\begin{definition}[Linearized Polynomial]\label{def:linearized_polynomial}
	A polynomial $a(x) $ is a linearized polynomial
	if it has the form
	\begin{equation*}
	a(x) = \sum_{i=0}^{d_a} a_i x^{[i]}, \quad a_i \in \Fqm \ \forall i \in \intervallincl{0}{d_a}.
	\end{equation*}
	The non-commutative univariate linearized polynomial ring with indeterminate $x$, consisting of all such polynomials over $\Fqm$, is denoted by $\Linpolyring$.
\end{definition}

	The $q$-degree of $a(x)$ is defined to be the largest $i \in [0,d_a]$ such that $a_i \neq 0$.

\begin{remark}
Linearized polynomials are a special case of skew polynomials\footnote{Skew polynomials become linearized polynomials when the derivation is zero and the Frobenius automorphism is used, i.e., when we consider only $\Fq$-linear maps. }, which were also introduced by \citet{Ore_TheoryOfNonCommutativePolynomials_1933}.
Since a lot of literature on rank-metric codes defines rank-metric codes as evaluation codes of the more general class of skew polynomials, we briefly outline the connection here.

Let $\Fqm$ be a field extension of $\Fq$. The Galois group of the field extension is denoted by $\Gal(\Fqm/\Fq)$, and consists of all automorphisms $\sigma$ of $\Fqm$ that fix the small field $\Fq$, i.e., $\sigma(a) = a$ for all $a \in \Fq$. For finite fields, the Galois group consists of all powers of the Frobenius automorphism $\phi_q \, : \, \Fqm \to \Fqm, \, a \mapsto a^q$, i.e.,
\begin{align*}
\Gal(\Fqm/\Fq) = \left\{ \phi_q^i \, : \, 0 \leq i < m \right\}.
\end{align*}
The skew polynomial ring w.r.t.~$\Fqm$ and $\sigma \in \Gal(\Fqm/\Fq)$ is the set of polynomials of the form
\begin{align*}
a = \sum_{i=0}^{d_a} a_i x^i,
\end{align*}
where $a_i \in \Fqm$ and $d_a \in \ZZ_{\geq 0}$. Addition is defined as usual, i.e., component-wise, but multiplication is defined using the multiplication rule $x \cdot a \coloneqq \sigma(a) x$ for all $a \in \Fqm$, and extended to arbitrary degree polynomials by associativity and distributivity.
Hence, the closed-form expression for the multiplication of two polynomials $a,b \in \Fqm[x;\sigma]$ is
\begin{equation}
a \cdot b = \sum_{i} \left( \sum\limits_{j=0}^{i} a_j \sigma^j (b_{i-j}) \right) x^i.
\end{equation}
This multiplication rule is in general non-commutative.
The degree of a skew polynomial is defined by $\deg a \coloneqq \max\{i \, : \, a_i \neq 0\}$ for $a \neq 0$, and $\deg a \coloneqq -\infty$ for $a =0$.
The operator evaluation of a skew polynomial $a$ is the map
\begin{align*}
a(\cdot) \, : \, \Fqm &\to \Fqm, \\
\alpha &\mapsto \sum_i a_i \sigma^i(\alpha).
\end{align*}
Note that there are several ways to define evaluation of skew polynomials, see, e.g., \citep{BoucherUlmer-LinearCodesSkewPolynomialsAutomorphisms}.
Other types of evaluations, such as the remainder evaluation, also have applications in coding theory.

The ring of linearized polynomials is isomorphic to the ring of skew polynomials with Frobenius automorphism $\sigma = \phi_q$ through the obvious isomorphism
\begin{align*}
\varphi \, : \, \Linpolyring &\to \Fqm[x;\sigma], \\
\sum_{i=0}^{d_a} a_i x^{q^i} &\mapsto \sum_{i=0}^{d_a} a_i x^{i}.
\end{align*}
This can be easily seen by replacing $\sigma$ by the Frobenius automorphism in the multiplication rule formula.
Furthermore, the degree of a skew polynomial equals the $q$-degree of its corresponding linearized polynomial, and the operator evaluation of a skew polynomial equals the ordinary evaluation of its corresponding linearized polynomial. \hfill $\diamond$
\end{remark}

Recall that for any $B \in \Fq$, $B^{[i]} = B$ holds for any integer $i$.
This provides the following lemma about evaluating linearized polynomials.

\begin{lemma}[Evaluation of a Linearized Polynomial {\citep[Theorem~11.12]{Berlekamp1984Algebraic}}]\label{lem:evaluate_linearizedpoly}
	Let $\Basis=\{\beta_0,\beta_1,\dots,\beta_{m-1}\}$ be a basis of $\;\Fqm$ over $\Fq$,
	let $a(x)$ be a linearized polynomial as in Definition~\ref{def:linearized_polynomial} and let $b \in \Fqm$.
	Denote $\extsmallfieldinput{b} = \vecelementsm{B}^\top \in \Fq^{m \times 1}$. %
	Then,
	\begin{equation*}
	a(b) = \sum_{i=0}^{m-1} B_i a\big(\beta_i\big).
	\end{equation*}
\end{lemma}
Lemma~\ref{lem:evaluate_linearizedpoly} establishes the origin of the name \emph{linearized} polynomials:
for all $A_1,A_2 \in \Fq$, all $b_1,b_2 \in \Fqm$, and $a(x) \in \Linpolyring$ it holds that
\begin{equation*}
a\big(A_1 b_1+A_2 b_2\big) = A_1 a\big(b_1\big)+ A_2 a\big(b_2\big).
\end{equation*}
Hence, any $\Fq$-linear combination of roots of a linearized polynomial $a(x)$ is also a root of $a(x)$.

\begin{theorem}[Roots of a Linearized Polynomial {\citep[Theorem~11.31]{Berlekamp1984Algebraic}}]
	Let $a(x)\in \Linpolyring$ be a linearized polynomial and let the extension field $\Fqs$ of $\Fqm$ contain all roots of $a(x)$. Then, its roots form a linear space over $\Fq$ and each root has the same multiplicity, which is a power of $q$.
\end{theorem}
The roots of $a(x)$ form a linear space of dimension $d_r \leq d_a$.
Let $\{\beta_0,\beta_1,\dots,\beta_{d_r-1}\}$ be a basis of this $d_r$-dimensional root space. Then, each distinct root $r \in \Fqs$ of $a(x)$ can be expressed uniquely as $r = \sum_{i=0}^{d_r-1} R_i \beta_i$, where $R_i \in \Fq \forall i$.
Conversely, the following lemma shows that the unique minimal subspace polynomial %
is always a linearized polynomial.

\begin{lemma}[Minimal Subspace Polynomial {\citep[Theorem 3.52]{Lidl-Niederreiter:FF1996}}]\label{lem:min_subspace_polynomial}
	Let $\myspace{U}$ be a linear subspace of $\,\Fq^m$, considered over $\Fqm$. %
	Let $u_0,u_1,\dots,u_{\dim(\myspace{U})-1} \in \Fqm$ be a basis of this subspace. %
	Then, the minimal subspace polynomial
	\begin{equation*}
	\MinSubspacePoly{u_0,u_1,\dots,u_{\dim(\myspace{U})-1}} \coloneqq \prod_{\vec{u} \in \myspace{U}}\big(x-\extsmallfieldinputInverse{\vec{u}}\big),
	\end{equation*}
	is a linearized polynomial over $\Fqm$ of $q$-degree $\dim(\myspace{U})$.
\end{lemma}

The \emph{$q$-Vandermonde matrix} was introduced by \citet{Moore-TwoFoldGeneralizaton_1986} and plays an important role in linearized interpolation, evaluation and the $q$-transform.
For a vector $\mathbf a = \vecelements{a}\in\Fqm^n$,
we obtain the $s \times n$ $q$-Vandermonde matrix by the following map:
\begin{align}
\MoorematNoInput_s:\quad\Fqm^n &\rightarrow \Fqm^{s\times n}\nonumber\\
\a = \vecelements{a} & \mapsto \Mooremat{s}{q}{\vec{a}}
\coloneqq \MoormatExplicit{a}{s}{n}.\label{eq:def_Moorematrix}
\end{align}

\begin{lemma}[Determinant of Moore Matrix {\citep[Lemma 3.15]{Lidl-Niederreiter:FF1996}}]\label{lem:determinant_qvandermonde}
	Let $\a = \vecelements{a} \in \Fqm^n$. The determinant of the square $n \times n$ Moore matrix, defined as in \eqref{eq:def_Moorematrix}, is
	\begin{equation*}
	\det \big(\Mooremat{s}{q}{\vec{a}} \big)  = a_0 \prod_{j=0}^{n-2} \ \prod_{B_0, \dots, B_{j} \in \Fq} \bigg(a_{j+1}-\sum_{h=0}^{j}B_h a_h\bigg).
	\end{equation*}
\end{lemma}
Hence, $\det \left(\Mooremat{s}{q}{\vec{a}}\right) \neq 0 $ if and only if $a_0,a_1,\dots, a_{n-1}$ are linearly independent over $\Fq$.
If $a_0$, $a_1, \dots$, $a_{n-1}$ are linearly independent over $\Fq$, then $\Mooremat{s}{q}{\vec{a}}$ has rank $\min\{s,n\}$.

\section{Rank-Metric Codes}\label{sec:rank-metric}
The rank distance between $\vec{a}$ and $\vec{b}$ is the rank of the difference of the two matrix representations, i.e.,
\begin{equation*}
d_{\sfR}(\vec{a},\vec{b}) \coloneqq \rank_q(\vec{a}-\vec{b}) = \rank_q(\vec{A}-\vec{B}).
\end{equation*}
The minimum distance of a rank-metric code $\mycode{C}\codeRank{n,M}$ over $\Fqm$ of length $n$ and cardinality $M$ is defined as
  \begin{align*}
    d_{\sfR} (\code)\coloneqq \min_{\substack{{\vec{a},\vec{b}} \in \code\\ \vec{a} \neq \vec{b}}}
 d_{\fontmetric{R}}(\vec{a},\vec{b}).
  \end{align*}
A linear rank-metric code \mycode{C} over $\Fqm$, denoted by $\codelinearRank{n,k,d}$, is a linear subspace of $\Fqm^n$ of dimension $k$ and minimum rank distance
\begin{equation*}
  d_{\sfR} (\code) = \min_{\substack{{\vec{a},\vec{b}} \in \code\\ \vec{a} \neq \vec{b}}}
   d_{\fontmetric{R}}(\vec{a},\vec{b})
  = \min_{\substack{{\vec{a}} \in \code\\ \vec{a} \neq \0}} \rank_q(\vec{a}),
\end{equation*}
where the second equality holds because in a linear code any codeword can be represented as a linear combination of other codewords, thus, $d_{\fontmetric{R}}(\vec{a},\vec{b})=d_{\fontmetric{R}}(\vec{a}-\vec{b},\0)$.

A \emph{sphere} in the rank metric of radius $\tau$ around a word $\a\in \Fqm^n$ is the set of all words in rank distance exactly $\tau$ from $\a$ and a \emph{ball} is the set of all words in rank distance at most $\tau$ from $\a$.
Such a sphere will be denoted by $\Sphere{\tau}{\a} = \Sphere{\tau}{\A}$ and such a ball by $\Ball{\tau}{\a}=\Ball{\tau}{\A}$.
The cardinality of $\Ball{\tau}{\a}$ can obviously be obtained by summing up the cardinalities of the spheres around $\a$ of radius from zero up to $\tau$. The number of matrices of a certain rank is given, e.g., in \citep{Migler2004Weight}. Therefore, we have %
\begin{align*}
|\Sphere{\tau}{\a}| &=\quadbinom{m}{\tau}_q \prod\limits_{j=0}^{\tau-1} (q^n-q^j),\\
|\Ball{\tau}{\a}| &= \sum\limits_{i=0}^{\tau} |\Sphere{i}{\a}| = \sum\limits_{i=0}^{\tau} \quadbinom{m}{i}_q \prod\limits_{j=0}^{i-1} (q^n-q^j).
\end{align*}
Note that the cardinalities $|\Ball{\tau}{\a}|$ and $|\Sphere{\tau}{\a}|$ are independent of the choice of their center. The following lemma gives upper and lower bounds on the cardinality of balls with rank radius $\tau$.
\begin{lemma}[Bounds on Ball Size {\citep[Lemma 5]{Gadouleau2008Packing}}]
  For $0\leq\tau\leq \min\{n,m\}$, $$q^{\tau(m+n-\tau)}\leq |\Ball{\tau}{\a}|< K_q^{-1}\cdot q^{\tau(m+n-\tau)},$$ where $K_q\coloneqq\prod_{j=1}^{\infty}(1-q^{-j})$.
\end{lemma}
The asymptotic behavior of $|\Ball{\tau}{\a}|$ as $n\to\infty$, while $\lim_{n\to\infty}\frac{n}{m}$ is a constant, can be found in \citep[Lemma 11]{Gadouleau2008Packing}.

The cardinality of the intersection of two balls rank of given rank radius can be found in \citep[Proposition 4,5]{Gadouleau2008Packing}, while the size of the union of any $K$ balls can be found in \citep[Lemma 2]{gadouleau2009bounds}.

The following theorem states analogs of the {sphere-packing} (Hamming) and {Gilbert--Varshamov bound} in the rank metric,
which can be proven similar to the Hamming metric \citep{GadouleauYan-PropertiesOfCodesWiththeRankmetric_2006,Loidreau2008PropertiesRankMetric,Gadouleau2008Packing,Loidreau-AsymptoticBehaviorOfRankMetric-2012}.
\begin{theorem}[Sphere Packing and Gilbert--Varshamov Bound in the Rank Metric \citep{GadouleauYan-PropertiesOfCodesWiththeRankmetric_2006}]\label{theo:sphere-packing-gv-bounds} %
	Let $\maxCardinalityRank{n,d}$ denote the maximum cardinality of a block
	code over $\Fqm$
	of length $n$ and minimum rank distance $d$ and let
	$\Uniquecorrcap =\dhalffrac$.
	Then,
	\begin{equation}\label{eq:Hamming_bound_rank}
	\frac{q^{mn}}{|\Ball{d-1}{\0}|} \leq \maxCardinalityRank{n,d}\leq \frac{q^{mn}}{|\Ball{\Uniquecorrcap}{\0}|}.
	\end{equation}
\end{theorem}
The LHS of \eqref{eq:Hamming_bound_rank} is the Gilbert--Varshamov bound in the rank metric and the RHS of \eqref{eq:Hamming_bound_rank} is the sphere packing bound in the rank metric.

A code is called \emph{perfect} in the rank metric if it fulfills the RHS of~\eqref{eq:Hamming_bound_rank} with equality.
For a perfect code, the balls of radius $\Uniquecorrcap=\dhalffrac$ around all codewords cover the whole space.
However, in contrast to the Hamming metric, perfect codes do not exist in the rank metric \citep[Proposition~2]{Loidreau2008PropertiesRankMetric}.

The \emph{Singleton bound} in the rank metric is given in the following theorem.

\begin{theorem}[Singleton Bound in the Rank Metric {\citep[Theorem~5.4]{Delsarte_1978}}]\label{theo:Singleton-like-bound}
	Let \mycode{C} be a code over $\Fqm$ of length $n$, cardinality $M$, and minimum rank distance $d$.
	The cardinality $M$ of \mycode{C} is restricted by
	\begin{equation}\label{eq:Singleton_like_rank}
	M \leq q^{\min\{n(m-d+1),\; m(n-d+1)\}} =q^{\max\{n,m \}(\min\{n,m \}-d+1)}.
	\end{equation}
\end{theorem}
If the cardinality of a code fulfills \eqref{eq:Singleton_like_rank} with equality, the code is called \emph{maximum rank distance} (MRD) code.
An  MRD (not necessarily linear) code over $\Fqm$ of length $n$, cardinality $M = q^{\max\{n,m \}(\min\{n,m \}-d+1)}$, and minimum rank distance $d$ is denoted by \MRD{n,M}.

For linear codes of length $n \leq m$ and dimension $k$, Theorem~\ref{theo:Singleton-like-bound} implies that $d \leq n-k+1$, cf.~\citep[Corollary, p.~2]{Gabidulin_TheoryOfCodes_1985}.
A \emph{linear} MRD code over $\Fqm$ of length $n \leq m$, dimension $k$, and minimum rank distance $d=n-k+1$ is therefore denoted by \MRDlinear{n,k} and has cardinality $M = q^{mk}$. If $n > m$, we simply transpose all matrices and apply the previous considerations.

\section{Weight Distribution of MRD Codes}
\label{sec:weightDist}
In \citep[Theorem~5.6]{Delsarte_1978} and \citep[Section~3]{Gabidulin_TheoryOfCodes_1985}, the weight distribution of linear MRD codes was derived.
Let $A_s(n,d)$ denote the number of codewords of an $\MRDlinear{n,k}$ code of rank $s$.
This number can be calculated by separating the code into subspaces of dimension $s$ and determining the number of words of the code in this subspace.
By this, we obtain the recursive equation:
\begin{equation*}
A_s(n,d) = \quadbinom{n}{s} A_s(s,d), \quad d \leq s \leq n,
\end{equation*}
where $A_s(s,d)$ denotes the number of codewords in each $s$-dimensional subspace. 

The rank weight distribution of \MRDlinear{n,k} codes can be given by~\citep{Gabidulin_TheoryOfCodes_1985}
\begin{equation*}
A_{d+s}=\quadbinom{n}{d+s}\sum\limits_{j=0}^{s} (-1)^{j+s} \underbrace{\quadbinom{d+s}{d+j}q^{(s-j)(s-j-1)/2}(q^{m(j+1)}-1)}_{=: B_j},
\end{equation*}
for $s=0,1,\dots,n-d$ and $n \leq m$.

Note that for the special case of $s=0$ (i.e., the number of codewords of weight exactly $d$), we obtain:
\begin{equation*}
A_{d}(n,d) = \quadbinom{n}{d} (q^m-1).
\end{equation*}

In the following, we provide upper and lower bounds as well as an approximation of the MRD weight distribution, which might be helpful to quickly estimate the number of rank-$s$ codewords.
The relation between $B_j$ and $B_{j+1}$ can be estimated as
\begin{align*}
\frac{B_j}{B_{j+1}} &= \frac{\quadbinom{d+s}{d+j}q^{(s-j)(s-j-1)/2}(q^{m(j+1)}-1)}{\quadbinom{d+s}{d+j+1}q^{(s-j-1)(s-j-2)/2}(q^{m(j+2)}-1)}\\
&\approx \frac{\quadbinom{d+s}{d+j}q^{(s-j-1)}}{\quadbinom{d+s}{d+j+1}q^m}\\
&\approx q^{(d+j)(s-j)-(d+j+1)(s-j-1)} \cdot q^{s-j-1-m}\\
&= q^{(d+j)-(s-j)+1}\cdot q^{s-j-1-m} = q^{d+j-m}.
\end{align*}

For an upper bound on $A_{d+s}$, we only consider the highest term, i.e., $j=s$, and obtain
\begin{equation*}
A_{d+s} \leq \quadbinom{n}{d+s} \quadbinom{d+s}{d+s} q^0 (q^{m(s+1)}-1)
\leq \quadbinom{n}{d+s}q^{m(s+1)}.
\end{equation*}

To obtain a lower bound, we consider the difference between the two highest terms:
\begin{align}
A_{d+s} &\geq \quadbinom{n}{d+s} \left( (q^{m(s+1)}-1) - \quadbinom{d+s}{d+s-1}q^0(q^{ms}-1)\right)\nonumber\\
&= \quadbinom{n}{d+s} \left( (q^{m(s+1)}-1) -(q^{d+s}-1)(q^{ms}-1)\right)\nonumber\\
& \approx \quadbinom{n}{d+s} \left( q^{m(s+1)} -q^{d+s}q^{ms}\right) = \quadbinom{n}{d+s}  q^{m(s+1)}\left(1 -q^{d+s-m}\right).\nonumber
\end{align}

The following expression provides a rough estimate of the weight distribution:
\begin{equation*}
A_{d+s} \approx \quadbinom{n}{d+s}q^{m(s+1)} \approx q^m q^{nd - d^2 - 2ds-s^2}.
\end{equation*}

\section{Constant-Rank Codes}\label{sec:constantRank}

A concept closely related to the rank weight distribution of a code are \emph{constant-rank codes}. Commonly, codes are designed to guarantee a lower bound on the minimum rank-distance of any two codewords. For linear codes, this implies a lower bound on the rank of each codeword, i.e., every codeword in a linear rank-metric code of minimum rank-distance $d$ has rank at least $d$. In most cases, an $[n,k,d]_{q^m}^{\sfR}$ rank-metric code designed to have minimum distance $d$ will contain words of any rank $w$ with $d\leq w\leq n$.
 In contrast, constant-rank codes only contain words of a given rank. Their equivalent in the Hamming metric, so-called constant-weight codes, play an important role in list decoding \citep{Johnson1962New,Bassalygo1965New} and a variety of other applications (see, e.g., \citep{agrell2000upper} and the references within).

Constant-rank codes were first considered by \citet{Gadouleau2010ConstantRank}, where they are used to solve problems related to constant-dimension codes (a class of subspace codes), which have application in noncoherent network coding (see \cref{subsec:list-decoding-subspace-codes}). Similar to the Hamming metric, they also have implications for list-decoding of rank-metric codes, as discussed in \cref{sec:list-dec-gab}.%

\begin{definition}[Constant-Rank Code {\citep{Gadouleau2010ConstantRank}}]
  A $\codeRank{n,M,d}$ code $\code \subset \F_{q^m}^{n}$ is said to be of \emph{constant-rank} $w$ if
  \begin{equation*}
    \rk_q(\c) = w ,\ \forall \ \c \in \code .
  \end{equation*}
\end{definition}

Note that any non-trivial ($w>0$) constant-rank code is necessarily non-linear, as it does not contain the all-zero codeword.

In \citep{Gadouleau2010ConstantRank} it was shown that a constant-rank code of a certain cardinality can be constructed from a pair of constant-dimension subspace codes (for more details on subspace codes, see \cref{sec:subspace-codes}) of the same cardinality. This result was later generalized to arbitrary cardinalities by \citet[Proposition~2]{wa13a}.

\begin{lemma}[{\citep[Prop.~3]{Gadouleau2010ConstantRank},\citep[Prop.~2]{wa13a}}]
  Let $\code_1$ and $\code_2$ be $(n_1,M_1)$ and $(n_2,M_2)$ constant-dimension $r$ codes of subspace distance $d_{\sfS,1}$ and $d_{\sfS,2}$, respectively, where $r\leq\min\{n_1,n_2\}$. Then, there exists an $\codeRank{n,M,d_\sfR}$ constant-rank $r$ code of cardinality $M=\min\{M_1,M_2\}$. Furthermore, the minimum rank distance $d_\sfR$ is
  \begin{align*}
    d_\sfR \geq \frac{1}{2} d_{\sfS,1} + \frac{1}{2} d_{\sfS,2} \ ,
  \end{align*}
  and, if $M_1=M_2$, then
  \begin{align*}
    d_\sfR \leq \frac{1}{2} \min\{d_{\sfS,1},d_{\sfS,2}\} +r \ .
  \end{align*}
\end{lemma}
Similarly, optimal constant-dimension codes can be constructed from optimal constant-rank codes \citep[Theorem~2]{Gadouleau2010ConstantRank}.

Further, \citet{Gadouleau2010ConstantRank} presents rank-metric analogs of several bounds on the cardinality of codes in the Hamming metric. Specifically, bounds resembling the Johnson bound, the Singleton bound, and the Bassalygo-Elias are introduced.

In \citep{wa13a} a close inspection of the achievable size of constant-rank codes results in upper and lower bounds on their cardinality. Interestingly, these results show that, unlike for codes in the Hamming metric, there does not exist a generic list decoding radius beyond the unique decoding radius that is guaranteed to be achievable for any code of a given length and minimum rank distance, independent of its structure (for more details, see \cref{sec:list-dec-gab}).

A classification of constant weight codes was developed in \citep{randrianarisoa2019geometric} based on a geometric approach.

\section{Covering Property}
\label{sec:coveringProperty}
The \emph{covering radius} of a code $\code\subseteq\Fqm^n$ is the smallest integer $\rho$ such that all vectors in the space $\Fqm^n$ are within distance $\rho$ to some codeword of $\code$, i.e.,
\begin{align*}
  \rho(\code)\coloneqq\max_{\x\in\Fqm^n}\min_{\c\in\code} d(\x,\c).
\end{align*}
The covering radius of is a fundamental property of a code, which is generally harder to compute than the minimum distance. It measures the maximum weight of a correctable error vector. It also characterizes the \textit{maximality} property of a code. A code $\code$ is said to be \textit{maximal} if there does not exist any code $\code'$ of the same length and minimum distance such that $\code\subset \code'$. A maximal code has covering distance less than its minimum distance \citep{cohen1985covering,byrne2017covering}, i.e., $\rho(\code)\leq d(\code)-1$. 

The \emph{covering problem} is to find the minimum cardinality of a code $\code\subseteq\Fqm^n$ with covering radius $\rho$.
Denote by $M_\mathsf{R}(q^m,n,\rho)$ the minimum cardinality of such a code.
This quantity for codes in the Hamming metric has been studied extensively (see \citep{bartoli2014covering,cohen1997covering} and the references therein).  
In the rank-metric, the covering property has been studied by \citet{Gadouleau2008Packing}, \citep{gadouleau2009bounds} for codes which are $\Fqm$-linear.
Several lower and upper bounds on $M_\mathsf{R}(q^m,n,\rho)$ can be found in \citep[Section V]{Gadouleau2008Packing} and \citep{gadouleau2009bounds}. The numerical results and comparisons of the bounds for small parameters ($\rho\leq 6, m\leq 7$) can be found in the tables provided by both references. From \citep[Table I]{gadouleau2009bounds} one can observe that there is gap between the best upper and lower bounds for these parameters. Finding tighter bounds for non-asymptotic parameters is still an open problem. The existing bounds are computationally expensive for larger code parameters. The asymptotic behavior of $M_\mathsf{R}(q^m,n,\rho)$ is stated in \citep[Theorem 1]{Gadouleau2008Packing}.

More results on the covering property of $\Fq$-linear codes endowed with the rank metric can be found in~\citep{byrne2017covering}.

\section{Gabidulin Codes}\label{sec:gabidulin}
Gabidulin codes are a special class of rank-metric codes and can be defined by their generator matrices.
\begin{definition}[Gabidulin Code \citep{Gabidulin_TheoryOfCodes_1985}] \label{def:GabCode}
	A linear $\Gabcode{n}{k}$ Gabidulin code over $\Fqm$ of length $n \leq m$
	and dimension $k$, denoted by $\Gabcode{n}{k}$, is defined by its $k \times n$ generator matrix
	\begin{equation*}
	\mathbf G = \Mooremat{k}{q}{\vecelements{g}},
	\end{equation*}
	where $\vec{g}=\vecelements{g} \in \Fqm^n$ and $\rank_q(\vec{g}) = n$.
\end{definition}
It was shown by~\citet{Gabidulin_TheoryOfCodes_1985} that Gabidulin codes are MRD codes, i.e., they are of minimum rank-distance $d=n-k+1$.

Equivalently, we can define Gabidulin codes by evaluating $q$-degree restricted linearized polynomials:
\begin{align*}
&\Gabcode{n}{k} \coloneqq\Big\lbrace \veceval{f}{g} = f(\vec{g})  : f(x) \in \Linpolyring_{<k} \Big\rbrace,
\end{align*}
where the fixed elements $g_0,\dots, g_{n-1} \in \Fqm$ are linearly independent over $\Fq$ and $\Linpolyring_{<k}$ is the set of all linearized polynomials with $q$-degree less than $k$.

\begin{theorem}[Minimum Rank Distance of a Gabidulin Code]
	The minimum rank distance of a $\Gab{n,k}$ Gabidulin code over $\Fqm$ with $n \leq m$ is $d = n-k+1$.
\end{theorem}
\begin{proof}
	The evaluation polynomials $f(x)$ have $q$-degree less than $k$ and therefore
	the dimension of their root spaces over $\Fqm$ is at most $k-1$.\\
	Let $\Mat{C}=\extsmallfieldinput{\c} \in \Fq^{m \times n }$ denote the representation of $\c \in\Gab{n,k}$.
	Since the evaluation of a linearized polynomial at a basis is an $\Fq$-linear map,
	the dimension of the right kernel of $\Mat{C} \in \Fq^{m \times n }$ is equal to the
	dimension of the root space of the corresponding evaluation polynomial $f(x)$.
	Therefore,
	\begin{equation*}
	\dim \ker(\c) \leq k-1, \quad \forall \c \in \Gab{n,k}.
	\end{equation*}
	By linearity, there is a codeword $\c$ in $\Gab{n,k}$ of rank $d$ and due to the rank nullity theorem, for this codeword $\dim \ker(\c) = n-d$ holds.
	Hence,
	\begin{equation*}
	\dim \ker(\c)=n-d \leq k-1 \quad \Longleftrightarrow \quad d \geq n-k+1.
	\end{equation*}
	However, the Singleton-like bound \eqref{eq:Singleton_like_rank} implies that $d \leq n-k+1$ and hence, $d=n-k+1$.
\end{proof}
Gabdulin codes achieve the Singleton bound~\eqref{eq:Singleton_like_rank} with equality, thus they are MRD codes.

\begin{lemma}[Parity-Check Matrix of Gabidulin Code]\label{lem:paritycheck_gabidulin}
	Let $\G$ be a generator matrix of a \Gab{n,k} code, where $g_0,g_1,\dots,g_{n-1} \in \Fqm$ are linearly independent over $\Fq$.
	Let $h_0,h_1,\dots,h_{n-1}$ be a non-zero solution for the following $n-1$ linear equations:
	\begin{equation}\label{eq:gab_find_hmat}
	\sum\limits_{i=0}^{n-1} g_i^{[j]} h_i =0, \quad \forall j \in \intervallexcl{-n+k+1}{k}.
	\end{equation}
	Then, %
	the $(n-k)\times n$ matrix
	\begin{equation*}\label{eq:gab_paritycheckmatrix}
	\H\coloneqq
	\Mooremat{n-k}{q}{\vecelements{h}}=
	\begin{pmatrix}
	h_0^{[0]} & h_1^{[0]} & \dots & h_{n-1}^{[0]}\\
	h_0^{[1]} & h_1^{[1]} & \dots & h_{n-1}^{[1]}\\
	\vdots &\vdots&\ddots& \vdots\\
	h_0^{[n-k-1]} & h_1^{[n-k-1]} & \dots & h_{n-1}^{[n-k-1]}\\
	\end{pmatrix},
	\end{equation*}
	is a parity-check matrix of the \Gab{n,k} code.
\end{lemma}
\begin{proof}
	Since the dual of a \Gab{n,k} code is a \Gab{n,n-k} code \citep[Theorem~3]{Gabidulin_TheoryOfCodes_1985}, we have to prove that $\Mat{H}$ is a generator matrix of this dual code, i.e.,
	$\G \cdot \H^\top = \0$ has to hold, which is equivalent to the following $n-1$ linear equations:
	\begin{align*}
	\sum\limits_{i=0}^{n-1} g_i^{[l]} h_i^{[j]} &=0, \qquad  \forall l \in \intervallexcl{0}{k}, j \in \intervallexcl{0}{n-k},\nonumber\\
	\Longleftrightarrow \qquad  \sum\limits_{i=0}^{n-1} g_i^{[j]} h_i &=0, \qquad \forall j \in \intervallexcl{-n+k+1}{k}.
	\end{align*}
	Therefore, if $h_0,h_1,\dots,h_{n-1}$ are linearly independent over $\Fq$, $\H$ is a generator matrix of the dual code \Gab{n,n-k}.
	To prove this, denote $\widetilde{\vec{g}} = \vecelements{g^{[-n+k+1]}}$. Then, \eqref{eq:gab_find_hmat} is equivalent to
	\begin{equation}\label{eq:proof_paritycheckmatrix}
	\Mooremat{n-1}{q}{\widetilde{\vec{g}}} \cdot \vecelements{h}^\top = \0.
	\end{equation}
	The matrix $\Mooremat{n-1}{q}{\widetilde{\vec{g}}}$ is a generator matrix of a \Gab{n,n-1} code,
	since $g_0^{[-n+k+1]}$, $g_1^{[-n+k+1]}$, $\dots, g_{n-1}^{[-n+k+1]} \in \Fqm$ are linearly independent over $\Fq$.
	Due to \eqref{eq:proof_paritycheckmatrix}, the vector
	$\vecelements{h}$ is a codeword of a \Gab{n,1} code, i.e., of the \emph{dual code} of the \Gab{n,n-1} code.
	This \Gab{n,1} code has minimum rank distance $d=n-1+1 = n$ and therefore $\rk(\vecelements{h})=n$.
	Thus, $\H$ is a generator matrix of the dual \Gab{n,n-k} code and therefore a parity-check matrix of the \Gab{n,k} code.
\end{proof}

\section{Decoding of Gabidulin Codes}\label{sec:decoding_gabidulin_codes}

We now recall some well-known results on the decoding of Gabidulin codes.

\subsection{Decoding of Errors}
Following the descriptions of \citep{Gabidulin_TheoryOfCodes_1985,Roth_RankCodes_1991,Gabidulin1992Fast}, we explain the idea of syndrome-based bounded minimum distance (BMD) decoding, without going into detail about the different algorithmic possibilities.

Let $\r = \c +\e \in \Fqm^n$ be the received word, where $\c \in \Gab{n,k}$.
The goal of \emph{decoding} is to reconstruct $\c$, given only the received word $\r$.
Clearly, this is possible only if the rank of the error $\e$ is not too big.
Syndrome-based BMD decoding of Gabidulin codes follows similar steps as syndrome-based BMD decoding of Reed--Solomon codes.
For Reed--Solomon codes, the two main steps are determining the ``error locations'' and finding the ``error values'', where the second step is considered to be much easier.
Algebraic BMD decoding of Gabidulin codes also consists of two steps; however, the second one is not necessarily the easier one.
The starting point of decoding Gabidulin codes is to decompose the error, based on the well-known rank decomposition of a matrix.

\begin{lemma} [Rank Decomposition {\citep[Theorem~1]{MatsagliaStyan-EqualitiesAndInequalitiesForRanksOfMatrices}}]\label{lem:rank_decomp}  %
	For any matrix $\X \in \Fq^{m\times n}$ of rank $r$ there exist full rank matrices $\Y \in \Fq^{m\times r}$ and $\Mat{Z} \in \Fq^{r\times n}$  such that $\X=\Y\Mat{Z}$.
	Moreover, the column space of $\X$ is $\Colspace{\X}= \Colspace{\Y}\in \Grassm{m,r}$
and the row space is $\Rowspace{\X}=\Rowspace{\Mat{Z}} \in \Grassm{n,r}$.
\end{lemma}
Therefore, we can rewrite the matrix representation of $\e$ with $\rk_q(\e) = t$ by:
\begin{equation*}
\Mat{E} = \extsmallfieldinput{\e} = \A \cdot \Mat{B}, \quad \text{with} \ \A \in \Fq^{m\times t}, \; \Mat{B} \in \Fq^{t \times n},
\end{equation*}
and if we define $\a \coloneqq \extsmallfieldinputInverse{\A} \in \Fqm^{t}$:
\begin{equation}\label{eq:decoding_decompose_error}
\e = \extsmallfieldinputInverse{\Mat{E}} = \extsmallfieldinputInverse{\A}\cdot \Mat{B}  = \a \cdot \Mat{B} = \vecelementsArb{a}{t} \cdot \Mat{B}.
\end{equation}
This decomposition is clearly not unique, but any of them is good for decoding.
The two main steps of decoding Gabidulin codes are therefore: first, determine ``a basis of the column space'' of the error, i.e., find the vector $\a$ of a possible decomposition, and second, find the corresponding matrix $\Mat{B}$, which determines the row space\footnote{Note that it is possible to change the order of these two steps and search for a basis of the row space first and then find a corresponding matrix $\A$. This is a difference to Reed--Solomon codes, where we cannot interchange the two main steps.}.
Both steps are based on the \emph{syndrome}, which can be calculated out of the received word by
\begin{equation}\label{eq:decodegabi_defsyndrome}
\s = \vecelementsArb{s}{n-k}=\r \cdot \H^\top=\e \cdot \H^\top,
\end{equation}
where $\H$ is a parity-check matrix of the $\Gab{n,k}$ code (see Lemma~\ref{eq:gab_paritycheckmatrix}). We denote the associated syndrome polynomial by $s(x) = \sum_{i=0}^{n-k-1}s_i x^{[i]} \in \Linpolyring$. Its coefficients are calculated by
\begin{equation}\label{eq:proof_keyequation_syndromex}
s_i = \sum_{j=0}^{n-1} e_j h_j^{[i]} = \sum_{j=0}^{n-1}\sum_{l=0}^{t-1}a_l B_{l,j}h_j^{[i]} \eqqcolon \sum_{l=0}^{t-1} a_l d_l^{[i]}, \quad \forall i \in \intervallexcl{0}{n-k},
\end{equation}
with
\begin{equation}\label{eq:definition_di}
d_l \coloneqq  \sum_{j=0}^{n-1}B_{l,j}h_j.
\end{equation}

We define the \emph{error span polynomial} as the minimal subspace polynomial of the vector $\a$ to be %
\begin{equation}\label{eq:decoding_error_span_poly}
\Lambda(x) \coloneqq  \MinSubspacePoly{a_0,a_1,\dots,a_{t-1}}
= \prod_{B_0=0}^{q-1}\cdots \prod_{B_{t-1}=0}^{q-1} \Big(x-\sum_{i=0}^{t-1}B_ia_i\Big).
\end{equation}
Hence, due to Lemma~\ref{lem:min_subspace_polynomial}, the error span polynomial $\Lambda(x)$ is a linearized polynomial of $q$-degree $t$ and any $\Fq$-linear combination of roots of $\Lambda(x)$ is also a root of $\Lambda(x)$.

The first part of the decoding process is to determine $\Lambda(x)$, given the syndrome polynomial $s(x)$, and it is strongly based on the following theorem, the \emph{key equation} for decoding Gabidulin codes.
\begin{theorem}[Key Equation for Decoding Gabidulin Codes {\citep[Lemma~4]{Gabidulin_TheoryOfCodes_1985}}]\label{theo:gabidulin_key_equation}
	Let $\r = \c +\e \in \Fqm^n$ be given, where $\c \in \Gab{n,k}$ over $\Fqm$ and $\rk(\e) = t<n-k$. %
	Denote by $\s = \vecelementsArb{s}{n-k}=\r \cdot \H^\top \in \Fqm^{n-k}$ the syndrome as in \eqref{eq:decodegabi_defsyndrome} and by $s(x) = \sum_{i=0}^{n-k-1}s_i x^{[i]}$ its associated polynomial.

	Let the error span polynomial $\Lambda(x)$ with $\deg_q \Lambda(x)=t$ be defined as in \eqref{eq:decoding_error_span_poly}, where $\a = \vecelementsArb{a}{t}$
	is a basis of the column space of $\e$.
	Then,
	\begin{equation}\label{eq:key_equation_gabidulin}
	\Omega(x) \equiv \Lambda(s(x)) \mod x^{[n-k]},
	\end{equation}
	for some $\Omega(x) \in \Linpolyring$ with $\deg_q \Omega(x) <t$.
\end{theorem}
\begin{proof}
	With \eqref{eq:proof_keyequation_syndromex}, the $i$-th coefficient of $\Lambda(s(x))$ can be calculated by
	\begin{equation}\label{eq:proof_keyequation}
	\Omega_i \coloneqq  \sum_{j=0}^{i} \Lambda_j s_{i-j}^{[j]} = \sum_{j=0}^{i} \Lambda_j \Bigg(\sum_{l=0}^{t-1} a_l d_l^{[i-j]} \Bigg)^{[j]}
	= \sum_{l=0}^{t-1}d_l^{[i]}  \sum_{j=0}^{i} \Lambda_j \cdot a_l^{[j]}.
      \end{equation}
      Note that $\Omega_i$ is the coefficient of $x^i$ in $\Lambda(s(x))$.
	For $i \geq t$ this gives
	\begin{equation}\label{eq:proof_keyequation_2}
	\Omega_i =  \sum_{l=0}^{t-1}d_l^{[i]} \Lambda\big(a_l\big) = 0, \quad \forall i \geq t,
	\end{equation}
	since $\Lambda(x)$ has $a_i$, $\forall i \in \intervallexcl{0}{t}$, as roots, see \eqref{eq:decoding_error_span_poly}, and therefore $\deg_q \Omega(x) < \deg_q \Lambda(x) = t$.
\end{proof}
Alternatively, we can derive a key equation for the \emph{row} space of the error word.
\begin{theorem}[Row Space Key Equation for Decoding Gabidulin Codes {\citep{SilvaKschischang-FastEncodingDecodingGabidulin-2009}}]\label{theo:gabidulin_key_equation_rowspace}
	Let $\r = \c +\e \in \Fqm^n$ be given, where $\c \in \Gab{n,k}$ over $\Fqm$ and $\rk_q(\e) = t<n-k$. %
	Denote by $\s = \vecelementsArb{s}{n-k}=\r \cdot \H^\top \in \Fqm^{n-k}$ the syndrome as in \eqref{eq:decodegabi_defsyndrome} and by $s(x) = \sum_{i=0}^{n-k-1}s_i x^{[i]}$ its associated polynomial.

	Let the row error span polynomial be $\Gamma(x)= \MinSubspacePoly{d_0,d_1,\dots,d_{t-1}}$ with $\deg_q \Gamma(x)=t$,
	where $d_i$ is defined as in \eqref{eq:definition_di} for $i \in \intervallexcl{0}{t}$.
	Further, let
	\begin{equation}\label{eq:rowspace_ke_mod_syndrome}
	\widetilde{s}_i = s_{n-k-1-i}^{[i-n+k+1]}, \quad \forall i \in \intervallexcl{0}{n-k}
	\end{equation}
	and $\widetilde{s}(x) = \sum_{i=0}^{n-k-1}\widetilde{s}_i x^{[i]}$.
	Then,
	\begin{equation}\label{eq:key_equation_gabidulin_row}
	\Phi(x) \equiv \Gamma(\widetilde{s}(x)) \mod x^{[n-k]},
	\end{equation}
	for some $\Phi(x) \in \Linpolyring$ with $\deg_q \Phi(x) <t$.
\end{theorem}
\begin{proof}
	From \eqref{eq:definition_di}, we obtain
	\begin{equation}
	\widetilde{s}_i =\sum_{l=0}^{t-1}a_l^{[i-n+k+1]}d_l.
	\end{equation}
	The $i$-th coefficient of the linearized composition $\Gamma(\widetilde{s}(x))$ can then be calculated by
	\begin{align*}
	\Phi_i \coloneqq \big[\Gamma(\widetilde{s}(x))\big]_i &= \sum_{j=0}^{i} \Gamma_j \widetilde{s}_{i-j}^{[j]}\\
	&= \sum_{j=0}^{i} \Gamma_j \Bigg(\sum_{l=0}^{t-1} a_l^{[i-j-n+k+1]} d_l \Bigg)^{[j]}\\
	&= \sum_{l=0}^{t-1}a_l^{[i-n+k+1]}  \sum_{j=0}^{i} \Gamma_j \cdot d_l^{[j]}.
	\end{align*}
	For $i \geq t$ this gives
	\begin{equation*}
	\Phi_i =  \sum_{l=0}^{t-1}a_l^{[i-n+k+1]} \Gamma\big(d_l\big) = 0, \quad \forall i \geq t,
	\end{equation*}
	since $\Gamma(x)$ has all $d_i$, $\forall i \in \intervallexcl{0}{t}$, as roots and therefore $\deg_q \Phi(x) < \deg_q \Gamma(x) = t$.
\end{proof}

Based on the key equation from Theorem~\ref{theo:gabidulin_key_equation}, we explain the different steps of the standard decoding process of Gabidulin codes in the following and summarize them in Algorithm~\ref{algo:standard_decoding_gabidulin}. %
Similar steps have to be accomplished when we solve the row space key equation instead of the column space key equation. %

\subsubsection{Syndrome Calculation}
As mentioned before, the first step of decoding Gabidulin codes is calculating the syndrome based on a parity-check matrix $\H \in \Fqm^{(n-k)\times n}$ and the received word $\r \in \Fqm^n$ by
\begin{equation*}
\s = \vecelementsArb{s}{n-k}=\r \cdot \H^\top=\e \cdot \H^\top \in \Fqm^{n-k}.
\end{equation*}

\subsubsection{Solving the Key Equation}
The direct way to find $\Lambda(x)$ is to solve a linear system of equations based on the key equation \eqref{eq:key_equation_gabidulin}.
Due to \eqref{eq:proof_keyequation} and \eqref{eq:proof_keyequation_2} we have %
\begin{equation*}
\Omega_i= \sum_{j=0}^{i} \Lambda_j s_{i-j}^{[j]} = \sum_{j=0}^{t} \Lambda_j s_{i-j}^{[j]} = 0, \quad \forall i \geq t.
\end{equation*}
This is equivalent to the following homogeneous linear system of equations:
\begin{equation}\label{eq:key_equation_matrix_homogeneous}
\begin{pmatrix}
\Omega_{t}\\
\Omega_{t+1}\\
\vdots\\
\Omega_{n-k-1}\\
\end{pmatrix}
=
\begin{pmatrix}
s_t^{[0]}& s_{t-1}^{[1]} & \dots &s_0^{[t]} \\
s_{t+1}^{[0]}& s_{t}^{[1]} & \dots &s_1^{[t]} \\
\vdots &\vdots&\ddots& \vdots\\
s_{n-k-1}^{[0]}& s_{n-k-2}^{[1]} & \dots &s_{n-k-1-t}^{[t]} \\
\end{pmatrix} %
\cdot
\begin{pmatrix}
\Lambda_0 \\
\Lambda_1\\
\vdots\\
\Lambda_{t}
\end{pmatrix} %
=
\0.
\end{equation}
If the dimension of the solution space of \eqref{eq:key_equation_matrix_homogeneous} is one, then any solution of \eqref{eq:key_equation_matrix_homogeneous} provides the coefficients of the error span polynomial $\Lambda(x)$, defined as in \eqref{eq:decoding_error_span_poly}, except for a scalar factor. This scalar factor does not pose a problem, since it does not change the root space.
The following lemma provides a criterion to obtain the actual number of errors out of the syndrome matrix.

\begin{lemma}[Rank of Syndrome Matrix {\citep[Lemma, p. 132]{Gabidulin1992Fast}}]\label{lem:rank_syndromematrix}
	Let $\r = \c +\e \in \Fqm^n$, where $\c \in \Gab{n,k}$ and $\rk(\e) = t \leq \nkhalffrac$ and %
	let %
	$\vecelementsArb{s}{n-k}\in \Fqm^{n-k}$ denote the corresponding syndrome. %

	Then, for any $u \geq t$, the $u \times (u+1)$ matrix %
	\begin{equation}\label{eq:decoding_matrix_Su}
	\Mat{S}^{(u)} \coloneqq
	\begin{pmatrix}
	s_{u}^{[0]} & s_{u-1}^{[1]} &\dots &s_0^{[u]} \\
	s_{u+1}^{[0]} & s_{u}^{[1]} &\dots &s_1^{[u]} \\
	\vdots &\vdots&\ddots& \vdots\\
	s_{2u-1}^{[0]} & s_{2u-2}^{[1]} &\dots &s_{u-1}^{[u]} \\
	\end{pmatrix}
	\end{equation}
	has full rank $u$ if and only if $u=t$,
	where the $i$-th row of $\Mat{S}^{(u)} $ is defined to be all-zero if $i+u> n-k-1$, $\forall i= \intervallexcl{0}{u}$.
\end{lemma}
\begin{proof}
	Since there are $n-k$ non-zero syndrome coefficients,
	we can provide only $n-k-u$ non-zero rows of $\Mat{S}^{(u)}$. %
	Therefore, for $u > \nkhalffrac$, the matrix $\Mat{S}^{(u)}$ has
	only $n-k-u < u$ non-zero rows and therefore rank less than $u$.\\
	Let $a_i, d_i =0$ for $i \geq t$.
	For $u \leq\nkhalffrac$, we can decompose $\Mat{S}^{(u)}$ with \eqref{eq:proof_keyequation_syndromex} as
	\begin{equation*}
	\Mat{S}^{(u)} =
	\begin{pmatrix}
	d_0^{[u]} & d_1^{[u]} & \dots & d_{u-1}^{[u]}\\
	d_0^{[u+1]} & d_1^{[u+1]} & \dots & d_{u-1}^{[u+1]}\\
	\vdots &\vdots&\ddots& \vdots\\
	d_0^{[2u-1]} & d_1^{[2u-1]} & \dots & d_{u-1}^{[2u-1]}\\
	\end{pmatrix}
	\cdot
	\begin{pmatrix}
	a_0^{[0]} & a_0^{[1]} & \dots &  a_0^{[u]}\\
	a_1^{[0]} & a_1^{[1]} & \dots &  a_1^{[u]}\\
	\vdots &\vdots&\ddots& \vdots\\
	a_{u-1}^{[0]} & a_{u-1}^{[1]} & \dots &  a_{u-1}^{[u]}\\
	\end{pmatrix}.
	\end{equation*}
	Both matrices are Moore matrices and due to Lemma~\ref{lem:determinant_qvandermonde}, they have both full rank if and only if $d_0,d_1,\dots,d_{u-1}$ and $a_0, a_1,\dots, a_{u-1}$ are sets of elements which are linearly independent over $\Fq$. If $u>t$, this is not true, since $a_i, d_i =0$ for $i \geq t$. If $u = t$ this is true and the left matrix is a square matrix of rank $u$ and the right is a $u \times (u+1)$ matrix of rank $u$.
Since the first $u$ columns of the right matrix constitute a matrix of rank $u$, the statement follows.
\end{proof}

Thus, Lemma~\ref{lem:rank_syndromematrix} proves that for $t \leq \dhalffrac = \nkhalffrac$, $\Mat{S}^{(t)}$ has full rank and the dimension of the solution space of \eqref{eq:key_equation_matrix_homogeneous} is one.
For the algorithmic realization, we can set up $\Mat{S}^{(u)}$ for $u = \dhalffrac$ and check its rank. If the rank is not full, we decrease $u$ by one, control the rank, and so on, until we find $u$ such that the rank is full. Since we have to solve several linear systems of equations over $\Fqm$, the complexity of this step is in the order of at least $O(t^3) \leq O(n^3)$ operations in $\Fqm$ with Gaussian elimination (cf. \citep{Roth_RankCodes_1991,Gabidulin1992Fast}). See Subsection~\ref{sec:otherDecAlgo} for a list of asymptotically faster algorithms.

\subsubsection{Finding the Root Space of $\boldsymbol{\Lambda(x)}$}
After solving the key equation \eqref{eq:key_equation_matrix_homogeneous} for the coefficients of $\Lambda(x)$, we have to find a basis of the root space of $\Lambda(x)$.
This basis corresponds to one possible $\a = \vecelementsArb{a}{t}$ in the decomposition of \eqref{eq:decoding_decompose_error}.
Finding a basis of the root space of a linearized polynomial is relatively easy due to the structure of their roots. We can find the root space of $\Lambda(x)$ by finding the right kernel of its associated evaluated matrix, i.e.,
for some basis $\Basis = \setelements{\beta}{m}$ of $\Fqm$ over $\Fq$, we have to determine
\begin{equation*}
\ker \Big(\extsmallfield_\beta\big(\vecevalm{\Lambda}{\beta}\big)\Big). %
\end{equation*}
The kernel of this matrix %
is equivalent to $\extsmallfieldinput{\a}$ of one possible $\a$.
Thus, finding the root space of $\Lambda(x)$ involves solving a linear system of equations of size $m$ over $\Fq$, which has complexity at most $O(m^3)$ over $\Fq$.
This root-finding procedure was explained in detail in \citep{Lidl-Niederreiter:FF1996,Berlekamp1984Algebraic}.

\subsubsection{Determining the Error}
Knowing a possible vector $\a \in \Fqm^t$, we have to find the corresponding matrix $\Mat{B} \in \Fq^{t \times n}$ such that $\e = \a \cdot \Mat{B}$ as in \eqref{eq:decoding_decompose_error}.
This is basically done in two substeps. Based on \eqref{eq:proof_keyequation_syndromex}, we can set up the following system of equations which we have to solve for $\vec{d} = \vecelementsArb{d}{t}$\footnote{Notice that this system of equations from \eqref{eq:decoding_find_the_dis} can be used to do row-erasure-only correction, i.e., when $\a$ is known in advance due to the channel. For the concept of row and column erasures, see also Figure~\ref{fig:error_erasure}.}:
\begin{equation}\label{eq:decoding_find_the_dis}
\begin{pmatrix}
a_0^{[0]} & a_1^{[0]} & \dots &  a_{t-1}^{[0]}\\
a_0^{[-1]} & a_1^{[-1]} & \dots &  a_{t-1}^{[-1]}\\
\vdots &\vdots&\ddots& \vdots\\
a_0^{[-(n-k-1)]} & a_1^{[-(n-k-1)]} & \dots &  a_{t-1}^{[-(n-k-1)]}\\
\end{pmatrix}
\cdot
\begin{pmatrix}
d_0\\
d_1\\
\vdots\\
d_{t-1}
\end{pmatrix}
=
\begin{pmatrix}
s_0^{[0]}\\
s_1^{[-1]}\\
\vdots\\
s_{n-k-1}^{[-(n-k-1)]}
\end{pmatrix}.
\end{equation}
Solving this system of equations with Gaussian elimination requires complexity $O(n^3)$ operations over $\Fqm$, whereas the algorithm from \citep{Gabidulin_TheoryOfCodes_1985} requires complexity $O(n^2)$ over $\Fqm$ by using the Moore structure of the involved matrix.

After having found $\vec{d}$, we determine the matrix $\Mat{B}$ out of $d_l = \sum_{i=0}^{n-1} B_{l,i} h_i$ for all $l\in \intervallexcl{0}{t}$. The complexity of this calculation is negligible, since $\vecelements{h}$ has rank $n$ and we are looking for the representation of $\vec{d}$ over $\Fq$ using these linearly independent elements. %

Finally, we calculate $\e  = \a \cdot \Mat{B}$ and can reconstruct $\c = \r - \e$. A summary of this decoding procedure is given in Algorithm~\ref{algo:standard_decoding_gabidulin}.
Notice that the algorithm will most likely output a \texttt{decoding failure} when $t > \dhalffrac$.
In this case, with high probability, there is more than one solution of the linear system of equations that solves for the coefficients of the error span polynomial.%

\printalgoIEEE{
\DontPrintSemicolon
\caption{\newline$\c$ \textbf{or} \texttt{decoding failure} $\leftarrow$\textsc{DecodeGabidulin}$\big(\r; \H \big)$}
\label{algo:standard_decoding_gabidulin}
\KwIn{$\r = \vecelements{r}\in \Fqm^n$ with $n\leq m$, \newline
Parity-check matrix $\Mat{H} = \Mooremat{n-k}{q}{\vecelements{h}}$ of $\Gab{n,k}$}
\KwOut{Estimated codeword $\Mat{c} \in \Fqm^n$ or \normalfont{\texttt{decoding failure}}}
Syndrome calculation: $\Mat{s} \leftarrow \Mat{r} \cdot \Mat{H}^\top \in \Fqm^{n-k}$\;
\If{$\Mat{s} = \0$}{
Estimated codeword: $\Mat{c} \leftarrow \Mat{r}$\;
\Return{$\c$}}

\Else{
Set up $\Mat{S}^{(t)}$ as in \eqref{eq:decoding_matrix_Su} for $t = \lfloor\frac{n-k}{2}\rfloor$\newline %
\While{$\rk(\Mat{S}^{(t)}) < t$}
{$ t \leftarrow t -1$\;
Set up $\Mat{S}^{(t)}$ as in \eqref{eq:decoding_matrix_Su}}
Solve $\Mat{S}^{(t)}\cdot \boldsymbol{\Lambda}^\top = \0$ for $\boldsymbol{\Lambda} = (\Lambda_0 \ \Lambda_1 \ \dots \ \Lambda_{t}) \in \Fqm^{t+1}$\;
Find basis $\vecelementsArb{a}{\varepsilon}\in \Fqm^\varepsilon$ of the root space of $\Lambda(x) = \sum_{i=0}^{t}\Lambda_ix^{[i]}$ over $\Fqm$\;
\If{$\varepsilon = t$}{
Find $\vec{d}= \vecelementsArb{d}{t}\in \Fqm^t$ by solving \eqref{eq:decoding_find_the_dis}\;
Find $\Mat{B} = \big(B_{i,j}\big)^{i \in \intervallexcl{0}{t}}_{j \in \intervallexcl{0}{n}} \in \Fq^{t \times n}$ such that $d_i = \sum_{j=0}^{n-1} B_{i,j} h_j$\;
Estimated codeword: $\c \leftarrow \Mat{r}- \Mat{a} \cdot \Mat{B}$\;
\Return{$\c$}}
\Else{
\Return{\normalfont{\texttt{decoding failure}}}}}

}

\subsection{Error-Erasure Decoding}\label{subsec:error-erasure}
For a short description on error-erasure decoding of Gabidulin codes, denote by $\vec{C}\in \Fq^{m \times n}$  the transmitted codeword (i.e., the matrix representation of $\vec{c}\in \Fqm^n$) of a $\Gabcode{n}{k}$ code that is corrupted by an additive error $\mathbf E \in \Fq^{m \times n}$. At the receiver side, only the received matrix $\mathbf R \in \Fq^{m \times n}$, where $\mathbf R = \vec{C}+ \mathbf E$, is known.
The channel might provide additional side information in the form of erasures:
\begin{itemize}
	\item $\numbRowErasures$ {row erasures} (in \citep{silva_rank_metric_approach} called ``deviations'') and
	\item $\numbColErasures$ {column erasures} (in \citep{silva_rank_metric_approach} called ``erasures''),
\end{itemize}
such that the received matrix can be decomposed into
\begin{equation}\label{eq:decomp_errrorerasures}
\vec{R}
=\vec{C}+ \underbrace{\vec{A}^{(R)} \vec{B}^{(R)} + \vec{A}^{(C)} \vec{B}^{(C)} + \vec{E}^{(E)}}_{= \,\vec{E}_\mathrm{total}},
\end{equation}
where $\vec{A}^{(R)} \in \Fq^{m \times \numbRowErasures}$, $\vec{B}^{(R)} \in \Fq^{\numbRowErasures \times n}$,
$\vec{A}^{(C)} \in \Fq^{m \times \numbColErasures}$, $\vec{B}^{(C)} \in \Fq^{\numbColErasures\times n}$ are full-rank matrices, respectively, and $\vec{E}^{(E)}$ $\in \Fq^{m \times n}$ is a matrix of rank $t$. For an illustration see Figure~\ref{fig:error_erasure}.
The decoder knows $\vec{R}$ and additionally $\vec{A}^{(R)}$ and $\vec{B}^{(C)}$.
Further, $t$ denotes the number of errors without side information.
The rank-metric error-erasure decoding algorithms from \citep{GabidulinPilipchuck_ErrorErasureRankCodes_2008,silva_rank_metric_approach,wachter2014list}
can then reconstruct $\vec{c}_\mathcal{G} \in \Gabcode{n}{k}$ with asymptotic complexity $O(n^2)$ operations over $\Fqm$, or in sub-quadratic complexity using the fast operations described in \citep{PuchingerWachterzeh-ISIT2016,puchinger2018fast}, if
\begin{equation}\label{eq:errorerasurecond}
2t + \numbRowErasures + \numbColErasures \leq d-1 = n-k
\end{equation}
is fulfilled.

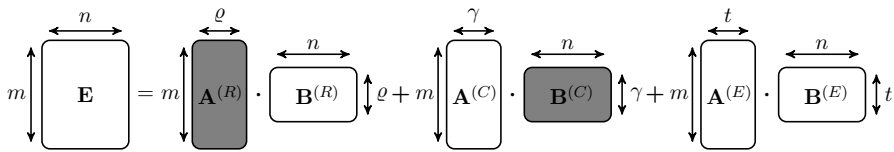
\begin{figure}[htb]
	\centering
	\scalebox{0.75}{
	\begin{tikzpicture}[scale = 1]
  
\node[mybox,minimum height=5em,minimum width=4em,label=left:$m\ $] (errortotal) {$\mathbf{E}$} ;
\node (leftarrownorth) [left=5ex of errortotal.north]{};
\node (leftarrowsouth) [left=5ex of errortotal.south]{};
\node (leftarrowwest) [above=6ex of errortotal.west]{};
\node (leftarroweast) [above=6ex of errortotal.east]{};
\draw[black,<->,thick] (leftarrownorth) -- (leftarrowsouth);
\draw[black,<->,thick] (leftarrowwest) -- (leftarroweast);
\node[rectangle,above=5pt of errortotal](lengthcodeword){$ n$}; 
\node[rectangle,right=0pt of errortotal](equal){$ = $}; 

\node[mybox_small,right=15pt of equal,minimum height=5em, minimum width = 2.5em,label=left:$m\, $,fill=gray,opacity=1] (matrixAR) {$\!\!\!\mathbf{A}^{(R)}$} ;
\node (leftarrownorth2) [left=3ex of matrixAR.north]{};
\node (leftarrowsouth2) [left=3ex of matrixAR.south]{};
\node (leftarrowwest2) [above=6ex of matrixAR.west]{};
\node (leftarroweast2) [above=6ex of matrixAR.east]{};
\draw[black,<->,thick] (leftarrownorth2) -- (leftarrowsouth2);
\draw[black,<->,thick] (leftarrowwest2) -- (leftarroweast2);
\node[rectangle,above=5pt of matrixAR](lengthmatrixAR){$ \numbRowErasures$}; 
\node[rectangle,right=0pt of matrixAR](cdot1){$\boldsymbol{\cdot} $};

 \node[mybox_small,right=11pt of matrixAR,minimum width=4em, minimum height = 2.5em,label=right:$\ \, \numbRowErasures$] (matrixBR) {$\!\!\mathbf{B}^{(R)}$} ;
 \node (leftarrownorth3) [right=5ex of matrixBR.north]{};
 \node (leftarrowsouth3) [right=5ex of matrixBR.south]{};
 \node (leftarrowwest3) [above=3.5ex of matrixBR.west]{};
 \node (leftarroweast3) [above=3.5ex of matrixBR.east]{};
 \draw[black,<->,thick] (leftarrownorth3) -- (leftarrowsouth3);
 \draw[black,<->,thick] (leftarrowwest3) -- (leftarroweast3);
 \node[rectangle,above=5pt of matrixBR](lengthmatrixBR){$n$}; 
 \node[rectangle,right=13pt of matrixBR](plus1){$\boldsymbol{ +} $};

\node[mybox_small,right=14pt of plus1,minimum height=5em, minimum width = 2.5em,label=left:$m\, $] (matrixAR) {$\!\!\!\mathbf{A}^{(C)}$} ;
\node (leftarrownorth2) [left=3ex of matrixAR.north]{};
\node (leftarrowsouth2) [left=3ex of matrixAR.south]{};
\node (leftarrowwest2) [above=6ex of matrixAR.west]{};
\node (leftarroweast2) [above=6ex of matrixAR.east]{};
\draw[black,<->,thick] (leftarrownorth2) -- (leftarrowsouth2);
\draw[black,<->,thick] (leftarrowwest2) -- (leftarroweast2);
\node[rectangle,above=5pt of matrixAR](lengthmatrixAR){$\numbColErasures$}; 
\node[rectangle,right=0pt of matrixAR](cdot1){$\boldsymbol{\cdot} $};

 \node[mybox_small,right=11pt of matrixAR,minimum width=4em, minimum height = 2.5em,label=right:$\ \, \numbColErasures$,fill=gray,opacity=1] (matrixBR) {$\!\!\mathbf{B}^{(C)}$} ;
 \node (leftarrownorth3) [right=5ex of matrixBR.north]{};
 \node (leftarrowsouth3) [right=5ex of matrixBR.south]{};
 \node (leftarrowwest3) [above=3.5ex of matrixBR.west]{};
 \node (leftarroweast3) [above=3.5ex of matrixBR.east]{};
 \draw[black,<->,thick] (leftarrownorth3) -- (leftarrowsouth3);
 \draw[black,<->,thick] (leftarrowwest3) -- (leftarroweast3);
 \node[rectangle,above=5pt of matrixBR](lengthmatrixBR){$n$}; 
 \node[rectangle,right=13pt of matrixBR](plus2){$\boldsymbol{ +} $};

\node[mybox_small,right=14pt of plus2,minimum height=5em, minimum width = 2.5em,label=left:$m\, $] (matrixAR) {$\!\!\!\mathbf{A}^{(E)}$} ;
\node (leftarrownorth2) [left=3ex of matrixAR.north]{};
\node (leftarrowsouth2) [left=3ex of matrixAR.south]{};
\node (leftarrowwest2) [above=6ex of matrixAR.west]{};
\node (leftarroweast2) [above=6ex of matrixAR.east]{};
\draw[black,<->,thick] (leftarrownorth2) -- (leftarrowsouth2);
\draw[black,<->,thick] (leftarrowwest2) -- (leftarroweast2);
\node[rectangle,above=5pt of matrixAR](lengthmatrixAR){$t$}; 
\node[rectangle,right=0pt of matrixAR](cdot1){$\boldsymbol{\cdot} $};

 \node[mybox_small,right=11pt of matrixAR,minimum width=4em, minimum height = 2.5em,label=right:$\ \, t$] (matrixBR) {$\!\!\mathbf{B}^{(E)}$} ;
 \node (leftarrownorth3) [right=5ex of matrixBR.north]{};
 \node (leftarrowsouth3) [right=5ex of matrixBR.south]{};
 \node (leftarrowwest3) [above=3.5ex of matrixBR.west]{};
 \node (leftarroweast3) [above=3.5ex of matrixBR.east]{};
 \draw[black,<->,thick] (leftarrownorth3) -- (leftarrowsouth3);
 \draw[black,<->,thick] (leftarrowwest3) -- (leftarroweast3);
 \node[rectangle,above=5pt of matrixBR](lengthmatrixBR){$n$}; 
\end{tikzpicture}}
	\caption{Illustration of row erasures, column erasures and (full) errors in the rank metric. The known matrices (given by the channel) are filled with gray.}
	\label{fig:error_erasure}
\end{figure}

\subsection{Other Decoding Algorithms}
\label{sec:otherDecAlgo}
The syndrome-based decoder that we presented above is based on the algorithms presented in \citep{Gabidulin_TheoryOfCodes_1985,Gabidulin1992Fast}. However, instead of solving a linear system of equations to find the error span polynomial (we call this the \emph{first step} below), Gabidulin suggested to use the analog of the extended Euclidean algorithm (EEA) for linearized polynomials. This gives a quadratic complexity in the code length for this step of the decoder.
Roth independently proposed an alternative decoder that also solves the first step using a linear system of equations in \citep{Roth_RankCodes_1991}.
Furthermore, there are several adaptations of the Berlekamp--Massey algorithm \citep{Paramonov_Tretjakov_BMA_1991,RichterPlass_DecodingRankCodes_2004,Richter_RankCodes_2004,HassanSidorenko_FastLinShift_2010,Sidorenko2011Linearized,sidorenko2014fast}, which allow to recover the error span polynomial in $O(n^2)$ or $O^\sim(n^{1.69})$ (divide-and-conquer variant) %
operations over $\Fqm$, where~$n$ is the code length. Note that in the latter case, the second step of the decoder becomes the bottleneck.

All of the mentioned algorithms are rank-metric counterparts of classical decoding algorithms for Reed--Solomon codes, which first find an error locator polynomial and then determine the error values.
For Reed--Solomon codes, the latter step has negligible complexity compared to the first part of the algorithm.
In the Gabidulin code case, the second step is rather heavy. Determining the error from the known $\a \in \Fqm^t$ is as fast as an algorithm with (soft-)quadratic cost over $\Fqm$: computing a basis of the root space of the error span polynomial costs $O(n^2 m)$ operations in $\Fq$, which asymptotically costs---up to logarithmic factors---as much as $O(n^2)$ operations over $\Fqm$ using the bases in \citep{couveignes2009elliptic}.

There are also some decoding algorithms that directly output the message polynomial and thus avoid the second step.
\citet{loidreau2006welch} proposed a Welch--Berlekamp-type decoder, which directly returns the message polynomial by first finding two linearized polynomials that fulfill certain evaluation conditions and degree constraints, followed by a division of one polynomial by the other. The resulting decoder has cubic complexity in the code length over $\Fqm$.

The decoder in \citep{WachterAfanSido-FastDecGabidulin_DCC_journ} is based on a key equation that contains the message polynomial, which is sometimes referred to as ``Gao-like key equation'' since it can be seen as the rank-metric analog of \citep{Gao_ANewDecodingAlgorithm_2002}.
The key equation can be solved using only multiplication, division, EEA, multi-point evaluation, minimal subspace polynomial computation, and interpolation of linearized polynomials.
Fast operations for these operations with linearized polynomials in \citep{PuchingerWachterzeh-ISIT2016,puchinger2018fast} led to the first sub-quadratic decoding algorithm for Gabidulin codes.

Module minimization \citep{puchinger2017row}, also called row reduction, of linearized polynomial matrices gives a general framework for decoding Gabidulin codes and related codes. For the case of Gabidulin codes, it is equivalent to the EEA, but for other classes, such as interleaved codes, it is more flexible. A divide-and-conquer %
version, Alekhnovich's algorithm, gives a sub-quadratic algorithm in $n$.
Most recently, minimal approximant bases for linearized polynomials have been studied in \citep{bartz2021fast}.
These bases enable more flexible and faster decoding algorithms for variants of Gabidulin codes, e.g., (lifted) interleaved and folded Gabidulin codes.

\section{Considerations on List Decoding Gabidulin Codes}\label{sec:list-dec-gab}

Given a word $\vec{r}\in\F_{q^m}^n$ (or alternatively, a matrix ${\Mat{R}\in\F_q^{m\times n}}$), a \textit{list decoding} algorithm outputs all codewords that are inside a ball of radius $\tau$, centered at $\vec{r}$, where $\tau$ is possibly larger than the unique decoding radius of the code. For a given code, a natural question to ask is: for which values of $\tau$ can list decoding be done efficiently?

Although Gabidulin codes can be seen as the rank-metric analog of Reed--Solomon codes, there are a few remarkable differences.

List decoding of rank-metric codes and Gabidulin codes was recently studied in \citep{wa13a,DingListDecRandomRankMetric2015,GuruswamiExplicit2016,RavivWachterzeh_GabidulinBounds_journal,ShuXingYuanListDecGab2017,XingYuan-NewRankListDec2018,TrombettZullo-ListDecRM2020}. In \citep{wa13a}, it was shown that Gabidulin codes cannot be list-decoded beyond the Johnson radius.
Hence, for $\tau \geq n-\sqrt{n(n-d)}$, general list decoding has exponential (in $n$) worst-case complexity since the list size grows exponential in $n$, see \citep[Theorem~1]{wa13a}.
This result was generalized to any rank-metric code by~\citet{DingListDecRandomRankMetric2015}. {When $m$ is sufficiently large,~\citet{DingListDecRandomRankMetric2015} also showed that with high probability a random rank-metric code can be efficiently list decoded}.
Further, it was shown in~\citep{wa13a} that there is no Johnson-like polynomial upper bound on the list size since there exists a non-linear rank-metric code with exponentially growing list size for any radius greater than the unique decoding radius.
In~\citep{GuruswamiExplicit2016}, an explicit subcode of a Gabidulin code was shown to be efficiently list-decodable.
The remarkable difference to Reed--Solomon codes is the following: There are families of Gabidulin codes of rate $R \geq 1/6$ for which some received words have exponential-sized lists even for decoding only one error beyond $\nkhalffrac$, see \citep{RavivWachterzeh_GabidulinBounds_journal}.
This result was recently generalized in \citep{TrombettZullo-ListDecRM2020} to more general classes of MRD codes.

Despite their analogy to Reed--Solomon codes and many unsuccessful trials by researchers in the last $20$ years, nobody has found a polynomial-time list decoding algorithm for Gabidulin codes analog to the Guruswami--Sudan decoder for RS codes. In fact, the above mentioned result in the list size of some Gabidulin codes proves that an equally general algorithm cannot exist.

Figures~\ref{fig:rank_decoding_region} and \ref{fig:gabidulin_decoding_region} illustrates the different decoding regions of rank-metric and Gabidulin codes.
\begin{figure}[htb!]
\begin{subfigure}[b]{0.49\textwidth}
	\centering
		{\begin{tikzpicture}[scale = 0.7]
\begin{axis}
  [
legend style={font=\footnotesize, at={(0.46,0.97)}},
xlabel = {$\delta = {d}/{n} $},
xlabel style = {font=\normalsize},
ylabel = {$\dfrac{\tau}{n}$},
ylabel style = {font=\normalsize,rotate=270}, %
tick label style= {font=\footnotesize},
xmin = 0,
xmax = 1,
ymin = 0,
ymax = 1,
line width = 0.8,
]
\addplot[black,fill = black, fill opacity = 0.1]{1}\closedcycle;
\addplot[black, fill=white]{0.5*x} \closedcycle;
\addplot[black, fill=white, fill opacity = 0.2]{0.5*x} \closedcycle;
\draw[black,-] (axis cs:1,0) -- (axis cs:1,1);
\node at (axis cs:0.2,0.65) [anchor=north west] {There is a code with};
\node at (axis cs:0.2,0.57) [anchor=north west]  {\textbf{exponential} list size};
\node at (axis cs:0.6,0.2) [anchor=north west] {unique};
\legend{BMD Radius}
\end{axis}
\end{tikzpicture}}
\caption{General codes in the rank metric}\label{fig:rank_decoding_region}
\end{subfigure}%
\hfill%
\begin{subfigure}[b]{0.49\textwidth}
		\centering
			{\begin{tikzpicture}[scale = 0.7]
\begin{axis}
[font=\normalsize,
legend style={font=\footnotesize, at={(0.5,0.97)}},
xlabel = {$\delta = {d}/{n} \approx 1-R$},
xlabel style = {font=\normalsize},
ylabel = {$\dfrac{\tau}{n}$},
ylabel style = {font=\normalsize,rotate=270}, %
tick label style= {font=\footnotesize},
xmin = 0,
xmax = 1,
ymin = 0,
ymax = 1,
line width = 0.8,
]
\addplot[black,fill = black, fill opacity = 0.1]{1}\closedcycle;
\addplot[black, dashed, samples=8000,fill=white]{1-sqrt(1-x)} \closedcycle;
\addplot[black, fill=white, fill opacity = 0.2]{0.5*x} \closedcycle;
\draw[black,-] (axis cs:1,0) -- (axis cs:1,1);
\node at (axis cs:0.55,0.63) [anchor=north east] { exponential for all};
\node at (axis cs:0.92,0.55) [anchor=north east,font=\scriptsize] {exponential};
\node[anchor=north east,font=\scriptsize](aaa) at (axis cs:0.9,0.5){ for many};
\node at (axis cs:0.6,0.2) [anchor=north west] {unique};
\legend{BMD Radius, Johnson Radius}
\end{axis}
\end{tikzpicture}}
		\caption{Gabidulin codes}\label{fig:gabidulin_decoding_region}
\end{subfigure}
	\caption{List size of codes in the rank metric, depending on normalized Bounded Minimum Distance (BMD) decoding radius $\tau_{BMD}/n = (d-1)/2n$ and normalized Johnson radius $\tau_J/n $
		and on the normalized minimum distance $\delta = {d}/{n}$. \label{fig:decoding_regions}}
\end{figure}
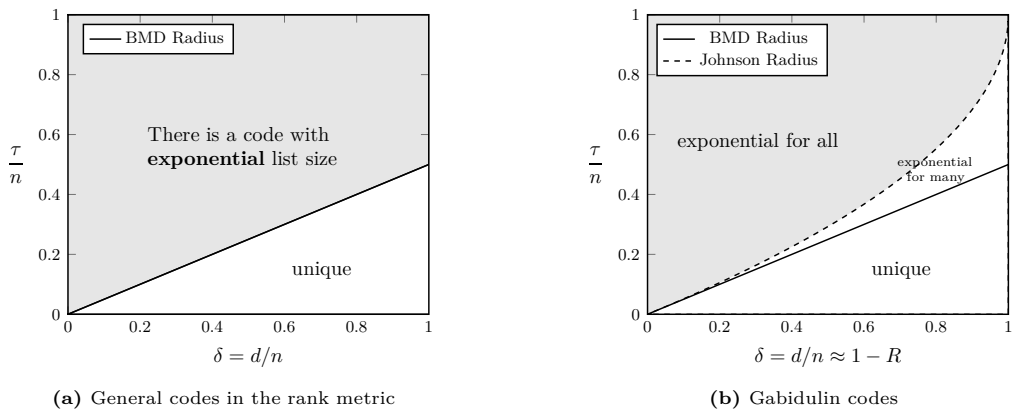

\section{Interleaved Gabidulin Codes}\label{sec:interleaved}

Interleaved Gabidulin Codes are a code class for which efficient decoders are known that are able to correct errors of rank larger than $\lfloor\frac{d-1}{2}\rfloor$.
\begin{definition} [Interleaved Gabidulin Codes~\citep{Loidreau_Overbeck_Interleaved_2006}]
  \label{def:IntGabCode}
A linear (vertically, heterogeneous) interleaved Gabidulin code $\IntGabcode{u;n,k^{(1)},\dots,k^{(u)}}$ over $\Fqm$ of length $n \leq m$, dimensions $k^{(i)}\leq n$, $\forall i\in[1,u]$, and interleaving order $u$ is defined by
\begin{equation*}
\IntGabcode{u;n,k^{(1)},\dots,k^{(u)}} \coloneqq
\left\lbrace
\begin{pmatrix}
\vec{c}_{\mycode{G}}^{(1)}\\
\vec{c}_{\mycode{G}}^{(2)}\\
\vdots\\
\vec{c}_{\mycode{G}}^{(u)}
\end{pmatrix}
: \vec{c}_{\mycode{G}}^{(i)} \in \Gabcode{n}{k^{(i)}} , \forall i \in [1,u]
\right\rbrace.
\end{equation*}
If $k{(i)}=k, \forall i\in[1,u]$, we call the code a \emph{homogeneous} interleaved Gabidulin code and denote it by $\IntGabcode{u;n,k}$.
\end{definition}
When considering random errors of rank weight $t$, the code $\IntGabcode{u;n,k}$ can be decoded uniquely with high probability up to
$t \leq \lfloor \frac{u}{u+1}(n-k)\rfloor$ errors\footnote{In this setting, an error of weight $t$ is a $u \times n$ matrix of $\Fq$-rank $t$. Note that this means that the tall $(u m) \times n$-matrix obtained by expanding the matrix component-wise over $\Fq$ has rank $t$.}, cf.~\citep{Loidreau_Overbeck_Interleaved_2006,Sidorenko2011SkewFeedback,wachter2014list}.
However, it is well-known that there are error patterns for which the known efficient decoders fail.

There are several decoding algorithms for interleaved Gabidulin codes.
The most common one is the syndrome-based decoder, cf. \citep{Sidorenko2011SkewFeedback}. For interleaved codes, we consider errors whose column spaces are the same. Therefore, we can stack the syndrome matrices~\eqref{eq:decoding_find_the_dis} from each row of the received word and solve the corresponding linear system of equations for the error span polynomial to obtain the error span polynomial which is the same for all rows. Counting the number of unknowns and equations results in the restriction $t \leq \lfloor \frac{u}{u+1}(n-k)\rfloor$. This maximum radius can only be achieved if the rank of the stacked syndrome matrix is full, else the linear system of equations does not have a unique solution and a syndrome-based decoder declares a decoding failure.

In the following, we shortly summarize the interpolation-based decoder by \citet{wachter2014list}.
The algorithm consists of two steps: the \emph{interpolation step} computes non-zero vectors of linearized polynomials
\begin{equation*}
\vec{Q}^{(i)}=[Q_0^{(i)},Q_1^{(i)},\dots,Q_u^{(i)}] \in \Linpolyring^{u+1} \setminus \{\0\}, \ \forall i=1,\dots,\intOrder
\end{equation*}
such that they fulfill certain degree and evaluation conditions with respect to the received matrix $\vec{C}+\vec{E}$.
The \emph{root-finding step} finds all message polynomial vectors $[f_1,\dots,f_u]$ of degrees $\deg f_j < k_j$ such that
\begin{equation*}
Q_0^{(i)} + \sum_{j=1}^{u} Q_j^{(i)} f_j = 0 \quad \forall \, i=1,\dots,\intOrder.
\end{equation*}
If the rank of the error matrix $\vec{E}$ is at most $\tfrac{u}{u+1}(n-\bar{k})$ with $\bar{k} \coloneqq \tfrac{1}{u}\sum_{j=1}^u k_j$, then at least one satisfactory interpolation vector $\vec{Q}^{(i)}$ exists~\citep{wachter2014list}. In this case, the output list contains the transmitted message polynomial vector.
The algorithm can be considered as a partial unique or list decoder~\citep{wachter2014list}.
Notice that the list decoder returns a \emph{basis} of the list of all codewords on the decoding list. The number of all words on the list can become large.

\section{Folded Gabidulin Codes}\label{subsec:FoldedGabidulin}

Variants of \emph{folded} Gabidulin codes were proposed independently in~\citep{Mahdavifar2012Listdecoding} and \citep{GuruswamiInsertionsDeletions}.
In~\citep{GuruswamiInsertionsDeletions} the coefficients of the message polynomial are restricted to belong to the subfield $\Fq$ of $\Fqm$.
In the following we consider folded Gabidulin codes as defined in~\citep{Mahdavifar2012Listdecoding}.

\begin{definition}[$h$-Folded Gabidulin Code]\label{def:hFoldedGab}
 Let $\pe$ be a primitive element of $\Fqm$ and let $n\leq m$.
 Let $\foldPar$ be a positive integer that divides $n$ and let $\lenFG=n/\foldPar$.
 An $\foldPar$-folded Gabidulin code \FGab{\foldPar,\pe;\lenFG,k} of length $\lenFG$ and dimension $k$ is defined as
 \begin{equation}
  \left\{%
  \begin{pmatrix}
   f(\pe^0) & f(\pe^1) & \dots & f(\pe^{\foldPar-1})
   \\[5pt]
   f(\pe^{\foldPar}) & f(\pe^{\foldPar+1}) & \dots & f(\pe^{2\foldPar-1})
   \\
   \vdots & \vdots & \ddots & \vdots
   \\
   f(\pe^{n-\foldPar}) & f(\pe^{n-\foldPar+1}) & \dots & f(\pe^{n-1})
  \end{pmatrix}
  :
  f(x)\in\LinpolyringK
  \right\}.
 \end{equation}
\end{definition}
Defining $\foldedVec=\left(\alpha^0 \ \pe^{\foldPar} \ \dots \ \pe^{n-\foldPar}\right)^\top$, we can write each codeword of \FGab{\foldPar,\pe;\lenFG,k} as
\begin{equation}\label{eq:defFGabShort}
  \left(f(\foldedVec) \ f(\pe\foldedVec) \ \dots \ f(\pe^{\foldPar-1}\foldedVec)\right)
\end{equation}
where $f(x)\in\LinpolyringK$.
For a fixed basis of $\mathbb{F}_{q^{\foldPar m}}$ over $\Fqm$ a codeword of an $h$-folded Gabidulin code \FGab{\foldPar,\pe;\lenFG,k} can be represented as a column vector $\vec{c}\in\colVec{\mathbb{F}_{q^{\foldPar m}}}{\lenFG}$, a matrix $\vec{C}\in\Fqm^{\lenFG\times\foldPar}$ or a matrix $\vecFq{\Mat{C}}\in\Fq^{\lenFG\times \foldPar m}$.
The $j$-th row of $\Mat{C}$ for $j\in\intervallincl{0}{\lenFG-1}$ is
\begin{equation*}
 \vec{c}_j=\left(f(\pe^{j\foldPar}) \ f(\pe^{j\foldPar+1}) \ \dots \ f(\pe^{(j+1)\foldPar-1})\right)
\end{equation*}
and can be seen as an element of the field $\mathbb{F}_{q^{\foldPar m}}$.
Folded Gabidulin codes are codes of length $\lenFG$ over a large field $\mathbb{F}_{q^{\foldPar m}}$ that can be decoded over the small field $\Fqm$.
An $\foldPar$-folded Gabidulin code considered over $\mathbb{F}_{q^{\foldPar m}}$ is a \emph{nonlinear} code over $\mathbb{F}_{q^{\foldPar m}}$ since $\mathbb{F}_{q^{\foldPar m}}$-linear combinations of codewords do not necessarily give codewords.

A folded Gabidulin code is $\Fqm$-linear since the unfolded code is an $\Fqm$-linear subspace of $\Fqm^n$.
This also implies that the code is \emph{linear} over $\Fq$.

The number of codewords in \FGab{\foldPar,\pe;\lenFG,k} is $|\FGab{\foldPar,\pe;\lenFG,k}|=q^{mk}$.
The code rate of a folded Gabidulin code is the same as the code rate of the unfolded code~\citep{Mahdavifar2012Listdecoding}, i.e.,
\begin{align}\label{eq:codeRateFGab}
 R=\frac{\log_{q^{\foldPar m}}\left(|\FGab{\foldPar,\pe;\lenFG,k}|\right)}{\lenFG}=\frac{k}{n}.
\end{align}

The following theorem shows that folded Gabidulin codes are MRD codes if and only if $\foldPar$ divides $k$.
\begin{theorem}[Minimum Distance of $\foldPar$-Folded Gabidulin Codes~{\citep[Theorem~1]{bartz2017folded}}]\label{thm:minimumDistance}
 \hskip5pt
 The min\-imum rank distance of an $\foldPar$-folded Gabidulin code $\FGab{\foldPar,\pe;\lenFG,k}$ of length $\lenFG=\frac{n}{\foldPar}$ is $\RankdistNoInput=\lenFG-\lceil\frac{k}{\foldPar}\rceil+1$.
\end{theorem}

\begin{proof}
 An $\foldPar$-folded Gabidulin code forms a group under addition since $\LinpolyringK$ forms an additive group over $\Fqm$.
 Thus the minimum distance of the code $\CRank=\FGab{\foldPar,\pe;\lenFG,k}$ is given by the minimum rank of a nonzero codeword, i.e.,
 \begin{equation*}
  \Rankdist{\CRank} = \min_{\Mat{C} \in \CRank^*} \rk_q(\vecFq{\Mat{C}})
 \end{equation*}
 where $\CRank^*\coloneqq\CRank\setminus\{\vec{0}\}$.
 Let $\Mat{C}\in\CRank^*$ be a codeword generated by the evaluation of $f(x)\in\LinpolyringK$ at the code locators $\{\pe^0,\pe^1,\dots,\pe^{n-1}\}$.
 Since $\lenFG\leq\foldPar m$ we have $\rk_q(\vecFq{\Mat{C}})\leq\lenFG$.
 If the row rank of $\Mat{C}$ is
 \begin{equation*}
  \rk_q(\vecFq{\Mat{C}})=\lenFG-z
 \end{equation*}
 then by $\Fq$-elementary row operations (Gaussian elimination) we get $\Mat{C}'$ with $z$ zero rows and $\lenFG-z$ linearly independent rows.
It follows that $\Mat{C}'$ is generated by the evaluation of $f(x)$ at the new code locators $\pe_i'$ that are obtained from $\pe_i$ by $\Fq$-elementary operations for all $i\in\intervallincl{0}{n-1}$.
 Thus the new code locators $\pe_0',\pe_1'\dots,\pe_{n-1}'$ are $\Fq$-linearly independent and we have $f(\pe_i')=0$ for all $i\in\intervallincl{0}{n-1}$ at most $k-1$ times.
 Hence the number of zero rows $z$ in $\Mat{C}'$ satisfies
 \begin{equation}
  z\leq\left\lfloor\frac{k-1}{\foldPar}\right\rfloor=\left\lceil\frac{k}{\foldPar}\right\rceil-1
 \end{equation}
 and
 \begin{equation}\label{eq:minRankCodewords}
  \Rankdist{\CRank} = \min_{\Mat{C} \in \CRank^*} \rk_q(\vecFq{\Mat{C}})\geq\lenFG-\left\lceil\frac{k}{\foldPar}\right\rceil+1.
 \end{equation}
 From the Singleton-like bound we have
 \begin{equation}\label{eq:SBfolded}
  \log_{q^{\foldPar m}} q^{mk} \leq \lenFG-\RankdistNoInput+1
 \quad \Longleftrightarrow \quad
 \RankdistNoInput \leq \lenFG-\frac{k}{h}+1.
 \end{equation}
 Combining~\eqref{eq:minRankCodewords} and~\eqref{eq:SBfolded} we get
 \begin{equation}
  \lenFG-\left\lceil\frac{k}{\foldPar}\right\rceil+1\leq \RankdistNoInput \leq \lenFG-\frac{k}{h}+1
 \end{equation}
 and the statement of the theorem follows, since $\RankdistNoInput$ is an integer.
\end{proof}

Thus folded Gabidulin codes fulfill the Singleton bound in the rank metric with equality, i.e., they are MRD codes, if and only if $\foldPar$ divides $k$.
If $h$ does \emph{not} divide $k$ then Theorem~\ref{thm:minimumDistance} shows that the code still has the best minimum distance for the given parameters $\lenFG,k$ and $\foldPar$ but the size of the code could be larger in this case.

\subsection{Decoding of Folded Gabidulin Codes}

Motivated by the results of Guruswami and Rudra~\citep{Guruswami2008Explicit} and Vadhan~\citep{Vadhan2011}, there is an interpolation-based decoding algorithm for folded Gabidulin codes~\citep{Mahdavifar2012Listdecoding}, which will be summarized in the following.

The algorithm consists of two steps: the \emph{interpolation step} computes $r\leq\intDim$ non-zero and (left) $\Linpolyring$-independent vectors of linearized polynomials
\begin{equation*}
\vec{Q}^{(i)}\!=\![Q_0^{(i)},Q_1^{(i)},\dots,Q_\intDim^{(i)}] \in \Linpolyring^{\intDim+1} \setminus \{\0\}, \ \forall i=1,\dots,r
\end{equation*}
such that they fulfill certain degree and evaluation conditions with respect to interpolation points obtained in a sliding-window manner from the received matrix $\vec{C}+\vec{E}$.
The \emph{root-finding step} finds all message polynomials $f$ of degrees $\deg f < k$ such that
\begin{equation*}
Q_0^{(i)} + \sum_{j=1}^{u} Q_j^{(i)} f(\pe^{j-1}x) = 0 \quad \forall \, i=1,\dots,r.
\end{equation*}
If the rank of the error matrix $\vec{E}$ is less than $\tfrac{\intDim}{\intDim+1}\left(\tfrac{\lenFG(\foldPar-\intDim+1)-k+1}{\foldPar-\intDim+1}\right)$, then at least one satisfactory interpolation vector $\vec{Q}^{(i)}$ exists, see~\citep{Mahdavifar2012Listdecoding}. The output list contains the transmitted message polynomial vector.
The algorithm can be considered as a partial unique or list decoder~\citep{bartz2017algebraic,bartz2017folded}.
A decoding scheme for folded Gabidulin codes with an improved decoding performance for high-rate codes was presented in~\citep{bartz2017algebraic,bartz2017folded}.

\section{Decoding of Symmetric Errors}
A \emph{symmetric error} is an error matrix where the transpose of the matrix coincides with the  matrix itself. When error matrices are symmetric, the correctable rank of the error can be increased by using special code constructions.

In~\citep{Gab05ic,GabidulinPilipchuk_SymmetricRankErrors_2004,GabidulinPilipchuck_SymmMatricesCorrectingErrors_2006}, it was shown that for Gabidulin codes that contain a linear subcode of symmetric matrices can correct symmetric error matrices  of rank  up to $(n -1)/2$.

In \citep{SpaceSymm-ISIT2021}, the condition of symmetric errors was relaxed and \emph{space-symmetric error matrices}, which have the property that their column and row spaces coincide, were considered. It is possible to use a Gabidulin code with the same property as in~\citep{Gab05ic,GabidulinPilipchuk_SymmetricRankErrors_2004,GabidulinPilipchuck_SymmMatricesCorrectingErrors_2006} to decode such space-symmetric errors of rank up to $2(n-k)/3$ with high probability.

\section{Further Classes of MRD Codes}\label{sec:furtherclasses}

In the last five years, there has been a growing interest in finding new MRD codes that are not equivalent\footnote{By equivalent, we mean that one code can be obtained by a semi-linear isometry of the rank metric: that is, multiplication of an invertible $n \times n$ matrix over $\Fq$ from the right, multiplication with a constant in $\Fqm^\ast$, and taking an automorphism (which fixes $\Fq$) of each codeword entry.} %
to a Gabidulin code. The first such family was discovered independently by \citet{sheekey2016new} (twisted Gabidulin codes) and \citet{otal2016explicit} (special case of twisted Gabidulin codes). These seminal works started a fruitful line of work, which resulted in various constructions of linear and non-linear MRD codes, for instance:
\begin{itemize}
	\item Generalizations of twisted Gabidulin codes in \citep[Remark~9]{sheekey2016new} and \citep{lunardon2018generalized, ot17, sheekey2019new, puchinger2017further}.
	\item Some further constructions: linear MRD codes of dimension 2 which are not generalized Gabidulin codes \citep{ho16}, new MRD codes from projective geometry \citep{csajbok2018new, csajbok2018newbis,bartoli2019new}, MRD codes with maximum idealizers, which are not generalized Gabidulin codes, for $n=7, q$ odd and $n=8, q\equiv 1\ (\text{mod } 3)$ \citep{csajbok2018maximum}, $\Fq$-linear MRD codes of $\Fq^{6\times 6}$ of dimension 12, minimum distance 5 \citep{marino2019mrd}.
        \end{itemize}

A recent survey by \citet{sheekey201913} gives a comprehensive overview of these code constructions. In the following, we briefly outline the construction of linear twisted Gabidulin codes.

A twisted Gabidulin code is defined by evaluating linearized polynomials, similar to Gabidulin codes. However, in contrast to Gabidulin codes, these polynomials do not have $q$-degree at most $k-1$, but have non-zero monomials of higher degree. The coefficients of these monomials are chosen in a special way such that the resulting code is an MRD code.
Figure~\ref{fig:illustration_twisted_Gabidulin_codes} illustrates twisted Gabidulin codes in different levels of generality.
Note that the codes in \citep{sheekey2016new} can be more general if we allow the codes to be non-linear over $\Fqm$, but we restrict ourselves to $\Fqm$-linear codes here.

\begin{definition}[{Twisted Gabidulin Code, \citep{sheekey2016new,puchinger2017further}}]\label{def:tgab_definition} %
  Let $n,k,\numTwists \in \NN$ with $k < n$ and $\ell \leq n-k$. Choose a
  \begin{itemize}
  \item \emph{hook vector}\footnote{For didactic reasons, this definition slightly differs from the one in \citep{puchinger2017further}, i.e., this is a special case.} $\hVec \in \{0,\dots,k-1\}^\numTwists$ and a
  \item \emph{twist vector} $\tVec \in \{1,\dots,n-k\}^\numTwists$ with distinct entries $t_i$, and let
  \item $\etaVec \in (\Fqm \setminus \{0\})^\numTwists$.
  \end{itemize}
  The set of \emph{\twisted linearized polynomials} over $\Fqm$ is defined by
  \begin{equation*}
    \evpolys = \left\{ f = \sum_{i=0}^{k-1} f_i x\qpow{i} + \sum_{j=1}^{\numTwists} \eta_j f_{h_j} x\qpow{k-1+t_j} : f_i \in \Fqm \right\}.
  \end{equation*}
  Let $\alpha_1,\dots,\alpha_n \in \Fqm$ be linearly independent over $\Fq$ and write $\alphaVec = [\alpha_1,\dots,\alpha_n]$.
  The \emph{\twistedC Gabidulin code} of length $n$ and dimension $k$ is given by
  \begin{equation*}
    \Cmult = \left\{ \big[f(\alpha_1),f(\alpha_2), \dots, f(\alpha_n)] \, : \, f \in \evpolys \right\}.
  \end{equation*}
\end{definition}

In Sheekey's case \citep{sheekey2016new} ($n=m$, $\numTwists=1$, $\hVec=(0)$, $\tVec=(1)$), a twisted Gabidulin code is MRD if $\eta = \eta_1$ has field norm $N_{\Fqm/\Fq}(\eta) \coloneqq \tfrac{\eta^{q^n-1}}{\eta^q-1}\neq (-1)^{nk}$. A non-zero $\eta$ with this property exists for any field size $q > 2$. \citet{puchinger2017further} also gave a necessary condition for the given more general class of twisted Gabidulin codes to be MRD: if we choose a chain of proper subfields $\Fq \subsetneq \mathbb{F}_{q^{s_0}} \subsetneq \mathbb{F}_{q^{s_1}} \subsetneq \dots \subsetneq \mathbb{F}_{q^{s_\numTwists}}=\mathbb{F}_{q^{m}}$, take the evaluation points $\alpha_i$ from $\mathbb{F}_{q^{s_0}}$ (this requires $n \leq s_0$), and choose $\eta_i \in \mathbb{F}_{q^{s_i}}\setminus \mathbb{F}_{q^{s_{i-1}}}$, then the resulting code is MRD. Note that we have $m \geq 2^\numTwists n$, so the codes are defined over relatively large fields for large $\numTwists$.

A finite semifield \citep{dickson1905finite,dickson1906commutative,albert1961generalized,knuth1963finite} is a finite division algebra with multiplicative identity, in which multiplication is not necessarily associative.
It is well-known (see, e.g., \citep{delaCruz2016algebraic}) that a certain class of finite semifields is in one-to-one correspondence to the set of (not necessarily linear) MRD codes of length $n$ and minimum rank distance $n$ over $\mathbb{F}_{q^n}$. Note that in this case, all non-zero codewords are invertible matrices, which corresponds to the fact that division is possible in a division ring. This correspondence gave rise to various constructions of MRD codes inequivalent to Gabidulin codes (see \citep{sheekey201913} for an overview), and also inspired some new MRD constructions for $d<n$. For instance, Sheekey's twisted Gabidulin codes \citep{sheekey2016new} coincide with generalized twisted fields \citep{albert1961generalized} for $k=1$, but still provide MRD codes for $k>1$.

When constructing new rank-metric codes, it is always of importance to show that the new codes are not equivalent to an existing code construction (see \citep{wan1996geometry, berger2003isometries, morrison2014equivalence} for the formal definition of equivalence in the rank metric). A practical tool that can show inequivalence for linear rank-metric codes in many cases was introduced by \citet{neri2020equivalence}. The method works well for several code constructions that are evaluation codes of linearized polynomials.

Decoding of Sheekey's twisted Gabidulin codes and its additive variants was studied in \citep{rosenthal2017decoding,randrianarisoa2017decoding,li2019interpolation,li2019decoding,kadir2021interpolation,kadir2020decoding}. The general decoding principle of \citep{rosenthal2017decoding,randrianarisoa2017decoding} is to solve a linear system of equations with one fewer equation than usual (cf.~Section~\ref{sec:decoding_gabidulin_codes}). The solutions of these equations form a one-dimensional space. In addition, we have one non-linear equation that arises from the special structure of the evaluation polynomials. It is shown in \citep{rosenthal2017decoding,randrianarisoa2017decoding} that we can efficiently find the unique solution of the decoder by solving this additional non-linear equation. So far, there is no polynomial-time decoder for twisted Gabidulin codes with multiple twists, or one twist $t_1>1$.

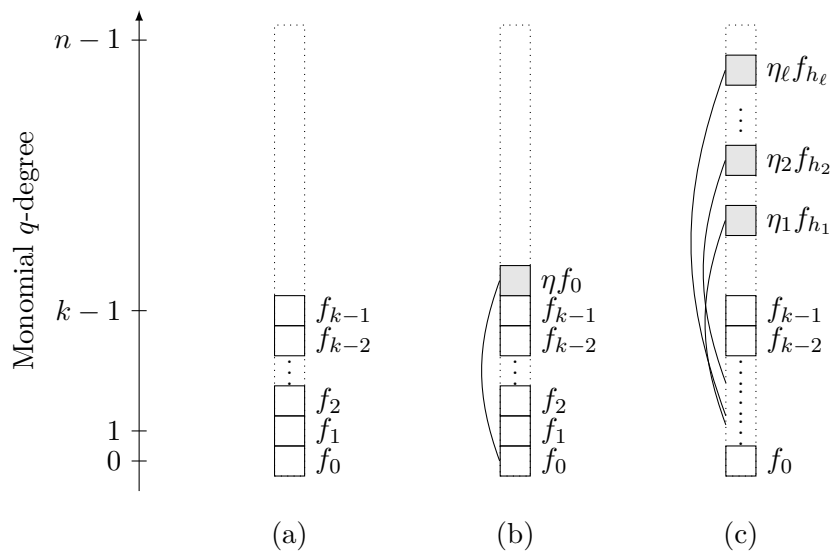
\begin{figure}[ht!]

	\begin{center}
		\begin{tikzpicture}
		\def\ylabelpos{-1cm}
		\def\xlevelminusone{-2cm}
		\def\xlevelzero{0cm}
		\def\xlevelone{3cm}
		\def\xleveltwo{6cm}
		\def\xlevelthree{6cm}
		\def\xwidth{0.4cm}
		\def\ywidth{0.4cm}
		\def\ydist{1cm}
		\def\nval{15}
		\def\kval{5}
		\def\twistval{6}
		\def\twistvaltwo{8}
		\def\twistvalthree{12}
		\def\myeps{0.1cm}
		\tikzstyle{coeffNodesPure}=[draw, rectangle, color=blue, minimum width=\xwidth, minimum height=\ywidth, inner sep=0pt]
		\tikzstyle{coeffNodes}=[draw, rectangle, minimum width=\xwidth, minimum height=\ywidth, inner sep=0pt]
		\tikzstyle{twistNodes}=[draw, rectangle, fill=black!10, minimum width=\xwidth, minimum height=\ywidth, inner sep=0pt]
		\tikzstyle{coeffNodesLevels}=[draw, rectangle, color=blue, minimum width=\xwidth, minimum height=\ywidth, inner sep=0pt]

		\draw[->,>=latex] (\xlevelminusone,-\ywidth) to node[above,rotate=90,yshift=1.2cm] {Monomial $q$-degree} (\xlevelminusone,\nval*\ywidth);
		\draw (\xlevelminusone+\myeps,0) -- (\xlevelminusone-\myeps,0) node[left] {$0$};
		\draw (\xlevelminusone+\myeps,\ywidth) -- (\xlevelminusone-\myeps,\ywidth) node[left] {$1$};
		\draw (\xlevelminusone+\myeps,5*\ywidth) -- (\xlevelminusone-\myeps,5*\ywidth) node[left] {$k-1$};
		\draw (\xlevelminusone+\myeps,\nval*\ywidth-\ywidth) -- (\xlevelminusone-\myeps,\nval*\ywidth-\ywidth) node[left] {$n-1$};

		\draw[dotted] (\xlevelzero-0.5*\xwidth,-0.5*\ywidth) rectangle (\xlevelzero+0.5*\xwidth,\nval*\ywidth-0.5*\ywidth);
		\node[coeffNodes,label=right:{$f_0$}]  (levelZerocoeff0) at (\xlevelzero+0,0*\ywidth) {};
		\node[coeffNodes,label=right:{$f_1$}]  (levelZerocoeff1) at (\xlevelzero+0,1*\ywidth) {};
		\node[coeffNodes,label=right:{$f_2$}]  (levelZerocoeff2) at (\xlevelzero+0,2*\ywidth) {};
		\node 								   (levelZerocoeffdots) at (\xlevelzero+0,3.2*\ywidth) {\footnotesize $\vdots$};
		\node[coeffNodes,label=right:{$f_{k-2}$}]  (levelZerocoeffk2) at (\xlevelzero+0,4*\ywidth) {};
		\node[coeffNodes,label=right:{$f_{k-1}$}]  (levelZerocoeffk1) at (\xlevelzero+0,5*\ywidth) {};

		\draw[dotted] (\xlevelone-0.5*\xwidth,-0.5*\ywidth) rectangle (\xlevelone+0.5*\xwidth,\nval*\ywidth-0.5*\ywidth);
		\node[coeffNodes,label=right:{$f_0$}]  (levelOnecoeff0) at (\xlevelone+0,0*\ywidth) {};
		\node[coeffNodes,label=right:{$f_1$}]  (levelOnecoeff1) at (\xlevelone+0,1*\ywidth) {};
		\node[coeffNodes,label=right:{$f_2$}]  (levelOnecoeff2) at (\xlevelone+0,2*\ywidth) {};
		\node 								   (levelOnecoeffdots) at (\xlevelone+0,3.2*\ywidth) {\footnotesize $\vdots$};
		\node[coeffNodes,label=right:{$f_{k-2}$}]  (levelOnecoeffk2) at (\xlevelone+0,4*\ywidth) {};
		\node[coeffNodes,label=right:{$f_{k-1}$}]  (levelOnecoeffk1) at (\xlevelone+0,5*\ywidth) {};
		\node[twistNodes,label=right:{$\eta f_0$}] (levelOnetwist1) at (\xlevelone+0,\twistval*\ywidth) {};
		\draw[bend angle=20, bend left] (levelOnecoeff0.west) to (levelOnetwist1.west);

		\draw[dotted] (\xlevelthree-0.5*\xwidth,-0.5*\ywidth) rectangle (\xlevelthree+0.5*\xwidth,\nval*\ywidth-0.5*\ywidth);
		\node[coeffNodes,label=right:{$f_0$}]  (levelThreecoeff0) at (\xlevelthree+0,0*\ywidth) {};
		\node 								   (levelThreecoeff1) at (\xlevelthree+0,1.2*\ywidth) {\footnotesize $\vdots$};
		\node 								   (levelThreecoeff2) at (\xlevelthree+0,2.2*\ywidth) {\footnotesize $\vdots$};
		\node 								   (levelThreecoeffdots) at (\xlevelthree+0,3.2*\ywidth) {\footnotesize $\vdots$};
		\node[coeffNodes,label=right:{$f_{k-2}$}]  (levelThreecoeffk2) at (\xlevelthree+0,4*\ywidth) {};
		\node[coeffNodes,label=right:{$f_{k-1}$}]  (levelThreecoeffk1) at (\xlevelthree+0,5*\ywidth) {};
		\node[twistNodes,label=right:{$\eta_1 f_{h_1}$}] (levelThreetwist1) at (\xlevelthree+0,8*\ywidth) {};
		\node[twistNodes,label=right:{$\eta_2 f_{h_2}$}] (levelThreetwist2) at (\xlevelthree+0,10*\ywidth) {};
		\node (levelThreetwistdots) at (\xlevelthree+0,11.6*\ywidth) {\footnotesize $\vdots$};
		\node[twistNodes,label=right:{$\eta_\ell f_{h_\ell}$}] (levelThreetwist3) at (\xlevelthree+0,13*\ywidth) {};
		\draw[bend angle=20, bend left] (\xlevelthree-0.5*\xwidth,1.2*\ywidth) to (levelThreetwist1.west);
		\draw[bend angle=20, bend left] (\xlevelthree-0.5*\xwidth,2.6*\ywidth) to (levelThreetwist2.west);
		\draw[bend angle=20, bend left] (\xlevelthree-0.5*\xwidth,1.5*\ywidth) to (levelThreetwist3.west);

		\node at (\xlevelzero, \ylabelpos) {$\mathrm{(a)}$};
		\node at (\xlevelone,  \ylabelpos) {$\mathrm{(b)}$};
		\node at (\xlevelthree,\ylabelpos) {$\mathrm{(c)}$};

		\end{tikzpicture}
	\end{center}
	\vspace{-0.5cm}
	\caption{Illustration of evaluation polynomials of twisted Gabidulin codes. Boxes $\Box$ correspond to possibly non-zero coefficients. Arcs connect the corresponding hook and twist (gray background) coefficients (cf.~\cref{def:tgab_definition}). (a) shows the evaluation polynomials of a Gabidulin code \citep{Delsarte_1978,Gabidulin_TheoryOfCodes_1985,Roth_RankCodes_1991}: $f_0,\dots,f_{k-1}$ can be chosen arbitrarily, and all higher-degree coefficients are zero. (b) shows the evaluation polynomials of Sheekey's \citep{sheekey2016new} twisted Gabidulin codes, where the $k$-th coefficient is a multiple of the zeroth coefficient. (c) shows the evaluation polynomials of the generalization of twisted Gabidulin codes in \cref{def:tgab_definition}: there are several non-zero coefficients of degree $\geq k$, and each depends on one of the coefficients $f_0,\dots,f_{k-1}$. \lh{Let float?}}
	\label{fig:illustration_twisted_Gabidulin_codes}
\end{figure}

\chapter{Applications to Code-Based Cryptosystems}\label{chap:crypto}

Most currently-used public-key cryptosystems are based on number-theoretic hard problems.
As soon as large-scale quantum computers exist, it will be possible to solve these seemingly hard problems in polynomial time using Shor’s algorithm \citep{shor1999polynomial}.
As it is likely that capable quantum computers will exist in the near future, there is a strong need to develop, analyze and standardize post-quantum cryptosystems.
The importance is also reflected in the currently running post-quantum cryptography standardization competition by the National Institute of Standards and Technology (NIST).

Code-based cryptosystems are, besides lattice-based systems, among the very few candidates for post-quantum-secure public key encryption (PKE) schemes and key encapsulation mechanisms (KEM). Their security is based on hard computational problems in coding theory and en-/decryption often corresponds to en-/decoding of a code. One of the biggest challenges in designing code-based cryptosystems is to reduce the size of the public key. The most prominent example of a code-based cryptosystem is the McEliece system \citep{mceliece1978public}, which is based on binary Goppa codes (in the Hamming metric).

Rank-metric code-based cryptosystems have been studied since 1991. Their main advantage is that the generic decoding problem in the rank metric, which is often the underlying hard problem of a code-based cryptosystem, is seemingly harder to solve than the corresponding problem in the Hamming metric. This results in significantly reduced key sizes compared to similar Hamming-metric schemes. Especially in the last 10 years, there have been several important developments related to rank-metric code-based cryptosystems. Two of the known systems, RQC and ROLLO, have been considered in the NIST standardization. According to the \textit{Status Report on the Second Round of the NIST Post-Quantum Cryptography Standardization Process}~\citep[Section 3.14, 3.16]{nist2020status}, although neither of them advanced on in the PQC standardization process due to that the security analysis of them needs more time to mature, NIST encouraged further research on rank-based cryptosystems, as their key and ciphertext sizes remain competitive compared to traditional Hamming-metric codes.

In this chapter, we present several known rank-metric code-based cryptosystems.
We start with a discussion on hard problems in the rank metric (Section~\ref{ssec:hard_problems}).
In Section~\ref{ssec:rank_crypto_in_McEliece}, we present rank-metric cryptosystems that can be seen as variants of the McEliece cryptosystem in the rank metric.
Most prominently, we discuss the GPT system (the first known rank-metric code-based cryptosystem which basically adapts the McEliece principle) and its variants, including the---as of today---unbroken one by Loidrau, the one based on Gabidulin matrix codes, and the one based on quasi-cyclic low-rank-parity-check (LRPC) codes (which is used in the NIST submission ROLLO).
We outline two schemes that are based on the hardness of list decoding rank-metric codes in Section~\ref{ssec:rank_crypto_list_decoding}.
In Section~\ref{ssec:rank_crypto_rqc}, we discuss the NIST submission RQC, which is based on rank quasi-cyclic (or, more generally, ideal rank) codes.
We briefly discuss signature schemes in Section~\ref{ssec:rank_crypto_signatures} and conclude the chapter with a parameter comparison of unbroken rank-metric code-based cryptosystems in Section~\ref{ssec:parameters}.

\section{The Hardness of Problems in the Rank Metric}\label{ssec:hard_problems}

\subsection{Random Codes}
In analogy to the Hamming metric, several problems in the rank metric can be defined and their hardness can be analyzed.
We start with defining problems for \emph{random} rank-metric codes.
For these problems, we need the definition of a (random code) \emph{rank syndrome decoding} (RSD) distribution as follows.
\begin{definition}[Random Code Rank Syndrome Decoding (RSD) Distribution]
   \phantom{v}
  \begin{itemize}
  \item
    Input: $q,n,k,w,m$
  \item
    Choose uniformly at random
    \begin{itemize}
    \item $\H \assignRand \{\A \in \Fqm^{(n-k)\times n}: \rank_{q^m}(\A) = n-k \}$
    \item $\x \assignRand  \{\a \in \Fqm^{n} : \wtR(\a) = w\} $
    \end{itemize}
  \item Output: $(\H,\H\x^{\top})$
  \end{itemize}
  \end{definition}
  Hence, the output is a random parity-check matrix and a random syndrome which stems from a random error of weight $w$.

The first problem that we are considering is the \emph{search RSD problem}.
\begin{problem}[Search RSD Problem]\label{prob:search-rsd}
  \phantom{v}
  \begin{itemize}
  \item
    Input: $(\H,\y)$ from the RSD Distribution.
  \item
    Goal: Find \textbf{one} $\x \in \{\a \in \Fqm^{n} : \wtR(\a) = w\}$ such that $\H\x^{\top} = \y^{\top}$.
  \end{itemize}
\end{problem}
  This problem is therefore equivalent to decoding of the random rank-metric code defined by $\H$.
If $w \leq \dhalffrac$, where $d$ denotes the minimum rank distance of the code defined by $\H$, the search RSD problem returns the unique decoding result. For larger radii, the search RSD problem returns only \emph{one} error of suitable weight.

The hardness of Problem~\ref{prob:search-rsd} has been investigated by \citet{GaboritZemor_HardnessRankMetric_2016}. Their main result shows that if there is a probabilistic polynomial-time algorithm that solves the search RSD problem, then every problem in the complexity class NP can be solved by a probabilistic polynomial-time algorithm.

In \citep{Chabaud1996,Ourivski2002,Gaborit_DecodingAttack_2016,aragon2018new,bardet2020algebraic,BBC20}, algorithms that solve the Search RSD problem are proposed. The recent work by \citet{BBC20} results in the smallest computational complexity for many sets of parameters.
This recent algorithmic breakthrough in the search for low-rank codewords in linear codes significantly reduced the complexity of solving the Search RSD problem. Due to this, the parameters of several cryptographic schemes had to be increased to ensure the same security level.
This algorithm consists of rewriting the problem in terms of a bivariate polynomial system, solving it
with the use of a Gröbner basis
\citep{bardet2020algebraic,BBC20}. This improvement had a major impact on the parameters of all cryptographic schemes whose security relies on the
difficulty of the Search RSD problem.

If $w>\dhalffrac$, there might be several codewords in the list and it makes sense to modify the problem.
That means, consider the case that multiple $\x$ with $\wtR(\x) = w$ such that $\H\x^{\top} = \y^{\top}$ exist, but only one of them corresponds to a codeword that is the message. If we decide on a ``wrong'' $\x$, we cannot retrieve any information about the message. Hence, we need to consider the following \emph{list} search RSD problem.

\begin{problem}[{List} Search RSD Problem]\label{pro:list-search-rsd}\phantom{v}
  \begin{itemize}
  \item
    Input: $(\H,\y)$ from the RSD Distribution.
  \item
    Goal: Find \textbf{all} $\x \in \{\a \in \Fqm^{n} : \wtR(\a) = w\}$ such that $\H\x^{\top} = \y^{\top}$.
  \end{itemize}
\end{problem}
Note that the \emph{list} search RSD problem is at least as hard as the search RSD problem as we have to find \emph{all} vectors $\vec{x}$ with this property.

We can also state the decisional version of the RSD problem as follows.
\begin{problem}[Decision RSD Problem]\label{prob:decision_RSD}
  \phantom{v}
  \begin{itemize}
  \item
    Input: $(\H,\y)$ from the RSD Distribution.
  \item
    Goal: Decide with non-negligible advantage whether $(\H,\y)$ came from the RSD distribution or the uniform distribution over $ \Fqm^{(n-k)\times n} \times \Fqm^{n-k}$.
  \end{itemize}
\end{problem}

\subsection{Gabidulin Codes}
As a second step, we consider the hardness of problems when decoding \emph{Gabidulin} codes.
Similar to the problems in the previous section, we first define a Gabidulin code RSD distribution.
\begin{definition}[Gabidulin Code Rank Syndrome Decoding (RSD) Distribution]\phantom{v}
  \begin{itemize}
  \item
    Input: $q,n,k,w,m$
  \item
    Choose uniformly at random
    \begin{itemize}
    \item $\H \xleftarrow{\$} \mathcal{H}$, where $\mathcal{H}$ is the set of all parity check matrices of Gabidulin codes $\Gab{n,k}$ over $\Fqm$
    \item $\x \xleftarrow{\$} \{\a \in \Fqm^{n} : \wtR(\a) = w\} $
    \end{itemize}
  \item
    Output: $(\H,\H\x^{\top})$.
  \end{itemize}
\end{definition}

\begin{problem}[Gabidulin Code Search RSD Problem]\label{pro:gab-search-rsd}\phantom{v}
  \begin{itemize}
  \item
    Input: $(\H,\y)$ from the Gabidulin Code RSD Distribution.
  \item
    Goal: Find \textbf{one} $\x \in \{\a \in \Fqm^{n} : \wtR(\a) = w\}$ such that $\H\x^{\top} = \y^{\top}$.
  \end{itemize}
\end{problem}
If $w \leq \dhalffrac$ where $d = n-k+1$ is the minimum distance of the Gabidulin code, this problem is efficiently solvable by any Gabidulin decoder, see e.g., Section~\ref{sec:decoding_gabidulin_codes}. The fastest known decoder from \citet{PuchingerWachterzeh-ISIT2016} has sub-quadratic decoding complexity over $\Fqm$ (see~\ref{sec:otherDecAlgo}).

For larger $w$, there is no known polynomial-time decoder, but the syndrome-based decoder can be adapted by simply searching through the solution space of the syndrome key equation. That means that all solutions of a linear homogeneous system of equations over $\Fqm$ with $n-k-w$ equations and $w+1$ unknowns have to be considered. Thus, the work factor of this search is $q^{m(2w-n-k)}$.
In \citep{ML-Gab-PQCrypto}, solving Problem~\ref{pro:gab-search-rsd} was accelerated compared to searching through the solution space of the key equation.
The algorithm by~\citet{ML-Gab-PQCrypto} consists of repeatedly guessing a
subspace that should have a large intersection with the error row and/or column
space. The guessed space is then used as erasures in an Gabidulin error-erasure
decoder. The algorithm terminates when the intersection of the
guessed space and the error row and/or column space is large enough such that
the decoder outputs a codeword that is close enough to the received word.
The expected work factor of this randomized decoding approach is a polynomial term times $q^{m(n-k)-w(n+m)+w^2+\min\{2\xi(\frac{n+k}{2}-\xi),wk\} }$, where $n$ is the code length, $q$ the size of the base field, $m$ the extension degree of the field, $k$ the code dimension, $w$ the number of errors, and $\xi := w-\tfrac{n-k}{2}$.

Further notice that the hardness of Problem~\ref{pro:gab-search-rsd} is upper bounded by the hardness of Problem~\ref{pro:gab-search-list-rsd} below where \emph{all} $\vec{x}$ with this property have to be found.

\begin{problem}[Gabidulin Code {List} Search RSD Problem]\label{pro:gab-search-list-rsd}
  \phantom{v}
  \begin{itemize}
  \item
    Input: $(\H,\y)$ from the Gabidulin Code RSD Distribution.
  \item
    Goal: Find \textbf{all} $\x \in \{\a \in \Fqm^{n} : \wtR(\a) = w\}$ such that $\H\x^{\top} = \y^{\top}$.
  \end{itemize}
\end{problem}

When $w \leq \nkhalffrac$, the decoding result is unique and the \emph{Gabidulin Code {List} Search RSD Problem} (Problem~\ref{pro:gab-search-list-rsd}) is equivalent to the \emph{Gabidulin Code Search RSD Problem} (Problem~\ref{pro:gab-search-rsd}) and efficiently solvable.
However, for $w > \nkhalffrac$, there are cases in which the two problems differ substantially. For instance, the problem of recovering the message from the ciphertext in the Faure--Loidreau (FL) system or in the RQC system can be reduced to the \emph{Gabidulin Code {List} Search RSD Problem} since by finding all errors of weight $w$, the actual error will be contained in the output list. The \emph{Gabidulin Code Search RSD Problem} is not necessarily useful when $w > \nkhalffrac$ as it find only one word of the list.

The \emph{Gabidulin Code {List} Search RSD Problem} (Problem~\ref{pro:gab-search-list-rsd}) when $w > \nkhalffrac$ is equal to list decoding a corrupted codeword of a Gabidulin code, where the error has rank weight $w$. The exact complexity of this problem is unknown, but there are some partial answers\footnote{See also Section~\ref{sec:list-dec-gab} for a summary on list decoding of rank-metric codes.}.
There are families of Gabidulin codes of rate $R \geq 1/6$ for which some received words have lists of exponential size even for decoding only one error beyond $\nkhalffrac$, see \citep{RavivWachterzeh_GabidulinBounds_journal}.
This result was recently generalized to more general classes of MRD codes by \citet{TrombettZullo-ListDecRM2020}.
These bounds on the list size of some Gabidulin codes proves that an equally general algorithm as the Guruswami--Sudan list decoding algorithm for Reed--Solomon codes cannot exist for Gabidulin codes.
For $w \geq n-\sqrt{n(n-d)}$, the \emph{Gabidulin Code {List} Search RSD Problem} (Problem~\ref{pro:gab-search-list-rsd}) has exponential worst-case complexity since the list size grows exponential in $n$, see \citep[Theorem~1]{wa13a}. 
The worst-case complexity can be considered as an indication for the average complexity which is usually considered in cryptography to assess the security level.

Notice that in \citep{GuruswamiWang2013}, subcodes of Gabidulin codes were efficiently list-decoded, but this result does not apply to Gabidulin codes themselves.
Based on the previous partial answers on the hardness of list decoding Gabidulin codes, it is widely conjectured that Problem~\ref{pro:gab-search-list-rsd} is hard.

For completeness, the decisional version of Problem~\ref{pro:gab-search-rsd} is shown below.

\begin{problem}[Gabidulin Code Decision RSD Problem]\phantom{v}
  \begin{itemize}
  \item
    Input: $(\H,\y)$ from the Gabidulin Code RSD Distribution.
  \item
    Goal: Decide with non-negligible advantage whether $(\H,\y)$ came from the Gabidulin Code RSD distribution or the uniform distribution over $ \Fqm^{(n-k)\times n} \times \Fqm^{n-k}$
  \end{itemize}
\end{problem}

To solve the correspoding decisional problem, no faster approach than trying to solve the associated search problems is known. This is usually done for all decoding-based problems. A similar decisional problem can be given for list decoding.

\subsection{Problems Related to Decoding Interleaved Rank-Metric Codes}
In this subsection, we consider problems related to random interleaved rank-metric codes and later also related to interleaved Gabidulin codes.
\begin{definition}[Interleaved Random Code RSD Distribution]
  \phantom{v}
  \begin{itemize}
  \item
    Input: $q,n,k,w,m$.
  \item
    Choose uniformly at random
    \begin{itemize}
    \item $\H \xleftarrow{\$} \{\A \in \Fqm^{(n-k)\times n}: \rank_{q^m}(\A) = n-k \}$
    \item $\X \xleftarrow{\$} \{\X \in \Fqm^{u\times n} : \wtR(\X) = w\} $
    \end{itemize}
  \item
    Output: $(\H,\H\X^{\top})$.
  \end{itemize}
\end{definition}

\begin{problem}[Interleaved Random Code Search RSD Problem]\label{pro:int-rand-search-rsd}
  \phantom{v}
  \begin{itemize}
  \item
    Input: $(\H,\Y)$ from the Interleaved Random Code RSD Distribution.
  \item
    Goal: Find \textbf{one} $\X \in \{\X \in \Fqm^{u\times n} : \wtR(\X) = w\}$ such that $\H\X^{\top} = \Y^{\top}$.
  \end{itemize}
  \end{problem}

This problem is clearly as least at hard as solving the same problem for a specific code, in particular, the \emph{Interleaved Gabidulin Code Search RSD Problem} as introduced in the following.

The hardness of both problems is discussed afterwards.

\begin{definition}[Interleaved Gabidulin Code RSD Distribution]
  \phantom{v}
  \begin{itemize}
  \item
    Input: $q,n,k,w,m$.
  \item
    Choose uniformly at random
    \begin{itemize}
    \item $\H \xleftarrow{\$} \mathcal{H}$, where $\mathcal{H}$ is the set of all parity check matrices of Gabidulin codes $\Gab{n,k}$ over $\Fqm$.
    \item $\X \xleftarrow{\$} \{\X \in \Fqm^{u\times n} : \wtR(\X) = w\} $.
    \end{itemize}
  \item
    Output: $(\H,\H\X^{\top})$.
  \end{itemize}
\end{definition}

\begin{problem}[Interleaved Gabidulin Code Search RSD Problem]\label{pro:int-search-rsd}
  \phantom{v}
  \begin{itemize}
  \item
    Input: $(\H,\Y)$ from the Interleaved Gabidulin Code RSD Distribution.
  \item
    Goal: Find \textbf{one} $\X \in \{\A \in \Fqm^{u \times n} : \wtR(\A) = w\}$ such that $\H\X^{\top} = \Y^{\top}$.
  \end{itemize}
\end{problem}

The \emph{Interleaved Gabidulin Code Search RSD Problem} is equal to decoding a corrupted codeword of an interleaved Gabidulin code of interleaving order $u$, where the error has rank weight $w$ over $\Fq$. This problem has been extensively studied and it turns out that the hardness depends on $w$ and the $\Fqm$-rank of $\X$.
In the following, we discuss the hardness of Problem~\ref{pro:int-rand-search-rsd} (interleaved random code) and Problem~\ref{pro:int-search-rsd} (interleaved Gabidulin code). Notice that Problem~\ref{pro:int-rand-search-rsd} is clearly at least as hard as Problem~\ref{pro:int-search-rsd}.
\begin{itemize}
\item For $w\leq\lfloor \frac{n-k}{2}\rfloor$, we can decode each of the $u$ rows separately as the decoding result is guaranteed to be unique due to the minimum distance of the Gabidulin code. Therefore, Problem~\ref{pro:int-rand-search-rsd} reduces to $u$ instances of Problem~\ref{prob:search-rsd} with the corresponding hardness considerations. Similarly, Problem~\ref{pro:int-search-rsd} reduces to $u$ instances of Problem~\ref{pro:gab-search-rsd} which is efficiently solvable by any {Bounded Minimum Distance} (BMD) decoder for Gabidulin codes.
\item For the special case of high interleaving order $u\geq w$ where $w \leq n-k-1$ and $\rank_{q^m}(\X)=w$, Problem~\ref{pro:int-rand-search-rsd} and Problem~\ref{pro:int-search-rsd} both have a polynomial-time solution (without failure probability) given by the rank-metric Metzner--Kapturowski decoder by~\citet{MK-ISIT2021}. Thus, decoding (any) linear interleaved rank-metric code of high interleaving order is an easy problem.
\item For $w\leq \lfloor\frac{u}{u+1} (n-k)\rfloor$, the \emph{Interleaved Gabidulin Code Search RSD Problem} (Problem~\ref{pro:int-search-rsd}) can be solved in $O(un^2)$ operations in $\Fqm$, cf.~\citep{Sidorenko2011SkewFeedback}, with high probability ($> 1-\frac{4}{q^m}$).

For the unlikely cases (with probability less than $\frac{4}{q^m}$) of decoding failure, it is not proven to be a hard problem. Although there is no such general complexity result, decoding explicitly the errors for which all known decoders fail
is considered hard by the community and has been subject to intensive research for more than 13 years (since the decoder by~\citet{Loidreau_Overbeck_Interleaved_2006} was proposed).
In the Hamming metric, the equivalent problem for RS codes has been studied since the work by~\citet{Krachkovsky1997Decoding} and more than a dozen of papers have dealt with the decoding algorithms for these codes since then. None of these papers was able to give a polynomial-time decoding algorithm for the case when Krachkovsy--Lee's algorithm fails.
Furthermore, the same results on the list size as for Gabidulin codes apply for \emph{interleaved} Gabidulin codes since it consists of $u$ parallel Gabidulin codewords and at the same time can be seen as a single codeword of a Gabidulin code with larger field size.
\item For $w \geq d$, the complexity of both interleaved problems grows exponentially in $n$.
\item For \emph{horizontal} interleaved random codes, the corresponding problem is called \emph{Rank Support Learning} (RSL) and was introduced by~\citet[Definition 7]{Gaborit2017IBE}.
The recent paper by \citet{BardetBriaud-RSL2021} proposed an algebraic attack on RSL which clearly outperforms the previous attacks.
\item The \emph{Decoding one Out of Many} approach by \citet{Sendrier-DOOM} in the Hamming metric provides a way to recover one of the $u$ error vectors (and therefore the span of the whole matrix $\vec{X}$). However, this work has not yet been adapted to the rank metric. If it will be adapted, it should reduce the exponent of the work factor of brute-force decoding for Problem~\ref{pro:int-rand-search-rsd} (of course this also applies for Problem~\ref{pro:int-search-rsd}, but this problem can efficiently be solved up to $w\leq \lfloor\frac{u}{u+1} (n-k)\rfloor$ in any case).
\end{itemize}

\vspace{3ex}
The following table provides an overview of the previously considered problems and the rank-metric code-based cryptosystems that rely on these problems.

\begin{table}[htb]
	\centering
	\begin{tabular}{l|p{4cm}|p{5cm}}
		Problem No. & Problem & Cryptosystems\\\hline
		Problem~\ref{prob:search-rsd}&Search RSD Problem & All GPT variants (including Loidreau's new system)\\
		Problem~\ref{pro:list-search-rsd} & List Search RSD Problem & RQC (ciphertext),\newline Faure--Loidreau (ciphertext),\newline RAMESSES (public key, ciphertext) \\
		Problem~\ref{pro:gab-search-rsd}&Gabidulin Code Search RSD Problem& Faure--Loidreau\\
		Problem~\ref{pro:gab-search-list-rsd}&Gabidulin Code List Search RSD Problem& RQC, Faure--Loidreau\\
		Problem~\ref{pro:int-search-rsd}&Interleaved Gabidulin Code Search RSD Problem&Original Faure--Loidreau\\
	\end{tabular}
	\caption{Overview of problems in the rank metric and the corresponding cryptosystems.}
\end{table}

\section{McEliece-like Systems}\label{ssec:rank_crypto_in_McEliece}

In 1978, \citet{mceliece1978public} proposed a public-key cryptosystem whose security is based on the generic decoding problem for linear error-correcting codes.

Compared to the widely-used Rivest, Shamir and Adleman (RSA) public-key cryptosystem~\citep{rivest1978method}, the McEliece cryptosystem suffers from larger public keys (when employing Goppa codes as originally suggested), which is a drawback for practical applications.

In general, to use a certain class of codes in the McEliece principle, the challenge is to find codes with a good error-correction capability and a weak code structure or a code structure that can be hidden effectively, in order to reduce the public key size.
Since the proposal of the original McEliece cryptosystem, which is based on binary Goppa codes,  several variants based on different code families were proposed.
Besides McEliece variants based on codes in the Hamming metric, there are McEliece-like cryptosystems that are based on codes in the rank metric.
One benefit of rank-metric code-based cryptosystems is, that the generic decoding problem in the rank metric is significantly harder than in the Hamming metric which allows for very small public keys~\cite{}.

In the following, we give a summary of existing McEliece-like cryptosystems based on rank-metric codes, including variants based on algebraic rank-metric codes (e.g., Gabidulin codes) as well as random rank-metric codes (e.g., LRPC codes).

\subsection{Description of the Cryptosystems}

For the description of the different rank-metric code-based cryptosystems we use the following terminology and notation.

\begin{itemize}
 \item Security level (bit): $\secLevel$
 \item Public key: $\pk$
 \item Secret key: $\sk$
 \item Plaintext: $\plainText$
 \item Ciphertext: $\cipherText$
 \item Ciphertext length $\ct$
 \item Set of global parameters: $\param$
 \item Deterministic assignment: $\assignDet$
 \item Random assignment: $\assignRand$
 \item Setup: $\param\gets\setup{1^\secLevel}$
 \item Key generation: $(\pk,\sk)\gets\keyGen{\param}$
 \item Encryption: $\cipherText\gets\enc{\pk, \plainText}$
 \item Decryption: $\plainText\gets\dec{\sk, \cipherText}$
 \item $\GL{n}{\Fq}$: the set of all full-rank matrices of $\Fq^{n\times n}$
\end{itemize}

Unless otherwise stated, the parameters $q,m,n$ and $k$ are chosen according to Table~\ref{tab:generalParameters}.
\begin{table}[h]
\renewcommand{\arraystretch}{1.6} %
\begin{center}
\begin{tabular}{c|l|l}
Parameter & Stands for & Restriction \\
\hline
$q$ & field order & prime power \\
$m$ & extension degree & $1 \leq m$ \\
$n$ & code length & $n \leq m$ \\
$k$ & code dimension & $k \leq n$ \\
\end{tabular}
\caption{Parameters of the GPT variants.}
\label{tab:generalParameters}
\end{center}
\end{table}

\subsection{McEliece-type Rank-Metric Cryptosystems based on Gabidulin Codes}

Gabidulin, Paramonov, Tretjakov (GPT) proposed a Gabidulin code-based cryptosystem~\citep{gabidulin1991ideals} with very compact keys compared to the original McEliece cryptosystem~\citep{mceliece1978public}.

Compared to codes in the Hamming metric, rank metric codes have an even stronger structure that can be exploited by an attacker to break the corresponding cryptosystems.
As a result, the GPT cryptosystem has a long history of attacks and fixes, which is illustrated in Figure~\ref{fig:overviewGTPattacks}.

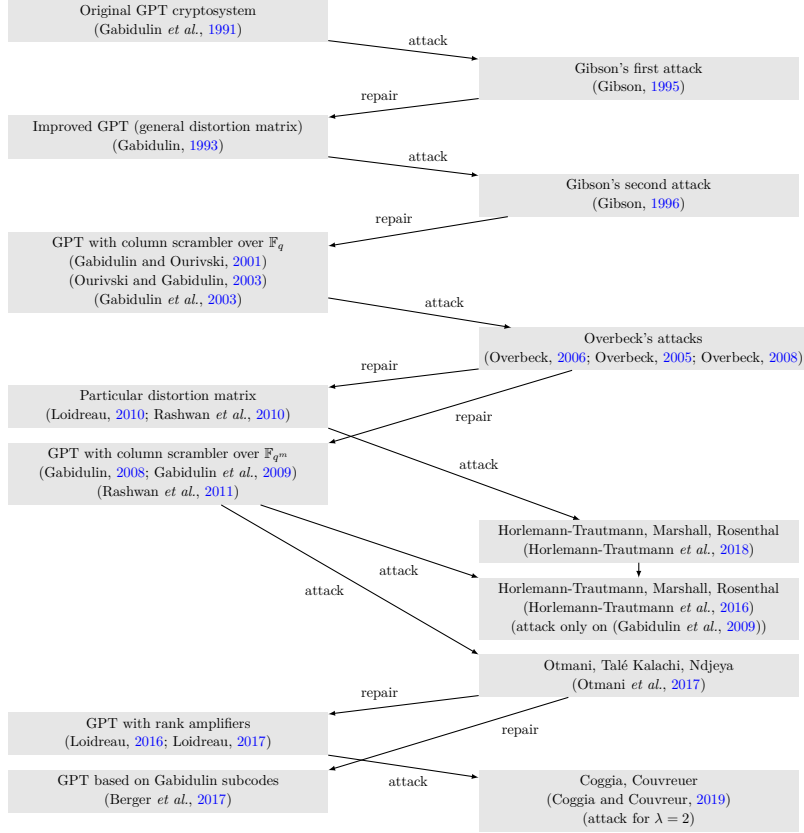
\begin{figure}[htb]
\begin{center}
\scalebox{0.5}{
\begin{tikzpicture}[node distance = 0.4cm and 4cm]
		\tikzset{systemnode/.style={rectangle, fill=black!10, minimum height=0.9cm, minimum width=8.5cm}}
		\tikzset{attacknode/.style={rectangle, fill=black!10, minimum height=0.9cm, minimum width=8.5cm}}
		\tikzset{attackarrow/.style={->, >= latex, thick}}
		\tikzset{repairarrow/.style={->, >= latex, thick}}

		\node[systemnode, align=center] (GPT91) at (0,0) {Original GPT cryptosystem \\ \citep{gabidulin1991ideals}};

		\node[attacknode, align=center, below right= of GPT91] (Gib95) {Gibson's first attack \\ \citep{gibson1995severely}};

		\node[systemnode, align=center,below left= of Gib95] (Gab93) {Improved GPT (general distortion matrix)\\ \citep{gabidulin1993linear}};

		\node[attacknode, align=center, below right= of Gab93] (Gib96) {Gibson's second attack \\ \citep{gibson1996security}};

		\node[systemnode, align=center, below left= of Gib96] (GOHA) {GPT with column scrambler over $\Fq$ \\ \citep{gabidulin2001modified}\\\citep{ourivski2003column}\\ \citep{gabidulin2003reducible}};

		\node[attacknode, align=center, below right= of GOHA] (Overbeck) {Overbeck's attacks \\ \citep{Overbeck-ExtendingGibsonAttack,Overbeck-StructuralAttackGPT,overbeck2008structural}};

		\node[systemnode, align=center, below left= of Overbeck] (LRGH) {Particular distortion matrix  \\ \citep{loidreau2010designing,rashwan2010smart}};

		\node[systemnode, align=center, below= of LRGH.south east, anchor=north east] (GRH) {GPT with column scrambler over $\Fqm$  \\ \citep{gabidulin2008attacks,gabidulin2009improving}\\ \citep{rashwan2011security}};

		\node[attacknode, align=center, below right= of GRH] (HMR18) {Horlemann-Trautmann, Marshall, Rosenthal \\ \citep{horlemann2018extension}};

		\node[attacknode, align=center, below = of HMR18.south west, anchor=north west] (HMR16) {Horlemann-Trautmann, Marshall, Rosenthal \\ \citep{horlemann2016considerations} \\ (attack only on~\citep{gabidulin2009improving})};

		\node[attacknode, align=center, below = of HMR16.south west, anchor=north west] (OKN) {Otmani, Tal\'e Kalachi, Ndjeya \\ \citep{Otmani-CryptanalysisRankMetric}};

		\node[systemnode, align=center, below left= of OKN] (Loi) {GPT with rank amplifiers \\ \citep{Loidreau-GPT-ACCT2016,Loidreau2017-NewRankMetricBased}};

		\node[attacknode, align=center, below right= of Loi] (CC) {Coggia, Couvreuer \\ \citep{coggia2019security} \\(attack for $\lambda=2$)};

		\node[systemnode, align=center, below= of Loi.south east, anchor=north east] (BGR) {GPT based on Gabidulin subcodes  \\ \citep{berger2017gabidulin}};

		\draw[attackarrow] (GPT91) to node[above right] {attack} (Gib95);
		\draw[repairarrow] (Gib95) to node[above left] {repair} (Gab93);
		\draw[attackarrow] (Gab93) to node[above right] {attack} (Gib96);
		\draw[repairarrow] (Gib96) to node[above left] {repair} (GOHA);
		\draw[attackarrow] (GOHA) to node[above right] {attack} (Overbeck);
		\draw[repairarrow] (Overbeck) to node[above left] {repair} (LRGH);
		\draw[repairarrow] (Overbeck) to node[below right] {repair} (GRH.north east);
		\draw[attackarrow] (LRGH.south east) to node[above right] {attack} (HMR18);
		\draw[attackarrow] (GRH) to node[below left, pos=0.75] {attack} (HMR16.north west);
		\draw[attackarrow] (HMR18) to node[right] {} (HMR16);
		\draw[attackarrow] (GRH) to node[below left] {attack} (OKN.north west);
		\draw[repairarrow] (OKN) to node[above left] {repair} (Loi);
		\draw[attackarrow] (Loi.south east) to node[below left, pos=0.7] {attack} (CC);
		\draw[repairarrow] (OKN) to node[below right, pos=0.3] {repair} (BGR.north east);
	\end{tikzpicture}}
	\caption{Overview of the history of GPT variants and attacks.}
	\label{fig:overviewGTPattacks}
\end{center}

\end{figure}

\subsubsection{Original GPT Cryptosystem}
The original GPT cryptosystem~\citep{gabidulin1991ideals} is a rank-metric variant of the McEliece cryptosystem that is based on Gabidulin codes.
In the GPT cryptosystem~\citep{gabidulin1991ideals}, the notion of a column scrambling matrix {(over $\Fq$)} does \emph{not} increase the security level of the system since the Moore-structure of the Gabidulin code generator matrix is preserved. 
Hence, a different distortion transformation using low-rank distortion matrix $\vec{X}$ is used to hide the inherent structure of the Gabidulin code in the public code described by $\pk$.
Detailed descriptions of the key generation, encryption and decryption of the original GPT cryptosystem are provided in Algorithms \ref{alg:KeyGenGPT91} -- \ref{alg:DecryptGPT91}.

\printalgoIEEE{
\DontPrintSemicolon
\caption{$\keyGen{\cdot}$~\citep{gabidulin1991ideals}}
\label{alg:KeyGenGPT91}
\KwIn{Parameters $\param = \left\{q,m,n,k,\tdist<\lfloor\frac{n-k}{2}\rfloor,\tpub=\lfloor\frac{n-k}{2}\rfloor-\tdist\right\}$}
\KwOut{Secret key $\sk$, public key $\pk$}
$\vec{g}\assignRand\Fqm^n:\rank_q(\vec{g})=n$ \\
$\vec{S}\assignRand\GL{k}{\Fqm}$ \\
$\vec{a}\assignRand\Fqm^k$, $\vec{b}\assignRand\Fqm^n:\rank_q(\vec{b})\leq\tdist$ \\
$\vec{G}_\mycode{C}=\Mooremat{k}{q}{\vec{g}}$\footnote{$\Mooremat{s}{q}{\vec{a}} \in \Fqm^{s \times n}$ the $s \times n$ Moore matrix for a vector $\vec{a} = (a_1,a_2,\dots,a_n) \in \Fqm^n$.} \\
$\vec{X}\assignDet\vec{a}^\top\vec{b}$ where $\rank_q(\vec{X})\leq \tdist$ \\
$\distTrans{\vec{G}_\mycode{C}}=\vec{S}\left(\vec{G}_{\mycode{C}}+\vec{X}\right)$ \\
\Return{$\sk\assignDet(\vec{S},\vec{G}_\mycode{C})$, $\pk\assignDet(\distTrans{\vec{G}_\mycode{C}},\tpub,\tdist)$}
}

\printalgoIEEE{
\DontPrintSemicolon
\caption{$\enc{\cdot}$~\citep{gabidulin1991ideals}}
\label{alg:EncryptGPT91}
\KwIn{Plaintext $\plainText\in\Fqm^k$, public key $\pk=(\distTrans{\vec{G}_\mycode{C}},\tpub,\tdist)$}
\KwOut{Ciphertext $\cipherText\in\Fqm^n$}
$\vec{e}\assignRand\Fqm^n:\rank_q(\vec{e})=\tpub=\lfloor\frac{n-k}{2}\rfloor-\tdist$ \\
$\cipherText\assignDet\plainText\distTrans{\vec{G}_\mycode{C}}+\vec{e}=\plainText\vec{S}\vec{G}_{\mycode{C}}+\plainText\vec{S}\vec{X}+\vec{e}$ \\
\Return{$\cipherText$}
}

\printalgoIEEE{
\DontPrintSemicolon
\caption{$\dec{\cdot}$~\citep{gabidulin1991ideals}}
\label{alg:DecryptGPT91}
\KwIn{Ciphertext $\vec{c}\in\Fqm^n$, secret key $\sk = (\vec{S},\vec{G}_\mycode{C})$}
\KwOut{Plaintext $\plainText\in\Fqm^k$}
$\cipherText$ is a codeword of the $(n,k)$ Gabidulin code $\mycode{C}$ described by $\vec{G}_\mycode{C}$ that is corrupted by an error $\tilde{\vec{e}}=\plainText\vec{S}\vec{X}+\vec{e}$ with $\rank_q(\tilde{\vec{e}})\leq \lfloor\frac{n-k}{2}\rfloor$ that can be corrected by a secret efficient decoder $\decode{\mycode{C}}{\cdot}$ for $\mycode{C}$ \\
$\tilde{\plainText}\assignDet\decode{\mycode{C}}{\cipherText}$ \\
\Return{$\plainText\assignDet\tilde{\plainText}\vec{S}^{-1}$}
}

 \subsubsection{Gibson's Attack on the Original GPT Cryptosystem}%
 \citet{gibson1995severely} proposed an attack on the original GPT cryptosystem that exploits the Moore structure of the generator matrix of Gabidulin codes to obtain an alternative secret key from the public key
 \begin{equation}\label{eq:pubKeyGPT}
	\pk=\vec{S}\left(\vec{G}+\vec{a}^{\top}\vec{b}\right)
 \end{equation}
 where $\vec{S}\in\Fqm^{k\times k}$ is a full-rank matrix, $\vec{a}\in\Fqm^{1\times k}$ and $\vec{b}\in\Fqm^{1\times n}$ with $\rank_q(\vec{b})\leq\tdist$.
 In particular, Gibson's attack computes matrices $\tilde{\vec{G}}$, $\tilde{\vec{S}}$ and $\tilde{\vec{X}}$ from~\eqref{eq:pubKeyGPT} such that $\tilde{\vec{G}}$ is a generator matrix of a Gabidulin code $\Gabcode{n}{k}$, $\pk=\tilde{\vec{S}}(\tilde{\vec{G}}+\tilde{\vec{X}})$ and $\tilde{\vec{X}}\in\Fqm^{n\times n}$ with $\rank_q(\tilde{\vec{X}})\leq\tdist$.
 In the following we give a brief description of the idea behind Gibson's attack.

 There exists a matrix $\vec{T}\in\Fqm^{k\times k}$ such that
 \begin{equation}\label{eq:matT}
	\vec{T}\left(\vec{I}_k \ | \ \vec{X}\right)=\vec{G}+\vec{a}^{\top}\vec{b}.
 \end{equation}
 Partition the matrix $\vec{G}$ and the vectors $\vec{g}$ and $\vec{b}$ into the first $k$ and remaining $n-k$ columns, i.~e.,
 \begin{equation*}
  \vec{G}=\left(\vec{P} \ | \ \vec{Q}\right), \quad \vec{g}=\left(\vec{p} \ | \ \vec{q}\right), \quad \vec{b}=\left(\vec{b}_1 \ | \ \vec{b}_2\right)
 \end{equation*}
 with $\vec{P}\in\Fqm^{k\times k}$, $\vec{Q}\in\Fqm^{k\times (n-k)}$, $\vec{p},\vec{b}_1\in\Fqm^{k}$ and $\vec{q},\vec{b}_2\in\Fqm^{n-k}$.
 Then we can write~\eqref{eq:matT} as
 \begin{align*}
  \vec{T}=\vec{P}+\vec{a}^{\top}\vec{b}_1
  \qquad\text{and}\quad
  \vec{T}\vec{X}=\vec{Q}+\vec{a}^{\top}\vec{b}_2.
 \end{align*}
 and get that
 \begin{equation}\label{eq:solvePX}
  \vec{P}\vec{X}=\vec{Q}+\vec{a}^{\top}\left(\vec{b}_2 -\vec{b}_1\vec{X}\right).
 \end{equation}
 Define the reversed vectors
 \begin{align*}
	\bar{\vec{p}}=\left(p_1^{[k-1]},p_2^{[k-1]},\dots,p_k^{[k-1]}\right),\\ \bar{\vec{q}}=\left(q_1^{[k-1]},q_2^{[k-1]},\dots,q_{n-k}^{[k-1]}\right),\\
	\bar{\vec{a}}=\left(a_k, a_{k-1}^{[1]}, a_{k-2}^{[2]},\dots, a_1^{[k-1]}\right),
 \end{align*}
 the all-one vector $\1\coloneqq\left(1,\dots,1\right)\in\Fqm^{k}$ and the matrices
 \begin{equation*}
	\vec{F}_j\coloneqq
	\begin{pmatrix}
	 f_j & & & &
	 \\
	  & f_j^{[1]}
	 \\
	  & & \ddots
	  \\
	  & & & f_j^{[k-1]}
	\end{pmatrix}
	\quad\text{and}\quad
	\vec{X}_j\coloneqq
	\begin{pmatrix}
	 \vec{x}_j & \vec{x}_j^{[1]} & \dots & \vec{x}_j^{[k-1]}
	\end{pmatrix}
 \end{equation*}
 where $\vec{x}_j$ denotes the $j$-th column of $\vec{X}$.
 Then we can rewrite~\eqref{eq:solvePX} as the reduced linear system of equations
 \begin{equation}\label{eq:GibsonAttackEq}
	\bar{\vec{p}}\vec{X}_j=\bar{q}_j\1+\bar{\vec{a}}\vec{F}_j,\quad\forall j=1,\dots, n-k.
 \end{equation}
 Gibson showed, that~\eqref{eq:GibsonAttackEq} can be solved in ${O}(m^3q^m)$ operations in $\Fqm$ and thus breaks the original GPT cryptosystem for the parameters proposed in~\citep{gabidulin1991ideals}.
 Although Gibson's attack breaks the original GPT cryptosystem for practical parameters ($n\leq 30$), the complexity of the attack is \emph{exponential} in the length of the code.

 \subsubsection{Modified GPT Variant by Gabidulin} %
 In order to prevent the attack by~\citet{gibson1995severely}, \citet{gabidulin1993linear} presented a modified variant of the GPT cryptosystem which uses a more general distortion matrix $\vec{X}$.
 The modified key generation is described in Algorithm~\ref{alg:KeyGenGPT93}.

 \printalgoIEEE{
\DontPrintSemicolon
\caption{$\keyGen{\cdot}$~\citep{gabidulin1993linear}}
\label{alg:KeyGenGPT93}
\KwIn{Parameters $\param = \left\{q,m,n,k,\tdist<\lfloor\frac{n-k}{2}\rfloor,\tpub=\lfloor\frac{n-k}{2}\rfloor-\tdist\right\}$}
\KwOut{Secret key $\sk$, public key $\pk$}
$\vec{g}\assignRand\Fqm^n:\rank_q(\vec{g})=n$ \\
$\vec{S}\assignRand\GL{k}{\Fqm}$ \\
$\vec{A}\assignRand\Fqm^{k\times\tdist}$, $\vec{B}\assignRand\Fqm^{\tdist\times n}:\rank_q(\vec{B})\leq\tdist$ \\
$\vec{G}_\mycode{C}=\Mooremat{k}{q}{\vec{g}}$ \\
$\vec{X}\assignDet\vec{A}\vec{B}$ where $\rank_q(\vec{X})\leq \tdist$ \\
$\distTrans{\vec{G}_\mycode{C}}=\vec{S}\left(\vec{G}_{\mycode{C}}+\vec{X}\right)$ \\
\Return{$\sk\assignDet(\vec{S},\vec{G}_\mycode{C})$, $\pk\assignDet(\distTrans{\vec{G}_\mycode{C}},\tpub,\tdist)$}
}

 \subsubsection{Gibson's Attack on the Modified GPT Cryptosystem} %
 \citep{gibson1996security} presented a second attack that breaks the modified GPT cryptosystem.
 Since the attack is based on the ideas of the initial attack~\citep{gibson1995severely}, a detailed description is omitted.

 \subsubsection{Modification with Right Scrambler over $\Fq$} %
 In order to prevent Gibson's attacks, new methods~\citep{gabidulin2001modified,ourivski2003column} were proposed to repair the GPT cryptosystem.
 In particular, GPT variants with a combination of a distortion matrix together with a column scrambling matrix $\vec{P}$ over $\Fq$ provided resilience against Gibson's attacks.
 The key generation, encryption and decryption are described in Algorithm~\ref{alg:KeyGenGO01} -- \ref{alg:DecryptGO01} respectively.

\printalgoIEEE{
\DontPrintSemicolon
\caption{$\keyGen{\cdot}$~\citep{gabidulin2001modified,ourivski2003column}}
\label{alg:KeyGenGO01}
\KwIn{Parameters $\param = \left\{q,m,n,k,r,\tdist<\lfloor\frac{n-k}{2}\rfloor,\tpub=\lfloor\frac{n-k}{2}\rfloor\tdist\right\}$}
\KwOut{Secret key $\sk$, public key $\pk$}
$\vec{g}\assignRand\Fqm^n:\rank_q(\vec{g})=n$ \\
$\vec{S}\assignRand\GL{k}{\Fqm}$ \\
$\vec{P}\assignRand\GL{n+r}{\Fq}$ \\
$\vec{X}_1\assignRand\Fqm^{k\times r}$, $\vec{X}_2\assignRand\Fqm^{k\times n}:\rank_q(\vec{X}_2)=\tdist$ \\
$\vec{P},\vec{X}_1,\vec{X}_2$ are chosen to have certain properties (see~\cite{gabidulin2001modified,ourivski2003column}) \\
$\vec{G}_\mycode{C}=\Mooremat{k}{q}{\vec{g}}$ \\
$\distTrans{\vec{G}_\mycode{C}}=\vec{S}\left(\left(\vec{0} \ | \ \vec{G}_{\mycode{C}}\right)+\left(\vec{X}_1 \ | \ \vec{X}_2\right)\right)\vec{P}=\vec{S}\left(\vec{X}_1 \ | \ \vec{G}_{\mycode{C}}+\vec{X}_2\right)\vec{P}$ with $\vec{0}\in\Fq^{k\times r}$ \\

\Return{$\sk\assignDet(\vec{S},\vec{P},\vec{G}_\mycode{C})$, $\pk\assignDet(\distTrans{\vec{G}_\mycode{C}},\tpub,\tdist,r)$}
}

\printalgoIEEE{
\DontPrintSemicolon
\caption{$\enc{\cdot}$~\citep{gabidulin2001modified,ourivski2003column}}
\label{alg:EncryptGO01}
\KwIn{Plaintext $\plainText\in\Fqm^k$, public key $\pk=(\distTrans{\vec{G}_\mycode{C}},\tpub,\tdist,r)$}
\KwOut{Ciphertext $\cipherText\in\Fqm^{n+r}$}
$\vec{e}\assignRand\Fqm^{n+1}:\rank_q(\vec{e})=\tpub=\lfloor\frac{n-k}{2}\rfloor-\tdist$ \\
$\cipherText\gets\plainText\distTrans{\vec{G}_\mycode{C}}+\vec{e}=\plainText\vec{S}\left(\left(\vec{0} \ | \ \vec{G}_{\mycode{C}}\right)+\left(\vec{X}_1 \ | \ \vec{X}_2\right)\right)\vec{P}+\vec{e}$\\
\Return{$\cipherText$}
}

\printalgoIEEE{
\DontPrintSemicolon
\caption{$\dec{\cdot}$~\citep{gabidulin2001modified,ourivski2003column}}
\label{alg:DecryptGO01}
\KwIn{Ciphertext $\vec{c}\in\Fqm^{n+r}$, secret key $\sk=(\vec{S},\vec{P},\vec{G}_\mycode{C})$}
\KwOut{Plaintext $\plainText\in\Fqm^k$}
Compute $\cipherText'\coloneqq\cipherText\vec{P}^{-1}=\plainText\vec{S}\left(\left(\vec{0} \ | \ \vec{G}_{\mycode{C}}\right)+\left(\vec{X}_1 \ | \ \vec{X}_2\right)\right)+\vec{e}\vec{P}^{-1}$ \\
Since $\vec{P}\in\Fq^{(r+n)\times(r+n)}$ we have that $\vec{e}'\coloneqq\vec{e}\vec{P}^{-1}$ satisfies $\rank_q(\vec{e}')\leq \tpub$ \\
From $\cipherText'$ extract the subvector $$\cipherText''=(c'_{r+1},c'_{r+2}\dots,c'_{r+n})=\plainText\vec{S}\vec{G}_{\mycode{C}}+\plainText\vec{X}_2+\vec{e}''$$ where $\vec{e}''$ corresponds to the last $n$ positions of $\vec{e}'=\vec{e}\vec{P}^{-1}$ \\
Since $\rank_q(\vec{e}'')\leq\rank_q(\vec{e}')\leq\tpub$ and $\rank_q(\plainText\vec{X}_2)\leq\tdist$ we have that $\vec{c}''$ is a codeword of the Gabidulin code $\mycode{C}$ that is corrupted by an error $\plainText\vec{X}_2+\vec{e}''$ of rank at most $\tdist+\tpub\leq\lfloor\frac{n-k}{2}\rfloor$ that can be corrected by a secret efficient decoder for $\mycode{C}$ \\
$\tilde{\plainText}\assignDet\decode{\mycode{C}}{\cipherText}$ \\
\Return{$\plainText\assignDet\tilde{\plainText}\vec{S}^{-1}$}
}

\subsubsection{Structural Attacks} %
After the series of Gibson's attacks and the subsequent repairs, \citep{Overbeck-StructuralAttackGPT} presented structural key-recovery attack on the original cryptosystem.
Unlike the attacks by Gibson, Overbeck's attacks run in polynomial time and thus break the GPT cryptosystem for all parameters.
The attack on the original GPT cryptosystem was extended in~\citep{Overbeck-ExtendingGibsonAttack} to break the GPT variants using a column scrambling matrix~\citep{ourivski2003column}.
Another variant of Overbeck's attacks~\citep{overbeck2008structural} combines methods from~\citep{Overbeck-StructuralAttackGPT,Overbeck-ExtendingGibsonAttack} in order to cryptanalyze~\citep{gabidulin2003reducible}.
In the following we describe the main idea of Overbeck's attacks.

For an $[n,k,d]_q^{\sfR}$ linear rank-metric code $\mycode{C}$
with generator matrix $\vec{G}$, we define $\mycode{C}^{q^i}$ to be the code obtained by taking each codeword in $\mycode{C}$ to the element-wise power $q^i$.
Note, that the code $\mycode{C}^{q^i}$ is generated by $\vec{G}^{q^i}$, where $\vec{G}^{q^i}$ is the matrix obtained by taking each element in $\vec{G}$ to the power of $q^i$.

\begin{definition}[$q$-Sum]
 Let $\mycode{C}$ be an $[n,k,d]_q^{\sfR}$ rank-metric code over $\Fqm$ and let $i\in\mathbb{N}_0$.
 Then the $i$-th $q$-sum of $\mycode{C}$ is defined as
 \begin{equation}
  \Lambda_i(\mycode{C})=\mycode{C}+\mycode{C}^{q}+\dots+\mycode{C}^{q^i}.
 \end{equation}
\end{definition}

For a random code $\mycode{C}$ we have $\dim{\Lambda_i(\mycode{C})}=\min\{n,ik\}$ with high probability.
However, if $\mycode{C}$ is a Gabidulin code $\Gab{n,k}$, we have that $\dim{\mycode{C}}=\min\{n,k+i\}$, which is significantly smaller than the dimension of the $q$-sum of a random code.
Hence, a Gabidulin code can be distinguished from a random code by checking the dimension of the $q$-sum.

Based on Overbeck's observation, one can distinguish a Gabidulin code from a random code by applying the $q$-sum to the public key and therefore derive an efficient decoder for the underlying decoder.

\subsubsection{Variants using Column Scramblers over $\Fqm$}%

The series of Overbeck's attacks broke all GPT variants for most practical parameters.
In order to prevent Overbeck's attack, GPT variants with column scrambling matrices over $\Fqm$ were proposed.
The first GPT variant using a column scrambler over $\Fqm$ was proposed by~\citet{gabidulin2008attacks}.
The GPT variants by~\citet{gabidulin2009improving,rashwan2010smart} are also based on this general idea.
In~\citep{Otmani-CryptanalysisRankMetric} it was shown that~\citep{gabidulin2009improving,rashwan2010smart} are special cases of~\citep{gabidulin2008attacks}.
Hence, we only describe the general system from~\citep{gabidulin2008attacks} in Algorithm~\ref{alg:KeyGenGab08} -- \ref{alg:DecryptGab08}.

\printalgoIEEE{
\DontPrintSemicolon
\caption{$\keyGen{\cdot}$~\citep{gabidulin2008attacks}}
\label{alg:KeyGenGab08}
\KwIn{Parameters $\param = \left\{q,m,n,k,r,\tpub,\tdist<\lfloor\frac{n-k}{2}\rfloor,t_2\right\}$}
\KwOut{Secret key $\sk$, public key $\pk$}
$\vec{g}\assignRand\Fqm^n:\rank_q(\vec{g})=n$ \\
$\vec{S}\assignRand\GL{k}{\Fqm}$ \\
$\vec{P}\assignRand\GL{n+\tdist}{\Fqm}$ such that
\begin{equation*}
      \vec{P}^{-1}=
      \begin{pmatrix}
       \vec{P}_{11} & \vec{P}_{12}
       \\
       \vec{P}_{21} & \vec{P}_{22}
      \end{pmatrix}
      \ \text{with}\
      \begin{array}{ll}
       \vec{P}_{11}\in\Fqm^{\tdist\times\tdist}, & \vec{P}_{12}\in\Fqm^{\tdist\times n}%
       \\
       \vec{P}_{21}\in\Fqm^{n\times\tdist}, & \vec{P}_{12}\in\Fq^{n\times n}
      \end{array}
     \end{equation*} and $\rank_q(\vec{P}_{12})<t_2$\\
$\vec{X}\assignRand\Fqm^{k\times \tdist}:\rank_q(\vec{X})=\tdist$\\
$\vec{G}_\mycode{C}=\Mooremat{k}{q}{\vec{g}}$ \\
$\distTrans{\vec{G}_\mycode{C}}=\vec{S}\left(\vec{X} \ | \ \vec{G}_{\mycode{C}}\right)\vec{P}$ \\

\Return{$\sk\assignDet(\vec{S},\vec{P},\vec{G}_\mycode{C})$, $\pk\assignDet(\distTrans{\vec{G}_\mycode{C}},\tpub,\tdist,t_2)$}
}

\printalgoIEEE{
\DontPrintSemicolon
\caption{$\enc{\cdot}$~\citep{gabidulin2008attacks}}
\label{alg:EncryptGab08}
\KwIn{Plaintext $\plainText\in\Fqm^k$, public key $\pk=(\distTrans{\vec{G}_\mycode{C}},\tpub,\tdist,t_2)$}
\KwOut{Ciphertext $\cipherText\in\Fqm^{n+\tdist}$}
$\vec{e}\assignRand\Fqm^n:\rank_q(\vec{e})=\tpub=\lfloor\frac{n-k}{2}\rfloor$ and $\vec{e}=\left(\vec{e}_1 \ | \ \vec{e}_2\right)$ where $\vec{e}_1\in\Fqm^{\tdist}$ and $\vec{e}_2\in\Fqm^n$ with $\rank_q(\vec{e}_2)=\tpub-t_2$  \\
$\cipherText\gets\plainText\distTrans{\vec{G}_\mycode{C}}+\vec{e}=\plainText\vec{S}\left(\vec{X} \ | \ \vec{G}_{\mycode{C}}\right)\vec{P}+\vec{e}$\\
\Return{$\cipherText$}
}

\printalgoIEEE{
\DontPrintSemicolon
\caption{$\dec{\cdot}$~\citep{gabidulin2008attacks}}
\label{alg:DecryptGab08}
\KwIn{Ciphertext $\vec{c}\in\Fqm^{n+\tdist}$, secret key $\sk=(\vec{S},\vec{P},\vec{G}_\mycode{C})$}
\KwOut{Plaintext $\plainText\in\Fqm^k$}

Compute $\cipherText'\coloneqq\cipherText\vec{P}^{-1}=\plainText\vec{S}\left(\vec{X} \ | \ \vec{G}_{\mycode{C}}\right)+\vec{e}\vec{P}^{-1}$

From $\cipherText'$ extract the subvector $$\cipherText''=(c'_{\tdist+1},c'_{\tdist+2}\dots,c'_{\tdist+n})=\plainText\vec{S}\vec{G}_{\mycode{C}}+\vec{e}''$$ where $\vec{e}''$ corresponds to the last $n$ positions of $\vec{e}\vec{P}^{-1}$ given by
    \begin{equation*}
     \vec{e}'' = \vec{e}_1\vec{P}_{12}+\vec{e}_2\vec{P}_{22}.
    \end{equation*}
    Since the rank of $\vec{P}_{12}$ satisfies $\rank_q(\vec{P}_{12})<t_2$ and $\vec{P}_{22}$ has its elements in $\Fq$ we have that
    \begin{equation*}
      \rank_q(\vec{e}'')\leq\rank_q(\vec{P}_{12})+\rank_q(\vec{e}_2)=t_2+\tpub-t_2=\tpub.
    \end{equation*}\\
Hence, $\vec{c}''$ is a codeword of the Gabidulin code $\mycode{C}$ that is corrupted by an error $\vec{e}''$ of rank at most $\tpub\leq\lfloor\frac{n-k}{2}\rfloor$ that can be corrected by a secret efficient decoder for $\mycode{C}$ \\
$\tilde{\plainText}=\plainText\vec{S}\gets\decode{\mycode{C}}{\cipherText}$ \\
\Return{$\plainText\gets\tilde{\plainText}\vec{S}^{-1}$} \\
}

\subsubsection{Variants using Particular Distortion Matrices}%

Different GPT variants that are resistant against Overbeck's attacks were proposed by~\citet{loidreau2010designing,rashwan2010smart}.
Rather than relying on a column scrambling matrix over $\Fqm$, the variants use particular distortion matrices to prevent Overbeck's structural attacks.
The variants can be described in Algorithm~\ref{alg:KeyGenLoi10} -- \ref{alg:DecryptLoi10}.

\printalgoIEEE{
\DontPrintSemicolon
\caption{$\keyGen{\cdot}$~\citep{loidreau2010designing,rashwan2010smart}}
\label{alg:KeyGenLoi10}
\KwIn{Parameters $\param = \left\{q,m,n,k,\tpub\leq\lfloor\frac{n-k}{2}\rfloor,\tdist\right\}$}
\KwOut{Secret key $\sk$, public key $\pk$}
$\vec{g}\assignRand\Fqm^n:\rank_q(\vec{g})=n$ \\
$\vec{S}\assignRand\GL{k}{\Fqm}$ \\
$\vec{P}\assignRand\GL{n+\tdist}{\Fq}$ \\
$\vec{X}\assignRand\Fqm^{k\times \tdist}:\rank_q(\vec{X})=\tdist$\\
$\vec{G}_\mycode{C}=\Mooremat{k}{q}{\vec{g}}$ \\
$\distTrans{\vec{G}_\mycode{C}}=\vec{S}\left(\vec{X} \ | \ \vec{G}_{\mycode{C}}\right)\vec{P}$ \\

\Return{$\sk\assignDet(\vec{S},\vec{G}_\mycode{C},\vec{P})$, $\pk\assignDet(\distTrans{\vec{G}_\mycode{C}},\tpub,\tdist)$}
}

If the right kernel of $\pk$ has dimension one, a decoder can be obtained from $\pk$ in polynomial time. This occurs with high probability if $\vec{X}$ is chosen uniformly at random. This implies the design criterion that $\vec{X}$ has to be chosen such that $\rank_q(\vec{X})\leq\frac{\tdist-\ell}{n-k}$ for some $\ell\geq1$ which implies $\tdist>(n-k)$ (see~\citep[Corollary~1]{loidreau2010designing}). However, this restriction increases the key size.

\printalgoIEEE{
\DontPrintSemicolon
\caption{$\enc{\cdot}$~\citep{loidreau2010designing,rashwan2010smart}}
\label{alg:EncryptLoi10}
\KwIn{Plaintext $\plainText\in\Fqm^k$, public key $\pk=(\distTrans{\vec{G}_\mycode{C}},\tpub,\tdist)$}
\KwOut{Ciphertext $\cipherText\in\Fqm^{n+\tdist}$}
$\vec{e}\assignRand\Fqm^{n+\tdist}:\rank_q(\vec{e})=\tpub=\lfloor\frac{n-k}{2}\rfloor$ \\
$\cipherText\gets\plainText\distTrans{\vec{G}_\mycode{C}}+\vec{e}=\plainText\vec{S}\left(\vec{X} \ | \ \vec{G}_{\mycode{C}}\right)\vec{P}+\vec{e}$\\
\Return{$\cipherText$}
}

\printalgoIEEE{
\DontPrintSemicolon
\caption{$\dec{\cdot}$~\citep{loidreau2010designing,rashwan2010smart}}
\label{alg:DecryptLoi10}
\KwIn{Ciphertext $\vec{c}\in\Fqm^{n+\tdist}$, secret key $\sk=(\vec{S},\vec{G}_\mycode{C},\vec{P})$}
\KwOut{Plaintext $\plainText\in\Fqm^k$}

Compute $\cipherText'\coloneqq\cipherText\vec{P}^{-1}=\plainText\vec{S}\left(\vec{X} \ | \ \vec{G}_{\mycode{C}}\right)+\vec{e}\vec{P}^{-1}$ \\
From $\cipherText'$ extract the subvector $$\cipherText''=(c'_{\tdist+1},c'_{\tdist+2}\dots,c'_{\tdist+n})=\plainText\vec{S}\vec{G}_{\mycode{C}}+\vec{e}''$$ where $\vec{e}''$ corresponds to the last $n$ positions of $\vec{e}\vec{P}^{-1}$.
Notice, that $\rank_q(\vec{e}')\leq\tpub$ since $\vec{P}$ is over $\Fq$. \\
Since $\rank_q(\vec{e}'')\leq\rank_q(\vec{e}')\leq\tpub$ we have that $\vec{c}''$ is a codeword of the Gabidulin code $\mycode{C}$ that is corrupted by an error $\vec{e}''$ of rank at most $\tpub\leq\lfloor\frac{n-k}{2}\rfloor$ that can be corrected by a secret efficient decoder for $\mycode{C}$ \\

$\tilde{\plainText}=\plainText\vec{S}\gets\decode{\mycode{C}}{\cipherText}$ \\

\Return{$\plainText\gets\tilde{\plainText}\vec{S}^{-1}$} \\
}

The approach by~\citet{rashwan2010smart} (also referred as the ``smart approach'') relies on the same structure as the system in~\citep{loidreau2010designing} with the difference that the increase of the key size in~\citep{loidreau2010designing} due to the rank restriction on $\vec{X}$ is avoided by imposing a structural restriction on $\vec{X}$.
In particular, $\vec{X}$ is constructed from a Moore matrix of rank $\ell$ and a non-Moore matrix of rank $\tdist-a$ to avoid Overbeck's attacks.

\subsubsection{Attack on Variant with Distortion Matrices}%

The above mentioned attempts to defend the GPT cryptosystem from Overbeck's attacks based on distortion matrices were broken by \citet{horlemann2018extension}. The main idea of the attack is based on recovering vectors of rank one from an extended public generator matrix that allow to recover the secret key.

\subsubsection{Attacks on Variant with Column Scramblers}%

The alternative GPT variants proposed to prevent Overbecks's attacks by using column scrambling matrices over $\Fqm$ were subject to several attacks.
First, the scheme by \citet{gabidulin2009improving} was attacked by \citet{horlemann2018extension} using ideas from~\citep{gaborit2015complexity}.
This attack was followed by a generalized Overbeck attack~\citep{Otmani-CryptanalysisRankMetric} that cryptanalyzes~\citep{gabidulin2008attacks,gabidulin2009improving,rashwan2010smart} which is more efficient than that in~\citep{horlemann2016considerations}.
The main result of~\citep{Otmani-CryptanalysisRankMetric} is that all GPT variants with a column scrambler over $\Fqm$ can be reduced to a cryptosystem with column scrambler over $\Fq$ with a slightly degraded Gabidulin code.
However, the error-correction capability of the degraded code is sufficient to recover the imposed errors for most parameters.

\subsubsection{GPT Variant with Rank Amplifiers}%

The GPT variants~\citep{Loidreau-GPT-ACCT2016, Loidreau2017-NewRankMetricBased} use particular column scrambling matrices over $\Fqm$, called \emph{rank amplifiers}. The variants are described in Algorithm~\ref{alg:KeyGenLoi16}--\ref{alg:DecryptLoi16}.
In particular, the coefficients of the inverse of the right scrambler are taken from some fixed-dimensional $\Fq$-linear subspace of $\Fqm$.
This idea is motivated by LRPC codes~\citep{Gaborit2013-LRPC} and the analogue cryptosystem in the Hamming metric proposed by \citet{Baldi2016enhanced}.

\printalgoIEEE{
\DontPrintSemicolon
\caption{$\keyGen{\cdot}$~\citep{Loidreau-GPT-ACCT2016, Loidreau2017-NewRankMetricBased}}
\label{alg:KeyGenLoi16}
\KwIn{Parameters $\param = \left\{q,m,n,k,\lambda,\tpub\leq\lfloor\frac{n-k}{2\lambda}\rfloor\right\}$}
\KwOut{Secret key $\sk$, public key $\pk$}
Random $\Fq$-linear subspace $\myspace{V}$ of dimension $\dim(\myspace{V})=\lambda$ \\
$\vec{g}\assignRand\Fqm^n:\rank_q(\vec{g})=n$ \\
$\vec{S}\assignRand\GL{k}{\Fqm}$ \\
$\vec{P}\assignRand\GL{n}{\Fqm}:p_{i,j}\in\myspace{V}$ for all $i,j\in[1,n]$ \\
$\vec{G}_\mycode{C}=\Mooremat{k}{q}{\vec{g}}$ \\
$\distTrans{\vec{G}_\mycode{C}}=\vec{S}\vec{G}_{\mycode{C}}\vec{P}^{-1}$ \\

\Return{$\sk\assignDet(\vec{S},\vec{G}_\mycode{C},\vec{P})$, $\pk\assignDet(\distTrans{\vec{G}_\mycode{C}},\tpub,\lambda)$}
}

\printalgoIEEE{
\DontPrintSemicolon
\caption{$\enc{\cdot}$~\citep{Loidreau-GPT-ACCT2016, Loidreau2017-NewRankMetricBased}}
\label{alg:EncryptLoi16}
\KwIn{Plaintext $\plainText\in\Fqm^k$, public key $\pk=(\distTrans{\vec{G}_\mycode{C}},\tpub,\lambda)$}
\KwOut{Ciphertext $\cipherText\in\Fqm^n$}
$\vec{e}\assignRand\Fqm^n:\rank_q(\vec{e})=\tpub=\tpub\leq\lfloor\frac{n-k}{2\lambda}\rfloor$ \\
$\cipherText\gets\plainText\distTrans{\vec{G}_\mycode{C}}+\vec{e}=\plainText\vec{S}\vec{G}_{\mycode{C}}\vec{P}^{-1}+\vec{e}$\\
\Return{$\cipherText$}
}

\printalgoIEEE{
\DontPrintSemicolon
\caption{$\dec{\cdot}$~\citep{Loidreau-GPT-ACCT2016, Loidreau2017-NewRankMetricBased}}
\label{alg:DecryptLoi16}
\KwIn{Ciphertext $\cipherText\in\Fqm^n$, secret key $\sk=(\vec{S},\vec{G}_\mycode{C},\vec{P})$}
\KwOut{Plaintext $\plainText\in\Fqm^k$}

Compute $\cipherText'\coloneqq\cipherText\vec{P}^{-1}=\plainText\vec{S}\vec{G}_{\mycode{C}}+\vec{e}\vec{P}$ \\
Since the entries of $\vec{P}$ span a $\lambda$-dimensional $\Fq$-linear subspace of $\Fqm$, we have that $\rank_q(\vec{e}\vec{P})\leq\lambda\rank_q(\vec{e})=\lambda\tpub$ (see~\citep[Proposition~1]{Loidreau-GPT-ACCT2016}). \\
Hence, $\vec{c}''$ is a codeword of the Gabidulin code $\mycode{C}$ that is corrupted by an error $\vec{e}\vec{P}$ of rank at most $\lambda\tpub\leq\lfloor\frac{n-k}{2}\rfloor$ that can be corrected by a secret efficient decoder for $\mycode{C}$ \\
$\tilde{\plainText}=\plainText\vec{S}\gets\decode{\mycode{C}}{\cipherText}$\\

\Return{$\plainText\gets\tilde{\plainText}\vec{S}^{-1}$} \\
}

\subsubsection{Attack for GPT with Rank Amplifiers ($\lambda=2$)}%

\citet{coggia2019security} proposed an attack on the rank-amplifier variant of the GPT cryptosystem~\citep{Loidreau-GPT-ACCT2016, Loidreau2017-NewRankMetricBased}.
The attack is successful for $\lambda=2$ and a public code with code rate $R_{\mathsf{pub}}\geq\frac{1}{2}$. In general, parameters where the public code has rate $R_{\mathsf{pub}}\geq 1-\frac{1}{\lambda}$ should be avoided.
However, there still are practical parameters for which the system~\citep{Loidreau-GPT-ACCT2016, Loidreau2017-NewRankMetricBased} cannot be broken by the attack by \citet{coggia2019security}.

\subsubsection{Variant using Gabidulin Matrix Codes}%

In~\citep{berger2017gabidulin} a cryptosystem based on matrix codes was proposed. These matrix codes are obtained from subcodes of binary images of Gabidulin codes to mask the structure of the underlying Gabidulin code in order to prevent Overbeck-like attacks.
This approach roughly gains a factor of 10 compared to the original McEliece cryptosystem and provides very compact private keys.
In Algorithms \ref{alg:KeyGenBGR17} -- \ref{alg:DecryptBGR17}, describe this cryptosystem.
  We give a high-level description of the $q$-ary image of a code, which is used in the key generation of this cryptosystem.
For any matrix $\A\in\Fq^{m \times n}$ let $\text{vec}(\A)\in\Fq^{mn}$ denote the row vector representation of $\A$ obtained by concatenating the transposed columns of $\A$.
Analogously, denote by $\text{mat}(\vec{a})$ the corresponding inverse mapping to obtain a matrix $\A\in\Fq^{m\times n}$ from a vector $\a\in\Fq^{mn}$.
Under a fixed basis of $\Fqm$ over $\Fq$, any code $\mycode{C}\subseteq{\Fqm^n}$ can be represented as a subset of $\Fq^{m \times n}$.
The $q$-ary image $\qImg{q}{\mycode{C}}\subseteq\Fq^{mn}$ of $\mycode{C}$ is then defined as the code obtained by applying the mapping $\text{vec}(\cdot)$ to each codeword of $\mycode{C}$ expanded over $\Fq$, which results in a code of length $mn$ and dimension $mk$ that can be represented by an $mk \times mn$ generator matrix over $\Fq$.

\printalgoIEEE{
\DontPrintSemicolon
\caption{$\keyGen{\cdot}$~\citep{berger2017gabidulin}}
\label{alg:KeyGenBGR17}
\KwIn{Parameters $\param = \left\{q,m,n,k,r,\tpub=\lfloor\frac{n-k}{2}\rfloor\right\}$}
\KwOut{Secret key $\sk$, public key $\pk$}
$\vec{g}\assignRand\Fqm^n:\rank_q(\vec{g})=n$ \\
$\vec{B}_0\assignRand\GL{m}{\Fq}$ \\
$\vec{B}_1\assignRand\GL{n}{\Fq}$ \\
$\vec{L}\assignRand\Fq^{s\times km-s}$ \\
$\vec{G}\assignDet$ generator matrix of $\qImg{q}{\Gabcode{\vec{g};n}{k}}$ (in $\Fq^{km\times mn}$)\\
$\widetilde{\vec{G}}=\vec{G}(\vec{B}_1\otimes\vec{B}_0^\top)$ with the corresponding generator matrix in systematic from $\widetilde{\vec{G}}_{\mathsf{sys}}=\left(\vec{I}_{km} \ | \ \vec{P}\right)$. Define $\widetilde{\vec{H}}=(-\vec{P}^\top \ | \ \vec{I}_{(n-k)m})$ \\
Define $\vec{U}=\left(\vec{L} \ | \ \vec{I}_s \ | \ \vec{0}_{s,(n-k)m}\right)$ \\
Define $\vec{H}_{\mathsf{pub}}= \left(\begin{array}{c}\vec{U} \\ \hline \widetilde{\vec{H}}\end{array}\right)$ and $\vec{G}_{\mathsf{pub}}\in\Fq^{(km-s)\times mn}$ s.t. $\vec{G}_{\mathsf{pub}}\vec{H}_{\mathsf{pub}}^\top=\vec{0}$.

\Return{$\sk\assignDet(\vec{B}_0,\vec{B}_1,\vec{g})$, $\pk\assignDet(\vec{G}_{\mathsf{pub}},\tpub)$}
}

\printalgoIEEE{
\DontPrintSemicolon
\caption{$\enc{\cdot}$~\citep{berger2017gabidulin}}
\label{alg:EncryptBGR17}
\KwIn{Plaintext $\plainText\in\Fq^{km-s}$, public key $\pk=(\vec{G}_{\mathsf{pub}},\tpub)$}
\KwOut{Ciphertext $\cipherText\in\Fq^{m\times n}$}
$\vec{E}\assignRand\Fq^{m\times n}:\rank_q(\vec{E})=\tpub$ \\
$\vec{C}\gets\text{mat}(\plainText\vec{G}_{\mathsf{pub}})+\vec{E}$\\
\Return{$\vec{C}$}
}

\printalgoIEEE{
\DontPrintSemicolon
\caption{$\dec{\cdot}$~\citep{berger2017gabidulin}}
\label{alg:DecryptBGR17}
\KwIn{Ciphertext $\cipherText$, secret key $\sk=(\vec{B}_0,\vec{B}_1)$}
\KwOut{Plaintext $\plainText$}
Compute $\tilde{\vec{C}}\assignDet\vec{B}_0^{-1}\vec{C}=\text{mat}(\vec{m}\vec{G}_{\mathsf{pub}})+\tilde{\vec{E}}$ where $\tilde{\vec{E}}=\vec{B}_0^{-1}\vec{E}$ with $\rank_q(\tilde{\vec{E}})=\tpub$ \\
Compute $\tilde{\vec{c}}$ from $\tilde{\vec{C}}$\\
Decode $\tilde{\vec{c}}$ and use $\vec{B}_1$ to obtain $\plainText$

\Return{$\plainText$}
}

\subsubsection{Twisted Gabidulin Codes in the GPT System}

\emph{Twisted Gabidulin codes} were applied to the GPT system by~\citet{puchinger2018twisted}. However, it turned out that there is a distinguisher when using such Twisted Gabidulin codes in the GPT system.
\begin{theorem}[$q$-Sum Dimension of Twisted Gabidulin Codes, \citep{puchinger2018twisted}]\label{thm:large_q-sum-dim_family}
	Let $n,k,\tVec,\hVec,\etaVec,\alphaVec$ be chosen as in Definition~\ref{def:tgab_definition} such that
	\begin{align*}
	&\Delta \coloneqq \tfrac{n-k-\ell}{\ell+1} \in \NN,\\
	&t_i \coloneqq (i+1)(\Delta+1), &&\forall \, i =1,\dots,\ell,\\
	&0 < h_1 < h_2 < \dots h_\ell < k-1 &&\text{s.t.} \notag\\
	&|h_{i+1}-h_i| > 1, &&\forall \, i = 1,\dots,\ell-1.
	\end{align*}
	For all $i \in \NN$, we then have
	\begin{equation*}
	\dim \Lambda_i(\Cmult) = \min\{k-1+(i+1)(\ell+1),n\}.
	\end{equation*}
\end{theorem}

In particular, for a small number of twists $\ell$, the dimension of the $q$-sum of Twisted Gabidulin codes is rather low compared to random codes of the same dimension (though larger than that of Gabidulin codes), which constitutes a distinguisher.

\citet{lavauzelle2020cryptanalysis} proposed a polynomial-time key-recovery attack on the Hamming-metric variant of this scheme. Furthermore, they showed that the presented attack can straightforwardly be applied to the variant based on twisted Gabidulin codes for certain parameters. However, this na\"ive adaption can be avoided by a careful choice of the parameters.

\subsubsection{Interleaving Loidreau's GPT System}

The ideas of using interleaved codes in the McEliece system \citep{elleuch2018public,holzbaur2019decoding} were combined with Loidreau’s GPT variant \citep{Loidreau2017-NewRankMetricBased} in \citep{renner2019interleavedLoi}. That means, $u$ parallel ciphertexts are considered where each is a codeword from Loidreau's Gabidulin code (i.e., a Gabidulin code scrambled with a matrix that contains elements from a subspace, \citep[see][]{Loidreau2017-NewRankMetricBased}) with rank multipliers plus a rank burst error (i.e., all the errors lie in a common row space). The dimension of the total row space of the parallel errors is restricted by $\left\lfloor\frac{u}{u+1}\frac{n-k}{2}\right\rfloor$.

In \citep{renner2019interleavedLoi}, it was shown that in principle, Loidreau’s system can be interleaved using classical decoders for interleaved Gabidulin codes. Similar to \citep{holzbaur2019decoding}, an attack based on an error code can be prevented by choosing the error matrix in a suitable way. The construction of (in this sense) secure errors requires rank-metric codes whose minimum distances are close to the Singleton bound. Gabidulin codes yield potentially insecure error patterns since the resulting error matrix can be distinguished from a random one. However, if the error matrix is drawn in a random way, it fulfills the requirements with high probability. For this choice, upper bounds on the decryption failure and secure parameter sets with potential key size reductions by approximately 15\% compared to Loidreau's system were provided.

\subsection{McEliece-Type Rank-Metric Cryptosystems based on QC-LRPC Codes}%

A GPT variant based on a randomized construction of rank-metric codes was proposed by \citet{Gaborit2013-LRPC}.
This random rank-metric coding approach uses low-rank parity-check (LRPC) codes in the GPT system in order to prevent the series of structural attacks based on the inherent algebraic structure of Gabidulin codes.
The NIST submission ROLLO~\citep{melchor2019rollo}, which is the merge of the initial round-1 submissions Rank-Ouroboros (formerly known as Ouroboros-R), LAKE and LOCKER, is a GPT variant based on LRPC codes.

\begin{definition}[Low-Rank Parity-Check Code]\label{def:LRPC}
 A $[\lambda;n,k]$ LRPC code of length $n$, dimension $k$, and rank $\lambda$ over $\Fqm$ is defined as a code with a parity-check matrix $\vec{H}\in\Fqm^{(n-k)\times n}$, where the $\Fq$-linear vector space
 \begin{equation}\label{eq:Hspace}
   \myspace{H}\coloneqq\spannedBy{\{h_{i,j}:i\in [1,n-k],j\in[1,n]\}}
 \end{equation}
 has dimension at most $\lambda$.
\end{definition}

LRPC codes can be constructed to be quasi-cylic codes. In this case, the parity-check matrix as well as the generator matrix consist of blocks of circulant matrices, that allow for a very compact representation, which is favorable in terms of the key size.
In particular, variants consisting of two circulant matrices (also called double-circulant (DC)) are of particular interest for cryptographic applications \citep{Gaborit2013-LRPC}.

In the following we give a high-level description of the decoding algorithm for LRPC codes from~\citep{Gaborit2013-LRPC}, where we define the \emph{rank support} as follows:
\begin{definition}[Rank Support]
	Let $\vec{x} = (x_1,x_2,\dots,x_n) \in \Fqm^n$. The support $\supp(\vec{x})$ is the $\Fq$-subspace of $\Fqm$ generated by the coordinates of $\vec{x}$:
	\begin{equation*}
	\supp(\vec{x}) = \Rowspace{x_1,\dots,x_n}.
	\end{equation*}
\end{definition}
Let
\begin{equation}
  \Phi=\{\phi_1,\phi_2,\dots,\phi_\lambda\}
  \quad\text{and}\quad
  \Gamma=\{\gamma_1,\gamma_2,\dots,\gamma_t\}
\end{equation}
be bases for $\myspace{H}$ and the error support $\myspace{E}\coloneqq\supp(\vec{e})$, respectively.
Then each element of the syndrome $\vec{s}=\vec{e}\vec{H}^{\top}$ can be written as
\begin{equation}
  s_i=\sum_{r=1}^{t}\sum_{\ell=1}^{\lambda}s_{i,\ell,r}\phi_\ell\gamma_r ,
\end{equation}
where
\begin{equation}\label{eq:recoverErrLoc}
  s_{i,\ell,r}=\sum_{j=1}^{n}h_{i,j,\ell}e_{j,r},\qquad \ell\in[1,\lambda],r\in[1,t],i\in[1,n-k].
\end{equation}
and $h_{i,j,\ell}$ is the expansion of $h_{i,j}$ over $\Fq$ with respect to $\Phi$.
Hence, each syndrome entry $s_i$ is an $\Fq$-linear combination of the elements
 \begin{equation}
   \{\phi_1\gamma_1,\phi_1\gamma_2,\dots,\phi_\lambda\gamma_t\}.
 \end{equation}
With high probability the elements above span the product space $\myspace{H}\myspace{E}\coloneqq\{ab:a\in\myspace{H},b\in\myspace{E}\}$.
Defining the syndrome space $\myspace{S}\coloneqq\supp(\vec{s})$ 
and $\myspace{S}_{\ell}\coloneqq\{\phi_\ell^{-1}\myspace{S}:\ell\in[1,\lambda]\}$ we may recover $\myspace{E}$ as
\begin{equation}
  \myspace{E}\subseteq \myspace{S}_{1}\cap\myspace{S}_{2}\cap\dots\cap\myspace{S}_{\ell}
\end{equation}
with high probability.
Once a basis of the product space is known the error locations can be recovered by solving the $\Fq$-linear system of equations~\eqref{eq:recoverErrLoc} that has a unique solution with high probability.
Knowing the error locations and the support $\myspace{E}$, the error vector $\vec{e}$ can be recovered. The detailed description of the cryptosystem is provided in Algorithms~\ref{alg:KeyGenGMRZ13}--\ref{alg:DecryptGMRZ13}.

\printalgoIEEE{
\DontPrintSemicolon
\caption{$\keyGen{\cdot}$~\citep{Gaborit2013-LRPC}}
\label{alg:KeyGenGMRZ13}
\KwIn{Parameters $\param = \left\{q,m,n,k,\lambda,\tpub\right\}$}
\KwOut{Secret key $\sk$, public key $\pk$}
Random $\Fq$-linear subspace $\myspace{H}$ of dimension $\dim(\myspace{V})=\lambda$ \\
$\vec{H}\assignRand\Fqm^{(n-k)\times n}: \rank_{q^m}(\vec{H})=n-k$ and $h_{i,j}\in\myspace{H}$ for all $i\in[1,n-k]$ and $j\in[1,n]$ \\
$\vec{S}\in\GL{k}{\Fqm}$ \\
$\vec{G}_\mycode{C}\assignDet\Fqm^{k\times n}:\rank(\vec{G})=k$ and $\vec{G}\vec{H}^{\top}=\vec{0}$ \\
$\distTrans{\vec{G}_\mycode{C}}=\vec{S}\vec{G}_{\mycode{C}}$ \\
\Return{$\sk\assignDet(\vec{H},\vec{S})$, $\pk\assignDet(\distTrans{\vec{G}_\mycode{C}},\tpub)$}
}

\printalgoIEEE{
\DontPrintSemicolon
\caption{$\enc{\cdot}$~\citep{Gaborit2013-LRPC}}
\label{alg:EncryptGMRZ13}
\KwIn{Plaintext $\plainText\in\Fqm^k$, public key $\pk=(\distTrans{\vec{G}_\mycode{C}},\tpub)$}
\KwOut{Ciphertext $\cipherText\in\Fqm^n$}
$\vec{e}\assignRand\Fqm^n:\rank_q(\vec{e})=\tpub$ \\
$\cipherText\gets\plainText\distTrans{\vec{G}_\mycode{C}}+\vec{e}=\plainText\vec{S}\vec{G}_{\mycode{C}}+\vec{e}$\\
\Return{$\cipherText$}
}

\printalgoIEEE{
\DontPrintSemicolon
\caption{$\dec{\cdot}$~\citep{Gaborit2013-LRPC}}
\label{alg:DecryptGMRZ13}
\KwIn{Ciphertext $\cipherText\in\Fqm^n$, secret key $\sk=(\vec{H},\vec{S})$}
\KwOut{Plaintext $\plainText\in\Fqm^k$}

Decode $\cipherText$ using $\vec{H}$: $\tilde{\plainText}\assignDet\decode{\mycode{C}}{\cipherText}$ \\
$\tilde{\plainText}=\plainText\vec{S}\gets\decode{\mycode{C}}{\cipherText}$ \\

\Return{$\plainText\gets\tilde{\plainText}\vec{S}^{-1}$} \\
}

\subsection{Attack on the QC-LRPC GPT Variant}%

An attack on the QC-LRPC cryptosystem that exploits the block-circulant structure of the parity-check matrix was presented in~\citep{hauteville2015new}.
As a possible repair, the application of so-called \emph{ideal} LRPC codes instead of QC-LRPC codes was suggested by \citet{hauteville2015new,aragon2019low}.
Ideal LRPC codes allow for the same compact code representation as QC-LRPC codes but are not vulnerable to the attack in~\citep{hauteville2015new}.
The current version of the NIST submission ROLLO is based on ideal LRPC codes~\citep{melchor2019rollo}.

In addition to the structural attack by~\citet{hauteville2015new} there is a reaction-based attack on the QC-LRPC GPT cryptosystem that exploits decoding failures that may occur during the decryption process to recover the secret key~\citep{aragon2019key}.
Let $\vec{G}_{LRPC}$ be a generator matrix of the code generated by a parity-check matrix $\vec{H}_{LRPC}$ of a $[\lambda;n,k]$ QC-LRPC code over $\Fqm$ that can correct $\tpub$ errors with high probability.
The main idea of this attack is to challenge the decoder with particularly chosen error patterns and observe if the decoder can decode or not.
The observation of the corresponding decoding failure rate can reveal the structure of the parity-check matrix $\vec{H}_{LRPC}$.
In particular, the failure event that the syndrome does not span the whole product space (which dominates the decoding failure probability for most practical parameters) is exploited.
In this case, the linear system of equations in~\eqref{eq:recoverErrLoc} is rank deficient and thus the error locations cannot be recovered.
Compared to the reaction-based attacks on QC-LDPC/MDPC codes in the Hamming metric~\citep{guo2016key}, the effectiveness of the reaction-based attack in the rank metric is much higher since there exist much more equivalent keys of the form $\vec{W}\vec{H}_{LRPC}$ where $\vec{W}\in\GL{n-k}{\Fq}$ that allow decryption (i.e. efficient decoding).

The reaction-based attack by~\citet{aragon2019key} is evaded in ROLLO by using \emph{ephemeral} keys.

\section{Systems based on the Hardness of List Decoding}\label{ssec:rank_crypto_list_decoding}

\subsection{The Faure--Loidreau Cryptosystem}

The Faure--Loidreau (FL) cryptosystem was proposed by \citet{faure2006new} as the rank-metric analog of \citep{AugotFiniasz-PKC-PolyReconstruction_2003}.
While the Augot--Finiasz cryptosystem is closely connected to (list) decoding Reed--Solomon codes, the
FL cryptosystem is connected to (list) decoding Gabidulin codes.
The FL system was broken for all parameters with a structural attack by \citet{Gaborit_DecodingAttack_2016}, which implicitly used the fact that the public key is a corrupted codeword of an interleaved Gabidulin code.

The security of the original Faure--Loidreau (FL) cryptosystem \citep{faure2006new} is based on Problem~\ref{pro:int-search-rsd}.
However, the relevant parameter $w$ was chosen such that it is easy to attack the system using a decoder for interleaved Gabidulin codes, which is the underlying idea of the attack by \citet{Gaborit_DecodingAttack_2016}.
The repaired Faure--Loidreau system, LIGA \citep{renner2021liga}, relies on a similar problem as Problem~\ref{pro:int-search-rsd}, with an additional restriction on $\vec{X}$, which avoids the attack from \citep{Gaborit_DecodingAttack_2016}.

The system works as follows.
Let $q,m,n,k,u,w,\tpub, \dimZ$ be positive integers that fulfill $k < n\leq m$, $2\leq u < k$,
\begin{align*}
\max \left\{n-k-\frac{k-u}{u-1}, \left\lfloor\frac{n-k}{2} \right\rfloor+1 \right\} &\leq w < \frac{u}{u+2} (n-k),
\end{align*}
and $\tpub = \left\lfloor \tfrac{n-k-w}{2} \right\rfloor$.
We consider three finite fields, $\Fq$, $\Fqm$, and $\Fqmu$, which are extension fields of each other, respectively:
\begin{equation*}
\Fq \subseteq \Fqm \subseteq \Fqmu.
\end{equation*}

The FL key generation is shown in Algorithm~\ref{alg:KeyGenFL}.
The public key consists of the tuple
\begin{align*}
\pk = (\vec{g}, k, \kpub, \tpub),
\end{align*}
where $\vec{g} \in \Fqm^n$ is the defining vector of a Gabidulin code of dimension $k$. By representing the vector
\begin{align*}
\kpub = \vec{x} \cdot \mathbf G_{\mycode{G}} + \vec{z} \in \Fqmu^n
\end{align*}
as a $u \times n$ matrix over $\Fqm$, it is a corrupted codeword of a $u$-interleaved Gabidulin code over $\Fqm$, where the error $\vec{z} = (\vec{s} \ | \ \0) \cdot \vec{P}^{-1}$ is chosen such that the $\Fq$-rank weight of each row of the matrix representation of the error is beyond the unique error correction capability of the Gabidulin code.

The secret key is the tuple $\sk = (\x,\P)$, where $\x$ is the encoded message in $\kpub$ and $\P^{-1}$ is the right factor of the error vector $\z$.

The encryption procedure is described in Algorithm~\ref{alg:EncryptFL}. It consists of encoding a zero-padded version $\m' \in \Fqm^k$ of the secret message $\m \in \Fqm^{k-u}$ with the Gabidulin code w.r.t.\ $\vec{g}$, adding the trace of a scalar multiple of $\kpub$, and an additional error $\e$. The trace of a scalar multiple of $\kpub$ is a corrupted codeword of the Gabidulin code. Hence, the ciphertext is a codeword of the Gabidulin code corrupted by two errors. An attacker is assumed to have no knowledge of the two errors and hence faces the problem of decoding a corrupted codeword of a large error weight.

The legitimate receiver has partial information about the error added through the public key: The first $w$ rows of $\P^{-1}$ are a basis of the row space of the error, which can be seen as an erasure in the rank metric. Hence, the receiver is able to retrieve the message $\m$ by an error-erasure decoder (see, e.g., Section~\ref{sec:decoding_gabidulin_codes}). This decryption procedure is outlined in Algorithm~\ref{alg:DecryptFL}.

The problem of recovering a valid private key from the public key is equivalent to Problem~\ref{pro:int-search-rsd} on page \pageref{pro:int-search-rsd}. For the chosen values of $w$, this problem can be solved for most errors (i.e., with high probability for random errors) using decoders of interleaved Gabidulin codes. This was first (implicitly) realized in \citep{Gaborit_DecodingAttack_2016}, where a structural attack on the Faure--Loidreau system was proposed.

\printalgoIEEE{
	\DontPrintSemicolon
	\caption{$\keyGen{\cdot}$~\citep{faure2006new}}%
	\label{alg:KeyGenFL}
	\KwIn{Parameters $\param = (q,m,n,k,u,w)$} %
	\KwOut{Secret key $\sk$, public key $\pk$}
	$\vec{g} \assignRand \Fqm^n:$ $\rank_q(\vec{g}) = n$ \\
	$\vec{x} \assignRand \Fqmu^k:$ $\{x_{k-u+1},\dots,x_k\}$ is a basis of $\Fqmu$ over $\Fqm$ \\
	$\vec{s} \assignRand \{ \a \in \Fqmu^w \,:\,\rank_{q} (\a) = w \}$ \label{line:keygen_s} \\
	$\vec{P} \assignRand \Fq^{n\times n}:$ $\vec{P}$ invertible \\
	$\mathbf G_{\mycode{G}}\assignDet\Mooremat{k}{q}{\vec{g}}$ \\
	$\vec{z} \assignDet (\vec{s} \ | \ \0) \cdot \vec{P}^{-1}$ \\
	$\kpub \assignDet \vec{x} \cdot \mathbf G_{\mycode{G}} + \vec{z}$ \\
	$\tpub \assignDet \left\lfloor \frac{n-w-k}{2} \right \rfloor$\\
	\Return{$\sk\assignDet(\vec{x},\vec{P})$, $\pk\assignDet(\vec{g}, k, \kpub, \tpub)$}
}

\printalgoIEEE{
	\DontPrintSemicolon
	\caption{$\enc{\cdot}$~\citep{faure2006new}}
	\label{alg:EncryptFL}
	\KwIn{Plaintext $\plainText \in \Fqm^{k-u}$, public key $\pk = (\vec{g}, k, \kpub, \tpub)$}
	\KwOut{Ciphertext $\vec{c} \in \Fqmu^n$}
	$\alpha \assignRand \Fqmu \setminus \{0\}$ \\
	$\vec{e} \assignRand \Fqm^n:$ $\rank_q(\vec{e}) = \tpub$ \\
	$\m' \assignDet (m_1,\dots,m_{k-u},0,\dots,0) \in \Fqm^k$ \\
	$\mathbf G_{\mycode{G}} \assignDet \Mooremat{k}{q}{\vec{g}}$ \\
	$\vec{c} \assignDet \vec{m}\cdot\mathbf G_{\mycode{G}} + \Tr_{q^{mu}/q^m}(\alpha \kpub) +\vec{e}$ \\
	\Return{$\vec{c}$}
}

\printalgoIEEE{
	\DontPrintSemicolon
	\caption{$\dec{\cdot}$~\citep{faure2006new}}
	\label{alg:DecryptFL}
	\KwIn{Ciphertext $\vec{c}$, secret key $\sk = (\vec{x},\vec{P})$}
	\KwOut{Plaintext $\plainText \in \Fqmu^{k-u}$}
	$\c' \assignDet \vec{c}\vec{P}|_{[w+1,n]}$ \\
	$\mycode{\mathcal{G}}' \assignDet$ Gabidulin code generated by $\vec{G}_{\mycode{G}}\vec{P}|_{[w+1,n]}$ \\
	$\plainText^{\prime\prime} \gets$ decode $\vec{c}'$ in $\mycode{\mathcal{G}}'$ (decode ``to message'' using the encoding mapping induced by $\vec{G}_{\mycode{G}}\vec{P}|_{[w+1,n]}$) \label{line:decrypt_decoding} \\
	$\{x_{k-u+1}^\ast,\dots,x_k^\ast\} \assignDet$ dual basis of $\Fqmu$ over $\Fqm$ to $\{x_{k-u+1},\dots,x_k\}$ \\
	$\alpha \gets \sum_{i=k-u+1}^{k}m_i^{\prime\prime}x_i^*$ \label{line:decrypt_alpha} \\
	$\plainText^{\prime} \gets \plainText^{\prime\prime}-\Tr_{q^{mu}/q^m}(\alpha\vec{x})$ \\
	$\plainText \assignDet (m_1',\dots,m_{k-u}')$ \\
	\Return{$\plainText$} \label{line:decrypt_m}
}

\subsection{Repair of FL Cryptosystem: LIGA}

A repair of the FL cryptosystem was proposed in \citep{wachterzeh2018repairing,renner2021liga}.
The resulting system is called LIGA, since it is based on the hardness of \underline{l}ist decoding and \underline{i}nterleaved decoding of \underline{Ga}bidulin codes.
The underyling idea of the repair is the well-known fact that all known decoders for interleaved Gabidulin codes fail for a large enough\footnote{Compared to the number of total possible error matrices, the fraction of matrices for which the decoders fail is small (i.e., the decoders suceed with high probability). However, this fraction of matrices is large enough to make a brute-force search through all of these "failing error matrices" more complex than the work factor of the system.} class of error patterns. Hence, we can modify the FL key generation algorithm to choose only errors for which the decoders fail.
This ensures that the resulting public key is not vulnerable to the structural attack in \citep{Gaborit_DecodingAttack_2016}.
More precisely, for a parameter $\zeta$ with $\zeta< \frac{w}{n-k-w}$ and $\zeta q^{\zeta w-m} \leq \tfrac{1}{2}$, replace Line~\ref{line:keygen_s} of Algorithm~\ref{alg:KeyGenFL} ($\keyGen{\cdot}$) by
{\RestyleAlgo{plain}
\begin{algorithm}%
\DontPrintSemicolon
{
	\setcounter{AlgoLine}{2}
	$\mathcal{A} \assignRand \Big\{ \text{subspace } \mathcal{U} \subseteq \Fqm^w \, : \, \dim \mathcal{U} = \zeta, \, \mathcal{U} \text{ has a basis consisting}$ $\text{only of elements that are $\Fq$-linearly independent} \Big\}$ \\[2ex]
	\setcounter{AlgoLine}{2}
	\SetNlSty{textbf}{}{'}
	$\begin{pmatrix}
	\s_1 \\
	\vdots \\
	\s_u
	\end{pmatrix} \assignRand \left\{ \begin{pmatrix}
	\s_1' \\
	\vdots \\
	\s_u'
	\end{pmatrix} \, : \, \langle\s_1',\dots,\s_u'\rangle_{\Fqm} = \mathcal{A}, \ \rank_{q}(\s_i') = w, \, \forall \, i \right\}$
}
\label{alg:s_construction}
\end{algorithm}
}

It is shown in \citep{renner2021liga} that, under the assumption that two decisional problems, called \textsf{ResIG-Dec} and \textsf{ResG-Dec}, are hard, the public-key encryption version of LIGA is IND-CPA secure in the standard model, and the key encapsulation mechanisms version is IND-CCA2 secure in the random oracle model.

Recently, it was shown in \citep{bombar2021decoding}, by proposing a message recovery attack on the system, that \textsf{ResG-Dec} is in fact not hard.
It is an open problem whether the FL/LIGA system can be protected against this attack by a further modification.

\subsection{RAMESSES}

In \citep{Ramesses}, the rank-metric code-based cryptosystem RAMESSES was presented.
Similar to the system by \citet{faure2006new}, the applied code is public, so the structure of the code does not need to be hidden. However, similar to LIGA \citep{renner2021liga}, it was recently broken by \citep{bombar2021decoding}.

\section{A System based on Rank Quasi-Cyclic Codes}\label{ssec:rank_crypto_rqc}
In this section, the system RQC, an efficient rank-metric encryption scheme from random quasi-cyclic codes is considered \citep{melchor2019rqc}, which was submitted to the NIST standardization competition \citep{melchor2019rqc}.
\subsection{Definitions}
Each element of $\Fqm^{n}$ can be uniquely represented by a polynomial of the ring $\Fqm[X]/\langle P \rangle$, where $\langle P \rangle$ denotes the ideal of $\Fqm[X]$ generated by the polynomial $P \in \Fq[X]$ of degree $n$, i.e.,
\begin{align*}
  \varphi: \Fqm^n &\rightarrow \Fqm[X]/\langle P \rangle \\
  (u_1,\hdots,u_n) &\mapsto \sum_{i=1}^{n} u_i X^{i-1}.
\end{align*}
We define the product of two vectors $\u,\v \in \Fqm^n$ by
\begin{equation*}
  \w = \u\v \in \Fqm^{n} \Leftrightarrow \varphi(\w) = \varphi(\u) \varphi(\v) \mod P.
\end{equation*}
The ideal matrix generated by $\v$ and $P$ is denoted by
\begin{align*}
  \idealCode{\v} \coloneqq
\begin{pmatrix}
    \varphi^{-1} (\varphi(\v) \mod P)\\
    \varphi^{-1} (X \varphi(\v) \mod P) \\
    \vdots \\
    \varphi^{-1} (X^{n-1} \varphi(\v) \mod P)
    \end{pmatrix}.
\end{align*}
The product $\u\v$ can then be written as a vector--matrix multiplication
\begin{equation*}
  \u \v = \u \idealCode{\v} =   \idealCode{\u}^{\top} \v = \v \u.
\end{equation*}

Ideal codes are a family of codes with a systematic parity-check matrix formed by blocks of ideal matrices.
\begin{definition}[$s$-Ideal codes]
  \label{def:idealCode}
  Let $P \in \Fq[X]$ be a polynomial of degree $n$. A parity-check matrix (under the systematic form) of an $s$-ideal code $\mycode{C}_{\mathsf{ideal}}(ns,n)$ is given by
\begin{equation*}
\H = \begin{pmatrix}
& \idealCode{\h_1}^{\top} \\
\I_{n(s-1)} & \vdots \\
& \idealCode{\h_{s-1}}^{\top}
\end{pmatrix}\in \Fqm^{n(s-1)\times ns},
\end{equation*}
where $\h_{i} \in \Fqm^{n}$ for $i=1,\hdots,s-1$. We say that $\h_1,\dots,\h_{s-1}$ generate the systematic parity-check matrix of $\mycode{C}$.
\end{definition}
For an $s$-ideal code with a systematic parity-check matrix $\H$ generated by $\h_1,\dots,\h_{s-1}$, the syndrome $\s = (\s_1,\dots,\s_{s-1}) \in \Fqm^{n(s-1)}$ of an error $\e = (\e_{1},\dots,\e_{s})\in \Fqm^{ns}$ is equal to
\begin{align*}
  \s_{i} &= \e_{i} \H^{\top}  \\
  & =  \e_{i} + \h_{i} \e_{s},
\end{align*}
for $i = 1,\dots,s-1$.

We define
\begin{align*}
\mathcal{S}^{n}_{w} \coloneqq \{\x \in \Fqm^{n}: \rank_{q}(\x) = w  \},
\end{align*}
\begin{align*}
\mathcal{S}^{n}_{1,w} \coloneqq \{\x \in \Fqm^{n}: \rank_{q}(\x) = w, 1 \in \spannedBy{x_1,\hdots,x_n} \},
\end{align*}
and
\begin{align*}
\mathcal{S}^{3n}_{w_1,w_2} \coloneqq \{ \x = (\x_1,\x_2,\x_3) &\in \Fqm^{3n}: \rank_{q}((\x_1,\x_3)) = w_1, \\ &\rank_{q}((\x_2)) = w_1+w_2, \spannedBy{\x_1,\x_3} \subset \spannedBy{\x_2} \}.
\end{align*}

\subsection{Public Key Encryption}
RQC uses two types of codes. The first code is a publicly known Gabidulin code $\Gabcode{n}{k}$ with a generator matrix $\Gabmat \in \Fqm^{k\times n}$ which is generated by $\Gabvec \in \Fqm^{n}$. The second code is a random $2$-ideal code $\Randcode(2n,n)$
ideal $[2n,n]$ code $\Randcode$
with a parity check matrix
\begin{align*}
\Randmat = \begin{pmatrix}
\I_{n} & \idealCode{\Randvec}^{\top}
\end{pmatrix} \in \Fqm^{n\times 2n}.
\end{align*}

\begin{table}[h]
\renewcommand{\arraystretch}{1.6} %
\begin{center}
\begin{tabular}{c|l|l}
Parameter & Stands for & Restriction \\
\hline
$q$ & field order & prime power \\
$m$ & extension degree & $1 \leq m$ \\
$n$ & length of the Gabidulin code & $n \leq m$ \\
$k$ & dimension of the Gabidulin code & $k \leq n$ \\
$w$ & error weight of the public key syndrome & $w\leq \lfloor \frac{n-k}{2} \rfloor$ \\
$\wR$ & error weight of the ciphertext syndrome  & $w\wR \leq \lfloor \frac{n-k}{2} \rfloor$ \\
\end{tabular}
\end{center}
\caption{Parameters of the RQC System}
\label{tab:rqcParameters}
\end{table}

Algorithms \ref{alg:KeyGenRQC_NIST}--\ref{alg:DecryptRQC_NIST} constitute the public key encryption version of the RQC scheme, where a description of the parameters is given in Table~\ref{tab:rqcParameters}.

\printalgoIEEE{
\DontPrintSemicolon
\caption{$\keyGen{\cdot}$~\citep{melchor2019rqc, Aguilar_2018}}
\label{alg:KeyGenRQC_NIST}
\KwIn{Parameters $\param = \left\{n,k,w,\wR, P\in\Fq[X] \text{ is an irreducible polynomial of degree $n$}\right\}$}
\KwOut{Secret key $\sk$, public key $\pk$}
$\Randvec \assignRand \Fqm^{n}$ \\
$\Gabvec \assignRand \mathcal{S}^{n}_{n}$ \\
$(\x,\y) \assignRand \mathcal{S}^{2n}_{1,w}$ \\
$\s \assignDet \x + \y \Randvec \mod P$\\
\Return{$\sk\assignDet(\x,\y)$, $\pk\assignDet(\Gabvec,\Randvec,\s)$}
}

\printalgoIEEE{
\DontPrintSemicolon
\caption{$\enc{\cdot}$~\citep{melchor2019rqc, Aguilar_2018}}
\label{alg:EncryptRQC_NIST}
\KwIn{Plaintext $\plainText\in\Fqm^k$, public key $\pk=(\Gabvec,\Randvec,\s)$}
\KwOut{Ciphertext $\cipherText\in\Fqm^{2n}$}
$(\r_1,\e,\r_2) \assignRand \mathcal{S}^{3n}_{w_1,w_2} $ \\
$\u \assignDet \r_1 + \r_2 \Randvec \mod P$ \\
$\v \assignDet \m\Gabmat + \s \r_2 + \e \mod P$ \\
$\cipherText \assignDet (\u,\v)$ \\
\Return{$\cipherText$ }
}

\printalgoIEEE{
\DontPrintSemicolon
\caption{$\dec{\cdot}$~\citep{melchor2019rqc, Aguilar_2018}}
\label{alg:DecryptRQC_NIST}
\KwIn{Ciphertext $\cipherText = (\u,\v) \in\Fqm^{2n}$, secret key $\sk=(\x,\y)$}
\KwOut{Plaintext $\plainText\in\Fqm^k$}
$\cipherText' \assignDet \v - \u \y$ \\
$\plainText\assignDet\decode{\mycode{G}}{\cipherText'}$ \hfill \tcp{Decode in the Gabidulin code $\Gabcode{n}{k}$ generated by $\Gabmat$ (see Section~\ref{sec:decoding_gabidulin_codes})}

\Return{$\plainText$} \\
}

\subsection{Attacks on the RQC System}
There are two known types of attacks on the RQC system which both decode in a random $2$-ideal code $\Randcode(2n,n)$ for $(\x,\y) \in \mathcal{S}^{2n}_{1,w}$ or in a random $3$-ideal code $\Randcode(3n,n)$ for $(\e,\r_1,\r_2) \in \mathcal{S}^{3n}_{\wR}$.
There is no known attack that utilizes the ideal structure of the code.
Hence, ordinary attacks on the RSD problem give the best-known attacks.
These include combinatorial attacks, such as the one by \citet{aragon2018new}, as well as the algebraic attacks by \citet{bardet2020algebraic,BBC20}.

\section{Parameters of Public-Key Encryption Schemes}\label{ssec:parameters}

Table~\ref{tab:comparison} compares the key sizes for different security levels (pre- and post-quantum) of---as of today---unbroken rank-metric code-based public-key encryption schemes. %

\begin{table*}[!t]
  \begin{center}
    \begin{tabular}{l||c|c|c|c|c }
      System name & $\sk$ & $\pk$ & $\ct$ & Security & DFR \\
      \hline  \hline
      RQC-I & 40 & 1834 & 3652 & 128 bit & no \\
      ROLLO-I-128 & 40 & 696 & 696 & 128 bit & yes \\
      Loidreau-128 & {---} & 6720 & 464 & 128 bit & no \\
      BIKE-2 Level 1 & 249 & 1271 & 1271 &  128 bit & yes \\
      McEliece348864 & 6452 & 261120 & 128 & 128 bit & no \\

      \hline
      RQC-II & 40 & 2853 & 5690 & 192 bit & no \\
      ROLLO-I-192 & 40 & 958 & 958 & 192 bit & yes \\
      Loidreau-192 & {---} & 11520 & 744 & 192 bit & no \\
      BIKE-2 Level 2 & 387 & 2482 & 2482 & 192 bit & yes \\
      McEliece460896 & 13568 & 524160 & 188 & 192 bit & no \\

      \hline
      RQC-III & 40 & 4090 & 8164 & 256 bit & no \\
      ROLLO-I-256 & 40 & 1371 & 1371 & 256 bit & yes \\
      Loidreau-256 & {---} & 16128 & 1024 & 256 bit & no \\
      BIKE-2 Level 3 & 513 & 4094 &  4094 & 256 bit & yes\\
      McEliece6688128 & 13892 & 1044992 & 240 & 256 bit & no  \\
      \hline

    \end{tabular}
  \end{center}
  \caption{Comparison of memory costs of $\sk$, $\pk$ and the ciphertext $\ct$ in Byte of IND-CCA-secure Loidreau~\citep{Bellini2019indcca} and the NIST proposals RQC~\citep{melchor2019rqc}, ROLLO~\citep{melchor2019rollo}, BIKE~\citep{aragon2019bike} and Classic McEliece~\citep{bernstein2019mceliece}. The entry `yes' in the column DFR indicates that a scheme has a decryption failure rate larger than 0.}
  \label{tab:comparison}
\end{table*}

\section{Signature Schemes}\label{ssec:rank_crypto_signatures}
Signature schemes are used to guarantee authenticity and integrity.
Code-based cryptography is almost exclusively known for PKEs and Key Encapsulation Mechanisms (KEMs)
(e.g., Classic McEliece, BIKE, HQC), but not for signature schemes. Three code-based signature
schemes had been submitted to the first round of NIST's process for PQC standardization in
2017 (pqsigRM, RaCoSS, RankSign), but none of these advanced to the third or even the second round as they were all broken.
Since the start of the NIST PQC process, there was however further progress on signatures: on the one hand, from initially very diverse submissions, only lattice-based schemes remained as finalists. On the other hand, cryptographic research has led to new promising schemes, e.g., based on isogenies. %
Code-based signatures have evolved as well. Recent proposals for signature schemes include Wave by \citet{wave} and Durandal by \citet{durandal}. Wave is based on generalized $(U,U+V)$ codes, while Durandal is an adaption of the Schnorr--Lyubashevsky \citep{lyub} approach in the rank metric. While this proposal is considered secure, its Hamming-metric counterpart was attacked in~\citep{lyubattack}.

\chapter{Applications to Storage}\label{chap:storage}

Motivated by the popularity of cloud services, coding theoretic solutions for problems related to storage and distribution of content have surged in interest in recent years. In this chapter we introduce two such settings---locality in distributed storage and coded caching---and discuss the application of rank-metric codes to address these problems. In Section~\ref{sec:locality} we provide a high-level description of the property exploited by many constructions of (MR) LRCs and explore the connection between rank-metric codes and codes with locality, both in the Hamming and the rank metric. In Section~\ref{sec:coded-caching} we present the application of Maximum Rank Distance (MRD) codes in the coded caching scheme by~\citet{TC18}.

\section{Locality in Distributed Data Storage}\label{sec:locality}

The goal of a distributed storage system is to store data such that the failure of a number of nodes is guaranteed not to incur data loss. The simplest solution is replication, where identical copies of the data are stored at each node. While this has the advantage that a failed node is simply recovered by creating another replica, the downside is the significant storage overhead. To reduce this overhead, systems such as Facebook's f4 storage system~\citep{muralidhar2014f4} and the Google File System~\citep{google2010} have transitioned to employing maximum distance separable (MDS) codes. Given a number of node failures that has to be tolerated, the MDS property guarantees a minimal storage overhead. However, if even a single node fails in such a system, which is the most likely failure event, it also implies that a large number of nodes needs to be involved in the repair process. To mitigate this effect, storage codes with locality were introduced\footnote{This should be regarded as a specific notion of locality that caters to the requirements of distributed storage systems. The concept of locality in general has a long history in coding theory and is the underlying property for, e.g., the majority logic decoding algorithm for Reed--Muller (RM) codes~\citep[Ch. 13]{MacWilliamsSloane_TheTheoryOfErrorCorrecting_1988}.}, which enforce linear dependencies between smaller subsets of positions. This suggests a separation of the parity check equations into \emph{local parities}, which are prescribed to have a specific support, and \emph{global parities}, which are unrestricted. The seminal work by \citet{gopalan2012locality} introduced a bound on the minimum Hamming distance of a code with one local parity for every subset of a partition of the code positions and an arbitrary number of global parities. Codes for this setting are commonly referred to as \emph{locally recoverable codes} (LRCs). This class forms a special case of the general definition of \emph{codes for topologies}, formally introduced in \citep{gopalan2014explicit}. There, a topology is defined as a restriction on the support of the parity-check matrix of a code. Given such a topology, a subclass of particular interest are \emph{maximally recoverable} (MR) codes \citep{Huang2007,chen2007maximally}, which guarantee to correct any erasure pattern that is theoretically correctable, given the locality (support) constraints. The special case of MR LRCs \citep{Huang2007, chen2007maximally, blaum2013partial,gopalan2014explicit,balaji2015partial,blaum2016construction,calis2016general,hu2016new,gopalan2017maximally,Horlemann-Trautmann2017,gabrys2018constructions,martinez2019universal}, also referred to as partial MDS (PMDS) codes, has received considerable attention in recent years. Interestingly, many constructions of MR LRCs \citep{blaum2013partial,rawat2014optimal,calis2016general,hu2016new,gabrys2018constructions} and MR codes for other topologies \citep{gopalan2014explicit,holzbaur2021correctable} rely on some variation of MRD codes, specifically Gabidulin codes, as their global parities, even if this connection is often not made explicit.

The first part of this section is dedicated to highlighting this connection between MRD codes and codes with locality, which shows that codes originally designed for the rank metric, also have application in the Hamming metric. Furthermore, we briefly introduce the concept of codes with locality in the rank metric.

\subsection{Codes with Locality in the Hamming Metric}

In this section, we consider codes in the Hamming metric, constructed using MRD codes. For completeness, we briefly recall the definition of the Hamming distance here. A linear $[n,k,\dminHam]_q^{\sfH}$ code $\code$ is a $k$-dimensional subspace of $\F_q^{n}$ and its minimum Hamming distance is defined to be
\begin{align*}
  \dminHam %
  = \min_{\c \in \code \setminus\{\0\}} \wtHam(\c)\ ,
\end{align*}
where $\wtHam(\c)$ is the number of non-zero positions in $\c$. It is well-known \citep[Chapter~17]{MacWilliamsSloane_TheTheoryOfErrorCorrecting_1988}, that any code fulfills the Singleton bound, given by
\begin{equation} \label{eq:singletonBound}
  \dminHam \leq n-k+1
\end{equation}
and codes that meet the bound with equality are referred to as maximum distance separable (MDS) codes.

A code with locality does not only impose distance constraints on the full codewords, but also on subsets of their positions. There are several different notions of codes with locality which can be classified into three main classes. Informally, locally repairable codes (LRCs) require that every position can be recovered from a small subset of other positions. Maximally recoverable (MR) LRCs\footnote{These codes are also referred to as partial MDS codes in literature.} are additionally required to correct any erasure pattern that is possibly correctable given the locality constraints.
Finally, codes for topologies are simply defined by a restriction on the support of a parity-check matrix of the code. The latter were defined by \citet{gopalan2014explicit} and we slightly adapt the definition here to better reflect the common separation of parities in codes with locality into local and global parities.

\begin{definition}[Codes for topologies \citep{gopalan2014explicit}]\label{def:codesForTopologies}
  Let $T \subset ([n-k-\eglob] \times [n])$ be a subset of indices. We say an $[n,k]_q^{\sfH}$ code $\code$ is a \emph{code for the topology} $T$ if there exists a parity-check matrix of $\code$ given by
  \begin{align*}
    \H =
    \begin{pmatrix}
      \Hlocal\\
      \Hglobal
    \end{pmatrix}
  \end{align*}
  with $\Hlocal \in \F_{q}^{n-k-\eglob\times n}$ and $\Hglobal \in \F_q^{\eglob \times n}$ such that $\Hlocal_{i,j} = 0 \ \forall (i,j) \in T$.
\end{definition}

Observe that the restriction on the support of the parity-check matrix given by the topology $T$ is a prerequisite for a code to have locality---for any subset of positions that allows for the recovery of another position, there must exist a codeword in the dual code that is only supported on these positions. Definition~\ref{def:codesForTopologies} includes LRCs and MR LRCs as special cases, where the matrix $\Hlocal$ is given by a block diagonal matrix (up to permutation of the columns).

\begin{example}
  Informally, a code has $r$-locality if any codeword position can be recovered from at most $r$ other codeword positions. Commonly, these \emph{repair sets} are assumed to partition the set of codeword indices. Consider, for example, a code of length $n=12$ and locality $r=3$, with local repair sets given by the partition of $[n]$ into the three sets $\cW_1= [1,4]$, $\cW_2= [5,8]$, and $\cW_3= [9,12]$. A code where each position $\c_i$ for $i\in \cW_j, j\in[1,3]$ can be recovered from the remaining positions $\{\c_l \ | \ l \in \cW_j \setminus \{i\}\}$ in the set $\cW_j$ is obtained by choosing
  \begin{align*}
    \Hlocal =
    \begin{pmatrix}
     1&1&1&1&0&0&0&0&0&0&0&0 \\
     0&0&0&0&1&1&1&1&0&0&0&0 \\
     0&0&0&0&0&0&0&0&1&1&1&1 \\
    \end{pmatrix} \ .
  \end{align*}
  To see that position $i=1$ can be recovered as required, it suffices to rearrange the first parity-check equation given by
  \begin{align*}
    \sum_{l=1}^{4} \c_l = 0 \quad \Leftrightarrow \quad \c_1 = - \sum_{l=2}^{4} \c_l \ .
  \end{align*}
  Obviously, this matrix fulfills the constraints of the topology
  \begin{align*}
    T \! = \! \big\{ (1,l)  |  l \in [5,12]\big\}\! \cup \! \big\{ (2,l)  |  l \in ([1,4] \! \cup \! [9,12])\big\} \! \cup \! \big\{ (3,l) |  l \in [1,8]\big\} \ .
  \end{align*}
\end{example}

Depending on the specific subclass under consideration, the global parities given by $\Hglobal$, which are not restricted in their support, are then used to either maximize the minimum Hamming distance of the code (LRCs) or the number of erasures that can be corrected once the local correction capabilities are exhausted (MR LRCs), where the latter is a strictly stronger property.
The goal of this chapter is to show that MRD codes are a natural choice to provide these global parities, a fact that has been exploited (more or less explicitly) by many of the known constructions, in particular for the stronger notion of locality of MR codes.

To this end, we first establish a well-known connection between the correctability of a given set of erasures and the properties of the generator and parity-check matrix of a linear code.

\begin{proposition}\label{prop:rankErasureCorrection}
  Let $\code$ be an $[n,k,\dminHam]_q^{\sfH}$ code. Denote by $\G$ and $\H$ an arbitrary generator and parity-check matrix of $\code$. Then a set of erasures $\cE \subset [n]$ is correctable if and only if the following equivalent conditions hold:
  \begin{enumerate}
  \item The generator matrix restricted to the non-erased positions is of full rank, i.e.,
    \begin{align*}
      \rank(\G|_{[n]\setminus \cE}) = k \ .
    \end{align*}
  \item The parity-check matrix restricted to the erased positions is of full rank, i.e.,
    \begin{align*}
      \rank(\H|_{\cE}) = |\cE| \ .
    \end{align*}
  \end{enumerate}
\end{proposition}

\subsection{Global Parities via MRD Codes}

To establish the goal in terms of erasure correction capability when designing codes with locality, we first discuss the maximal improvement achievable by adding global parities. A similar analysis was carried out for binary locality restrictions by \citet{gopalan2014explicit} and for the subclass of grid-like topologies by \citet{holzbaur2021correctable}. Here, we focus on the general definition of locality given in Definition~\ref{def:codesForTopologies} because the advantage of using MRD codes lies in their generality. However, note that the following also directly applies to MR LRCs, for which the set of correctable patterns was derived by \citet{blaum2013partial}.

Consider a code $\code$ as in Definition~\ref{def:codesForTopologies}. Assume the matrix $\Hlocal \in \F_q^{n-k-s}$ is given and fulfills the restrictions imposed by some topology $T$. We denote by $\bbE$ the set of erasure patterns of weight $n-k-s$ that the code spanned by the local parity checks of $\Hlocal$ can correct. By Proposition~\ref{prop:rankErasureCorrection} this set is exactly given by
\begin{align*}
  \bbElocal = \{\cE \subset [n] \ | \ |\cE| = n-k-s, \rank(\Hlocal|_{\cE}) = n-k-s \}
\end{align*}
Note that this set uniquely defines the set of erasure patterns of arbitrary weight correctable in this code, given by all subsets of the elements of $\bbElocal$, i.e., the code can correct all patterns in
\begin{align*}
  \{\cE \subset [n] \ | \ \exists \cE' \in \bbElocal \ \text{s.t. } \cE \subseteq \cE' \} .
\end{align*}
As the correctability of an erasure pattern $\cE'$ directly implies that all its subsets can be corrected, we only consider these \emph{maximal} patterns here.

Similarly, the set of maximal erasure patterns $\bbE$ correctable by the code $\code$ is given by
\begin{align}
  \bbE = \{\cE \subset [n] \ | \ |\cE| = n-k, \rank(\H|_{\cE}) = n-k \} \ . \label{eq:correctableAll}
\end{align}
Observe that
\begin{align*}
  \rank(\H|_{\cE}) =
  \rank \left(\begin{pmatrix}
      \Hlocal|_\cE\\
      \Hglobal|_\cE
    \end{pmatrix} \right) = n-k \ ,
\end{align*}
and the dimensions of the matrices directly imply the necessary condition that $\rank(\Hlocal|_{\cE}) = n-k-s$ for any $\cE \in \bbE$. Trivially, it follows that there exists some $\cE' \subset \cE$ with $|\cE'| = n-k-s$ such that $\rank(\Hlocal|_{\cE'}) = n-k-s$, i.e., the pattern $\cE$ can be written as a union of two disjoint sets $\cE = \cE' \cup \cI$ with $\cE' \in \bbElocal$ and $|\cI| = s$. This implies that
\begin{align}
  \bbE \subseteq \{\cE' \cup \cI \ | \ \cE' \in \bbElocal, \cI \in [n]\setminus \cE', |\cI| = s  \} \ . \label{eq:correctableSubset}
\end{align}
In other words, every erasure pattern correctable by $\bbE$ can be written as the union of a pattern in $\bbElocal$ and $s$ additional positions.

To see that this is exactly the set of correctable patterns, we require the following lemma, which is based on employing MRD codes as the global parities. Note that similar methods have been used in literature to prove properties of codes with locality, see, e.g., \citep[Claim 3]{hu2016new}.

\begin{lemma}\label{lem:independenceMRDParities}
  Let $\A\in \F_q^{a\times n}$ with $\rk(\A)=a$ and $\B \in \F_{q^m}^{b\times n}$ be a parity-check matrix of an $[n,n-b]_{q^m}^{\sfR}$ MRD code, where $b \leq n-a$. Then
  \begin{equation*}
    \rk_{q^m}\left(
      \begin{bmatrix}
        \A \\ \B
      \end{bmatrix} \right) = a+b \ .
  \end{equation*}
\end{lemma}
\begin{proof}
  The ranks of $\A,\B$ sum up if their $\Fqm$-row spans intersect trivially. First, consider an element in the row space $\myspan{\A}_{\Fqm}$ given by
  \begin{equation*}
    \v =\u \cdot \A
  \end{equation*}
  with $\u \in \Fqm^{a}$. Clearly, $\rk_{\Fq}(\u) \leq a$ and since $\A$ is in $\Fq$, we have $\rk_{\Fq}(\v) \leq a$.

  Now consider $\myspan{\B}_{\Fqm}$. As $\B$ is the parity-check matrix of an $[n,n-b]_{q^m}^{\sfR}$ MRD code, its row span is an $[n,b]_{q^m}^{\sfR}$ MRD code. Therefore, every non-zero element $\w \in \myspan{\B}$ is of $\Fq$-rank
  \begin{align*}
\rk_{\Fq}(\w) \geq n-b+1 \geq n-(n-a)+1 = a+1 \ .
  \end{align*}
  Hence, the $\Fqm$-row spans of $\A$ and $\B$ intersect trivially and the lemma statement follows.
\end{proof}

This lemma directly implies a construction of a code for the topology $T$ that is optimal in terms of its ``global'' erasure correction capability, i.e., given the matrix $\Hlocal$ it corrects any pattern that is possibly correctable by adding $s$ global parities. %

\begin{construction} \label{con:constructionMRDTopology}
  Given a matrix $\Hlocal \in \F_q^{n-k-s\times n}$ that fulfills the constraints imposed by the topology $T$, define $\code$ to be the code spanned by
  \begin{align*}
    \H =
    \begin{pmatrix}
      \Hlocal\\
      \Hglobal
    \end{pmatrix} \ ,
  \end{align*}
  where $\Hglobal \in \F_{q^m}^{s \times n}$ is the parity-check matrix of an $[n,n-s]_{q^m}^{\sfR}$ MRD code.
\end{construction}

Now consider an erasure pattern
\begin{align*}
  \cE \in \{\cE' \cup \cI \ | \ \cE' \in \bbElocal, \cI \in [n]\setminus \cE', |\cI| = s  \}
\end{align*}
and a code $\code$ as in Construction~\ref{con:constructionMRDTopology}. It is easy to see that $\Hglobal|_{\cE}$ spans a $[k+s,k]_{q^m}^{\sfR}$ MRD code. Furthermore, by \eqref{eq:correctableAll} the matrix $\Hlocal|_{\cE}$ is of full-rank.
It follows directly from Lemma~\ref{lem:independenceMRDParities} that $\H|_{\cE}$ is of full rank and therefore that $\code$ can correct the pattern $\cE$. Hence, the set of correctable patterns is exactly the set give in \eqref{eq:correctableSubset}, i.e.,
\begin{align*}
  \bbE = \{\cE' \cup \cI \ | \ \cE' \in \bbElocal, \cI \in [n]\setminus \cE', |\cI| = s  \} \ .
\end{align*}
In particular, notice that if $\Hlocal$ spans a code that is MR for the respective topology \emph{without} global redundancy, i.e., a code that can correct any erasure pattern that is correctable given the support constraints on $\Hlocal$ imposed by the topology $T$, Construction~\ref{con:constructionMRDTopology} gives a code that is MR for the topology $T$ including the ``global'' unrestricted parities in $\Hglobal$.

The construction given in Construction~\ref{con:constructionMRDTopology} is intended as a proof of concept on how MRD codes can be applied to construct codes with locality. It is easy to see that requiring $\Hglobal$ to span an $[n,n-s]_{q^m}^{\sfR}$ MRD code is a stronger property than required for the application of Lemma~\ref{lem:independenceMRDParities}. Instead it suffices that for any $\cE \in \bbE$ the code spanned by $\Hglobal|_{\cE}$ is an $[k+s,k]_{q^m}^{\sfR}$ MRD code.

On a high level, the constructions by \citet{blaum2013partial,rawat2014optimal,calis2016general,hu2016new,gopalan2014explicit,holzbaur2021correctable} are based on ensuring this property (or a similar property on the generator matrix), either by explicitly making use of the structure of $\Hlocal$ in the case of MR LRCs or employing more generic methods that ensure that \emph{any} subset of $k+s$ positions spans an MRD code. While there are constructions of MR LRCs resulting in lower field size, e.g., based on linearized RS code \citep{martinez2019universal}, employing MRD codes can also allow for providing other desired properties, such as the possibility of accommodating array codes \citep{rawat2014optimal} or regeneration properties \citep{holzbaur2020partial}. However, the details of these applications are beyond the scope of this survey.
\subsection{Codes with Locality in the Rank Metric}
The previous sections were concerned with the construction of codes with locality in the Hamming metric, which is the metric best-motivated by storage applications. However, the concept of locality has also been considered in the rank metric by \citet{kadhe2019codes}.
\begin{definition}[Rank-locality {\citep[Definition~2]{kadhe2019codes}}]\label{def:ranklocality}
An $[n,k,d]^R_{q^m}$ code $\mathcal{C}$ is said to have $(r,\rho)$ rank-locality, if for every column $i \in[n]$, there exists a set $\Gamma(i)\subset[n]$ of indices such that
\begin{itemize}
    \item $i\in\Gamma(i)$,
    \item $|\Gamma(i)|\leq r+\rho -1$, and
    \item $\dminRank(\mathcal{C}|_{\Gamma(i)}) \geq \rho$.
\end{itemize}
\end{definition}
Similar to the Singleton-like bound in the Hamming metric derived by \citet{gopalan2012locality} and generalized by \citet{kamath2014codes}, \citet{kadhe2019codes} proves a bound on the minimum rank-distance of codes with locality in the rank metric independent of the field size.
\begin{theorem}[Bound on the rank-distance of codes with rank-locality~{\citep[Theorem~1]{kadhe2019codes}}]\label{thm:rankbound}
For any $[n,k,d]^\sfR_{q^m}$ code $\mathbb{C}$ with rank-locality $(r,\rho)$, the minimum rank-distance $d_R$ is bounded by
\begin{equation}\label{eq:rank_loc_bound}
    \dminRank(\mathbb{C})\leq n-k+1-\left(\ceil{\frac{k}{r}}-1\right) (\rho - 1).
\end{equation}
\end{theorem}

Along with this Singleton-like bound on the distance, \citep{kadhe2019codes} presented a code construction that achieves this bound with equality. This construction can be viewed as the skew-analog of the construction of Singleton-optimal codes with locality in the Hamming metric given in~\cite{tamo2014family} (see also \citep[Section~III.C]{kadhe2019codes}), replacing the use of RS codes with Gabidulin codes. %

\begin{definition}[Code with rank-locality {\citep[Construction~1]{kadhe2019codes}}]\label{cons:codeRankLocality}
  Let $m,n,k,r,$ and $\rho$ be positive integers such that $r|k$, $(r+\rho-1)|n$, and $n|m$. Define $\mu=n/(r+\rho-1)$. Let $\cA = \{\alpha_1,\ldots,\alpha_{r+\delta-1}\}$ be a basis of $\F_{q^{r+\delta-1}}$ over $\Fq$ and $\cB=\{\beta_1,\ldots,\beta_{\mu}\}$ be a basis of $\F_{q^n}$ over $\F_{q^{r+\rho-1}}$. For $1\leq j \leq \mu$ define $\cP_j = \{\alpha_i\beta_j\ | \ i\in [r+\rho-1]\}$ and $\cP \coloneqq \bigcup_{j=1}^\mu \cP_j$.

  Define the code
  \begin{align*}
    \mycode{C} = \left\{ \big( f_\m(\gamma)\big)_{\gamma\in \cP} | \m \in \F_{q^m}^k \right\}
  \end{align*}
  with
  \begin{align*}
    f_\m(x) = \sum_{i=0}^{r-1} \sum_{j=0}^{\frac{k}{r}-1} m_{i+jr} x^{[(r+\delta-1)j+i]} \ .
  \end{align*}
\end{definition}

\begin{theorem}[{\citep[Theorem~2]{kadhe2019codes}}]
  The $[n,k]_{q^m}^R$ code of Definition~\ref{cons:codeRankLocality} has $(r,\rho)$ rank-locality and fulfills the bound of Theorem~\ref{thm:rankbound} with equality.
\end{theorem}
The idea of the proof of this statement is similar to the corresponding proof in the Hamming metric in \citep{tamo2014family} and based on determining a polynomial that is constant on all elements of $\cP_j$ for each $j\in[r+\rho-1]$. For sake of brevity, we omit the proof of this theorem and refer the interested reader to \citep{kadhe2019codes}.

In \cite{kadhe2019codes} it is further shown that lifting (see Definition~\ref{def:generalizedLifting} on Page~\pageref{def:generalizedLifting}) the codes of Definition~\ref{cons:codeRankLocality} results in codes with locality in the subspace metric.

\section{Coded Caching Scheme with MRD Codes}
\label{sec:coded-caching}

Caching is a commonly-used data management strategy to reduce the communication load during the peak-traffic time where the terminals of the communication system are equipped with local caches.

\subsection{System Description}
Consider a cache-aided broadcast system consists of a transmitting server which has access to a library of $N$ files $W_1,\dots,W_N$, and $K$ users, where each user has a cache that can store $M$ files (see Figure~\ref{fig:caching} for an illustration). The shared links between the server and users are error-free.
The communication is composed of two stages:
\begin{enumerate}
\item \emph{Placement Phase}: The users fill their caches $Z_1,\dots, Z_K$ with (coded) file segments from the library according to a placement protocol. The communication cost in this phase is negligible.
\item \emph{Delivery Phase}: The users reveal their demands $d_1,\dots,d_K$ and the server transmits a message $X_{d_1,\dots,d_K}$ so that each user $k\in[1,K]$ can recover its demanded file $W_{d_k}$ according to $X_{d_1,\dots,d_K}$ and its cache content $Z_k$, $\forall k \in[1,K]$.
\end{enumerate}
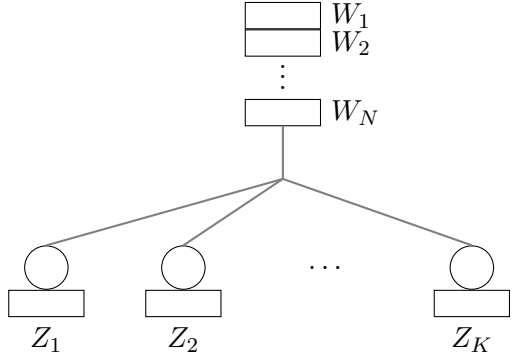
\begin{figure}[htb!]
  \centering
  \def\x{0.7}%

\begin{tikzpicture}
  [font=\normalsize,>=stealth',
  mycircle/.style={circle, draw=black, text width=.1em, minimum height=1.5em, text centered},
  mylink/.style={color=gray, thick},
  myblock/.style={rectangle, draw=black, minimum height=0.35cm, minimum width=1cm},]

  \coordinate (Source) at (0*\x,4*\x);
  \node[myblock, label=right:{$W_1$}] (W1) {};
  \node[myblock, below = \x*0pt of W1, label=right:{$W_2$}] (W2) {};
  \node[below = \x*5pt of W1] (dots) {$\vdots$};
  \node[myblock, below = \x*0pt of dots, label=right:{$W_N$}] (WN) {};  

  \node[below = \x*1cm of WN] (Mid) {};

  \node[mycircle,below left = \x*1cm and \x*4cm of Mid] (Z1) {};
  \node[myblock, below = \x*0pt of Z1, label=below:{$Z_1$}] (M1) {};
  \node[mycircle,right = \x*50pt of Z1] (Z2) {};
  \node[myblock, below = \x*0pt of Z2, label=below:{$Z_2$}] (M2) {};
  \node[right = \x*50pt of Z2] (Zdot) {$\dots$};
  \node[mycircle,right = \x*50pt of Zdot] (ZK) {};
  \node[myblock, below = \x*0pt of ZK, label=below:{$Z_K$}] (MK) {};

  \path[] (WN) edge[mylink] (Mid);

  \path[] (Mid.north) edge[mylink] (Z1.north)
  edge[mylink] (Z2.north)
  edge[mylink] (ZK.north)
  ;

\end{tikzpicture}
  \caption{Illustration of a Cache-Aided Broadcast System.}
  \label{fig:caching}
\end{figure}
\citet{MAN14}, proposed a \emph{coded caching} scheme that outperforms uncoded caching, where in the delivery phase the requested files (or segments of files) which are not cached are sent to each user individually. In the work by~\citet{MAN14} and further improved scheme by~\citet{YMA18}, binary coding is used in the communication phase while non-coded file segments are stored in the local caches.
Rank-metric codes are used in the coded caching scheme with coded placement~\citep{TC18}. This scheme is shown to outperform the optimal scheme with uncoded placement~\citep{YMA18} in the regime of small cache size.

For schemes with uncoded placement, the principle for designing the placement is that coded multicasting opportunities are created simultaneously for all possible requests in the delivery phase. When the placement applies coding, in addition to the principle above, the coding coefficients should be chosen such that a set of full rank conditions are satisfied, which guarantees that all the file segments that are part of the linear combinations of the cached symbols can be decoded after receiving sufficient coded symbols in the delivery phase~\citep[Section IV.A]{TC18}.
However, specifying the coding coefficients for generic parameters turns out to be difficult. \citet{TC18} resolved the issue by a combination of rank metric codes (used in placement phase) and MDS codes (used in delivery phase), which provides an explicit solution apart from an existence proof by~\citet{MedardKoetterKargerEffrosShiLeong-RLNCMulticast_2006}. %

\subsection{Coded Placement Scheme with Rank-Metric Codes}
In the following, we describe the scheme from \citep{TC18} with a focus on the placement phase as it uses rank-metric codes.

Fix an integer parameter $1\leq t\leq K$
and partition each file $W_n$ into $\binom{K}{t}$ subfiles $W_{n,  \mathcal{S}}$, where $\mathcal{S} \subseteq [1,K], |\mathcal{S} |=t$.

A Gabidulin code $\mycode{G}(P_0,P)$ is used to encode the cached symbols, where
$$P = \binom{K-1}{t-1}N$$
and
$$P_o = 2\binom{K-1}{t-1}N-\binom{K-2}{t-1}(N-1) .$$ The code locators are denoted by $g_1, g_2,\dots,g_{P_0}\in\mathbb{F}_{q^m}$ and are linearly independent in $\mathbb{F}_{q}$. %

The cache $Z_k$ of each user $k$ is filled as follows:
\begin{enumerate}
\item Collect $P$ subfiles $$ \{ W_{n,  \mathcal{S}}:  \text{ for all }n\in [1,N]\text{ and }\mathcal{S}\text{ such that }k\in \mathcal{S}\}.$$
  Interpret each subfile as an element in $\Fqm$. Denote by $v_i$ the $i$-th subfile in this set, for all $i\in[1,P]$. %
\item Encode the subfiles by evaluating the linearized polynomial
  $$f(x)=\sum_{i=1}^{P} v_i x^{q^{i-1}},  v_i\in \mathbb{F}_{q^m} $$
  at the last $P_0-P$ code locators $g_{P_0-P+1}, \dots,g_{P_0}\in\mathbb{F}_{q^m}$ and place these coded symbols into the cache. %

\end{enumerate}
It can be seen that the cached content are the linear combination of subfiles of different files. Once the demand $d_k$ of the user $k$ is known, the subfiles of $W_n, n\neq d_k$, which are components in the linear combination, will be seen as interference to the subfiles of $W_{d_k}$. The strategy of a delivery scheme is sending symbols to eliminate this interference as well as sending the remaining subfiles of $W_{d_k}$ which are not part of the linear combinations stored in the cache $Z_k$.
The delivery scheme by~\citet{TC18} utilizes MDS codes and consists of three main steps:
\begin{enumerate}
\item For each file $W_n$, send the uncoded subfiles that are \emph{only} cached by the users who do not request $W_n$.
\item For each file $W_n$, collect all subfiles cached by the users of whom \emph{some} do not request file $W_n$ (skipping the subfiles that are already sent in Step 1), encode them with a systematic MDS code over $\Fq$ and send the parity-check symbols. %
\item For each file $W_n$, collect all subfiles cached by the users of whom \emph{all} request $W_n$, encode them with a systematic MDS code over $\Fq$ and send the parity-check symbols. %
\end{enumerate}
The length of the MDS code in Step 2 is chosen such that each user receives $2P-P_0$ symbols which are linear combinations of the same subfiles as its cached symbols which are $\Fq$-linearly independent of the cached symbols. %
Therefore, after Step 1 and 2, each user can eliminate the interference and decode all the subfiles of the requested file which are part of the basis of its cached symbols, due to the fact it has collected $P$ $\Fq$-linear combinations of $P$ subfiles. The purpose of Step 3 is to guarantee that each user decodes the subfiles of its requested file which are not part of the linear combinations of its cached symbols. The length of the MDS in this step is therefore chosen to guarantee that all users receive sufficient linear independent symbols to achieve that.

The sketch above is a description of the delivery scheme
for the case where all files are requested. For details and a variation for the case where some files are not requested, we refer interested readers to the original work~\citep{TC18}.

\chapter{Applications to Network Coding}
\label{chap:network_coding}
Network coding is an elegant technique introduced to improve network throughput and performance. With its simple premise that intermediate nodes in the network can process incoming packets instead of only forwarding them, algebraic coding is a very straightforward approach to cope with this problem.
In this chapter we introduce the applications of rank-metric codes in finding solutions for deterministic networks and error correction in random networks.
In Section~\ref{sec:NC_intro} we present several classifications of network coding problems.
Section~\ref{sec:GeneralizedCombinationNetworks} presents a class of constructions based on MRD codes for a class of deterministic multicast networks, which guarantee that all the receivers decode all the messages. Two error models in networks that are often considered are described in Section~\ref{sec:error-model}. Section~\ref{sec:subspace-codes} introduces subspace codes, with the focus on the constructions by lifting rank-metric codes. Section~\ref{sec:upperbound-size} provides upper bounds on the size of subspace codes and an analysis of list decoding subspace codes. %
\section{Introduction}
\label{sec:NC_intro}
Consider a network consisting of one or multiple sources and one or multiple receivers with intermediate relay nodes as possible connections between the sources and receivers.
The \emph{network coding problem} can be formulated as follows: for each node in the network, find a function of its incoming messages to transmit on its outgoing links, such that each receiver can recover all (or a predefined subset of all) the messages.
The set of functions that solve the network coding problem is a \emph{solution} of the network. A network is said to be \emph{solvable} if such functions exists.
In the seminal paper by~\citet{Ahlswede_NetworkInformationFlow_2000} it was shown that network coding increases the throughput compared to simple routing. Further, it was shown that network coding
achieves the capacity of multicast networks (i.e., one-to-many scenarios).
We can distinguish between linear and non-linear network coding and between deterministic and random network coding.

\begin{itemize}
\item Linear vs.~Non-Linear Network Coding

  If the functions of a solution are restricted to be linear, the {network coding} is referred to as \emph{linear network coding}.
It was shown by~\citet{LiYeungCai-LinearNetworkCoding_2003} that forwarding \emph{linear} combinations of the incoming packets at intermediate nodes suffices to achieve the multicast capacity.
This observation is important since linear combinations are simple operations that can be performed efficiently at the intermediate nodes as well as at the transmitter and receiver.
The performance of network coding depends on the choice of functions (also known as \emph{coding coefficients}) at the intermediate nodes since an unfortunate choice of coding coefficients might cancel packets out. An algebraic formulation of the linear network coding problem and its solvability can be found in~\citep{KoetterMedard-AlgebraicApproachNetworkCoding_2003}. Some works investigated the performance of network coding for practical uses, such as coded TCP~\citep{CodedTCP2011} and forwarding coded packets in wireless mesh networks~\citep{COPE2008}.
If a non-linear function is employed in the network, then it is classified as \emph{non-linear network coding}.
Non-linear codes may employ a smaller alphabet than linear codes~\citep{LLjan2004}, however, they are rarely considered for practical usage due to the lack of efficient decoding algorithms.
\item Deterministic vs.~Random Network Coding

A network coding problem is called \emph{deterministic} if the coding functions of all nodes are fixed and known by the receivers.
A deterministic algorithm to compute the optimal coding coefficients that achieve the min-cut maximum flow capacity for the in-network linear combinations in polynomial time was proposed by~\citet{JaggiDeterministicLNC}.
In Section~\ref{sec:GeneralizedCombinationNetworks} we will discuss the application of rank-metric codes and subspace codes to \emph{(generalized) combination networks}.

It was shown by~\citet{MedardKoetterKargerEffrosShiLeong-RLNCMulticast_2006} that the multicast capacity can be reached with successful decoding probability approaching 1 as the field size of the \emph{randomly} chosen coefficients of the linear combinations goes to infinity.
In \emph{random linear network coding} (RLNC), each node can choose the coding functions randomly.
This RLNC approach has no need for central coordination of the network coding coefficients and thus can be used for dynamic networks.

\item Scalar vs. Vector Network Coding

If both, the coding coefficients and packets, are scalars, the solution is called a \emph{scalar network coding solution}.
\citet{KoetterMedard-AlgebraicApproachNetworkCoding_2003} provided
an algebraic formulation for the linear network coding problem
and its scalar solvability.
Vector network coding as part of \emph{fractional network coding} was mentioned in~\citep{CannonsDoughertyFreilingZeger-2006}.
A solution of the network is called an $(s,t)$ \emph{fractional vector network coding solution},
if the edges transmit vectors of length $t$, but the message vectors
are of length $s \leq t$. The case $s=t=1$ corresponds to a scalar solution.
\citet{EbrahimiFragouli-AlgebraicAlgosVectorNetworkCoding}
have extended the algebraic approach from~\citep{KoetterMedard-AlgebraicApproachNetworkCoding_2003} to
\emph{vector network coding}.
Here, all packets are vectors of length~$t$ and the coding coefficients are matrices.
A set of $t \times t$ coding matrices for which all receivers can recover their
requested information, is called a \emph{vector network coding solution} (henceforth, it will be called \emph{vector solution}).
Notice that vector operations imply linearity over vectors; therefore, a vector solution is always a (vector) \emph{linear} solution.
In terms of the achievable rate, vector network coding outperforms scalar
linear network coding~\citep{MedardKargerEffrosKargerHo_2003,DoughertyFreilingZeger-NetworksMatroidsNonShannon_2007,etzion2018vector}.
In \citep{etzion2018vector}, it was shown that for special networks (generalized combination networks), vector coding
solutions based on rank-metric and subspace codes significantly reduce the required \emph{alphabet size}.
In one subfamily of these networks,
the scalar linear solution requires a field size $q_s = q^{(h-2)t^2/h + o(t)}$, for even $h \geq 4$,
where $h$ denotes the number of messages,
while \citep{etzion2018vector} provides a vector solution of field size $q$ and dimension~$t$.
Such a vector solution has the same alphabet size as a scalar solution of field size $q^t$.
\end{itemize}

\section{Solutions of Generalized Combination Networks}
\label{sec:GeneralizedCombinationNetworks}

An $(\eps,\ell)$-$\cN_{h,r,\alpha\ell+\eps}$ generalized combination network is a class of multicast networks, illustrated in Figure~\ref{fig:Network} (see also~\citep{etzion2018vector}).
\begin{figure}[!htb]
  \centering
  \captionsetup{justification=centering}
  \def\x{0.55}%

\begin{tikzpicture}
  [font=\normalsize,>=stealth',
  mycircle/.style={circle, draw=TUMdgray, very thick, text width=.1em, minimum height=1.5em, text centered},
  mycircle_small/.style={circle,draw=TUMdgray!90,very thick, inner sep=0,minimum size=1em,text centered},
  mylink/.style={color=TUMblue!70, thick},
  mylink_dir/.style={color=TUMgreen!70, thick, dashed}]

  \coordinate (Source) at (0*\x,4*\x);
  {\node[mycircle,label=above right:{$\ve{x}={(\ve{x}_1,\ve{x}_2,\dots,\ve{x}_h)}$}] (Source) {};}

  \node[mycircle_small,below left = \x*60pt and \x*100pt of Source] (M0) {$\vec{y}_1$};
  \node[mycircle_small,right = \x*20pt of M0] (M1) {$\vec{y}_2$};
  \node[mycircle_small,right = \x*20pt of M1] (M2) {};
  \node[mycircle_small,right = \x*20pt of M2] (M3) {};
  \node[draw=none,right = \x*20pt of M3] (M4) {$\dots$};
  \node[mycircle_small,right = \x*20pt of M4] (M5) {};
  \node[mycircle_small,right = \x*20pt of M5] (M6) {$\vec{y}_r$};

  \node[mycircle,below left = \x*60pt and \x*20pt of M0] (R0) {};
  \node[mycircle,right = \x*50pt of R0] (R1) {};
  \node[mycircle,right = \x*50pt of R1] (R2) {};
  \node[draw=none,right = \x*50pt of R2] (R3) {$\dots$};
  \node[mycircle,right = \x*50pt of R3] (R4) {};

  \path[] (Source) edge[mylink,bend right=15] (M0.north)
  edge[mylink,bend right=15] (M1.north)
  edge[mylink,bend right=15] (M2.north)
  edge[mylink,bend right=15] (M3.north)
  edge[mylink,bend left=15] (M5.north)
  edge[mylink,bend left=15] (M6.north);

  \path[] (R0) edge[mylink,bend left=15] (M0.south)
  edge[mylink,bend left=15] (M1.south);
  \path[] (R1) edge[mylink,bend left=15] (M0.south)
  edge[mylink,bend left=15] (M2.south);
  \path[] (R2) edge[mylink,bend left=15] (M1.south)
  edge[mylink,bend left=15] (M3.south);
  \path[] (R4) edge[mylink,bend right=15] (M5.south)
  edge[mylink,bend right=15] (M6.south);

  \draw [mylink,solid] ($(M0)+(\x*18pt,\x*15pt)$) arc (30:80:\x*15pt); %
  \draw [mylink,-] ($(R0)+(\x*20pt,\x*10pt)$) arc (0:100:\x*15pt); %
  \node[draw=none,above right = \x*-8pt and \x*2pt of M0] (ell) {\textcolor{TUMblue}{$\ell$}};
  \node[draw=none,above right=0pt and -5pt of M6] (rs) {$r$ middle nodes};
  \node[draw=none,right=0pt of R4] (N) {$N=\binom{r}{\alpha}$};
  \node[draw=none,below right=\x*-5pt and \x*-75pt of N] (recNodes) {receivers};

  \path[] (Source) edge[mylink,bend right=5] (M0.north)
  edge[mylink,bend right=5] (M1.north)
  edge[mylink,bend right=5] (M2.north)
  edge[mylink,bend right=5] (M3.north)
  edge[mylink,bend left=5] (M5.north)
  edge[mylink,bend left=5] (M6.north);

  \path[] (R0) edge[mylink,bend left=5] (M0.south)
  edge[mylink,bend left=5] (M1.south);
  \path[] (R1) edge[mylink,bend left=5] (M0.south)
  edge[mylink,bend left=5] (M2.south);
  \path[] (R2) edge[mylink,bend left=5] (M1.south)
  edge[mylink,bend left=5] (M3.south);
  \path[] (R4) edge[mylink,bend right=5] (M5.south)
  edge[mylink,bend right=5] (M6.south);

  \path[] (Source.west) edge[mylink_dir, bend right=45] (R0.west)
  edge[mylink_dir, bend right=50] (R0.west)
  (Source) edge[mylink_dir,bend left=5] (R1.east)
  edge[mylink_dir,bend left=10] (R1.east)
  (Source) edge[mylink_dir,bend left=15] (R2.east)
  edge[mylink_dir,bend left=20] (R2.east)
  (Source) edge[mylink_dir,bend right=0] (R4.west)
  edge[mylink_dir,bend right=5] (R4.west);

  \draw [mylink_dir,solid] ($(R0)+(-\x*15pt,\x*35pt)$) arc (100:60:\x*15pt); %
  \node[draw=none,above left = \x*10pt and 3pt of R0] (alphal) {\textcolor{TUMgreen}{$\varepsilon$}};
  \node[draw=none,right = 0pt of R0] (alphal) {\textcolor{TUMblue}{$\alpha\ell$}};

\end{tikzpicture}
  \caption{Illustration of $(\eps,\ell)-\mathcal{N}_{h,r,\alpha\ell+\eps}$ networks}
  \label{fig:Network}
\end{figure}
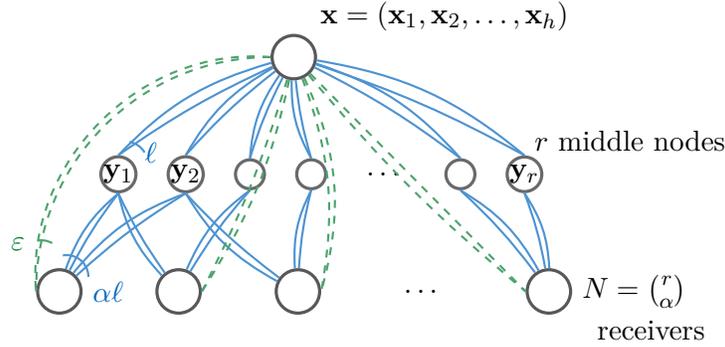
The network has three layers. The first layer consists of a source with $h$ source messages. The source transmits coded messages to $r$ middle nodes via $\ell$ parallel links (solid lines) between itself and each middle node. Any $\alpha$ middle nodes in the second layer are connected to a unique receiver (again, by $\ell$ parallel links each). Each receiver is also connected to the source via $\eps$ direct links (dashed lines). Each one of the $\binom{r}{\alpha}$ receivers demands all the $h$ messages.
It was shown by~\citet[Thm.~8]{etzion2018vector} that the $(\eps,\ell)$-$\cN_{h,r,\alpha\ell+\eps}$ network has a trivial solution if $h\leq\ell+\eps$ and it has no solution if $h>\alpha\ell+\eps$. Therefore we only consider the non-trivially solvable networks with $\ell+\eps< h \leq \alpha\ell+\eps$ in this section.

A linear solution of an $(\eps,\ell)$-$\cN_{h,r,\alpha\ell+\eps}$ network is a set of coding coefficients at all middle nodes, such that every receiver $j\in[N]$ can recover all the $h$ source messages $\ve{x}_1,\dots,\ve{x}_h$ from the received message $\ve{y}_{j_1}, \dots, \ve{y}_{j_\alpha}$.
If the messages $\ve{x}_1,\dots,\ve{x}_h$ are vectors in $\Fq^t$ and the coding coefficients for each middle node are matrices over $\F_q$, the corresponding solution is called a \emph{vector linear solution}.
If the messages $x_1,\dots,x_h$ are scalars in $\F_{q}$ (i.e., $t=1$) and the coding coefficient at each middle node is a vector over $\Fq$, then the corresponding solution is called a \emph{scalar linear solution}.

In the following, we provide two constructions for vector linear solutions of the $(\eps,\ell)$-$\cN_{h,r,\alpha\ell+\eps}$ networks using rank-metric codes in Section~\ref{sec:VecSol}. %

\subsection{Vector Solutions Using MRD Codes}
\label{sec:VecSol}

In this section, we present several constructions from~\citep{etzion2018vector} of vector solutions of $(\eps,\ell)$-$\cN_{h,r,\alpha\ell+\eps}$ networks.

\begin{theorem}[\citet{Roth_RankCodes_1991,Lusina2003Maximum}]\label{thm:code-square}
Let
$$\mycode{D}_t \coloneqq \{ \0_t , \I_t , \Mat{C} ,\Mat{C}^2 , \ldots , \Mat{C}^{q^t-2} \},$$
where $\Mat{C}$ is a companion matrix of a primitive polynomial $p(x) = p_0 + p_1 x + \dots + p_{t-2}x^{t-2} + p_{t-1}x^{t-1} + x^t \in \Fq[x]$, i.e.,
\begin{equation*}
\Mat{C} \coloneqq
\begin{pmatrix}
0 & 1 & 0 & \dots & 0 & 0\\
0 & 0 & 1 & \dots & 0 & 0\\
\vdots\\
0 & 0 & 0 & \dots & 0 & 1\\
-p_0 & -p_1 & -p_2 & \dots & -p_{t-2}& -p_{t-1},
\end{pmatrix}\in\Fq^{t\times t}.
\end{equation*}
and  is a primitive polynomial.
Then, $\mycode{D}_t$ is an $\MRD{t,q^{t}}$ code of~$q^t$ pairwise \emph{commutative} matrices.
\end{theorem}
The following corollary considers block Vandermonde matrices which will be used for the vector solution.
Note that $\I_t = \Mat{C}^{q^t-1}\in \mycode{D}_t$. %
\begin{corollary}[\citet{etzion2018vector}]\label{cor:q-vand-block-matrix}
  Let $\mycode{D}_t$     %
  be the $\MRD{t,q^t}$ code defined by the companion matrix $\Mat{C}$
(Theorem~\ref{thm:code-square}).
Let $\Mat{C}_i$, $i=1,\dots,h$, be distinct codewords of~$\mycode{D}_t$.
Define the following $(ht) \times (ht)$ block matrix:
\begin{equation*}
\Mat{M} =
\begin{pmatrix}
\I_t & \Mat{C}_1 & \Mat{C}_1^2 & \dots & \Mat{C}_1^{{h-1}}\\
\I_t & \Mat{C}_2 & \Mat{C}_2^2 & \dots & \Mat{C}_2^{{h-1}}\\
\vdots & \vdots& \vdots & \ddots & \vdots\\
\I_t & \Mat{C}_{h} & \Mat{C}_{h}^2 & \dots & \Mat{C}_{h}^{{h-1}}\\
\end{pmatrix}.
\end{equation*}
Then, any  $(h t) \times (\ell t)$ submatrix consisting of $h \ell$ blocks of
\textbf{consecutive columns} has full rank $\ell t$, for any $\ell=1,\dots,h$.
\end{corollary}
Note that the blocks of rows do not have to be consecutive, but a block has to be included with all its $t$ rows in the submatrix.

Based on this corollary, we can now provide several constructions for the generalized combination networks with different parameters. %

\begin{construction}[Construction for $(0,1)$-$\mathcal{N}_{h,r,h}$ combination network]
\label{constr:comb-network-gen}
Let $$\mycode{D}_t =\{\Mat{C}_1, \Mat{C}_2,\dots, \Mat{C}_{q^t}\}\subset\Fq^{t\times t}$$ be the $\MRD{t,q^t}$ code defined by the companion matrix~$\Mat{C}$ (Theorem~\ref{thm:code-square}) and let $r \leq q^t+1$. Consider the $(0,1)$-$\mathcal{N}_{h,r,h}$ combination network with message vectors $\vec{x}_1, \dots, \vec{x}_h\in\Fq^t$.
One node from the middle layer receives and transmits $\vec{y}_{r} = \vec{x}_h$ and
the other $r-1$ nodes of the middle layer transmit
$$\vec{y}_{i} =
\begin{pmatrix}
\I_t \ \Mat{C}_i \ \Mat{C}_i^2 \ \dots \ \Mat{C}_i^{h-1}
\end{pmatrix}
\cdot
\begin{pmatrix}
\vec{x}_1 \
\vec{x}_2 \
\dots \
\vec{x}_h
\end{pmatrix}^{\top}
\in \Fq^t,$$
for $i=1,\dots,r-1$.
\end{construction}
The matrices $\I_t,\Mat{C}_i, \Mat{C}_i^2,\dots,\Mat{C}_i^{h-1}$, $i=1,\dots,r-1$, are the coding coefficients of the incoming and outgoing edges of node $i$ in the middle layer.
\begin{theorem}[\citet{etzion2018vector}]\label{thm:solution-combination}
Construction~\ref{constr:comb-network-gen} provides a vector linear solution of field size $q$ and
dimension $t$ to the $\mathcal{N}_{h,q^t+1,h}$ combination network, i.e., $\vec{x}_1, \dots, \vec{x}_h$ can be reconstructed at all receivers.
\end{theorem}
\begin{proof}
Each receiver $i$ obtains $h$ vectors $\vec{y}_{i_1},\dots, \vec{y}_{i_h}$ and has to solve one of the following two systems of linear equations:
\begin{equation*}
\begin{pmatrix}
\vec{y}^\top_{i_1}\\
\vec{y}^\top_{i_{2}}\\
\vdots\\
\vec{y}^\top_{i_h}
\end{pmatrix}
=
\begin{pmatrix}
\I_t & \Mat{C}_{i_1} & \Mat{C}_{i_1}^2 & \dots & \Mat{C}_{i_1}^{{h-1}}\\
\I_t & \Mat{C}_{i_2} & \Mat{C}_{i_2}^2 & \dots & \Mat{C}_{i_2}^{{h-1}}\\
\vdots & \vdots& \vdots & \ddots & \vdots\\
\I_t & \Mat{C}_{i_h} & \Mat{C}_{i_h}^2 & \dots & \Mat{C}_{i_h}^{{h-1}}\\
\end{pmatrix}.
\begin{pmatrix}
\vec{x}^\top_1\\
\vec{x}^\top_2\\
\vdots\\
\vec{x}^\top_h
\end{pmatrix}
\end{equation*}
or
\begin{equation*}
\begin{pmatrix}
\vec{y}^\top_{i_1}\\
\vdots\\
\vec{y}^\top_{i_{h-1}}\\
\vec{y}^\top_r
\end{pmatrix}
=
\begin{pmatrix}
\I_t & \Mat{C}_{i_1} & \Mat{C}_{i_1}^2 & \dots & \Mat{C}_{i_1}^{{h-1}}\\
\vdots & \vdots& \vdots & \ddots & \vdots\\
\I_t & \Mat{C}_{i_{h-1}} & \Mat{C}_{i_{h-1}}^2 & \dots & \Mat{C}_{i_{h-1}}^{{h-1}}\\
\0_t & \0_t & \0_t & \dots & \I_t
\end{pmatrix}.
\begin{pmatrix}
\vec{x}^\top_1\\
\vec{x}^\top_2\\
\vdots\\
\vec{x}^\top_h
\end{pmatrix},
\end{equation*}
for some distinct $i_1,\dots,i_h \in \{2,\dots,r\}$.
According to Corollary~\ref{cor:q-vand-block-matrix}, in both cases, the corresponding matrix
has full rank and there is a unique solution for $(\vec{x}_1 \ \vec{x}_2 \ \dots \ \vec{x}_h)$.
\end{proof}

\begin{construction}[Construction for the $(1,2)$-$\mathcal{N}_{4,r,5}$ generalized combination network]
\label{constr:mrd-notfull-rank-comb-2}
Let $$\mycode{C}=\{\Mat{C}_1,\Mat{C}_2, \dots, \Mat{C}_{q^{2t^2 + 2t}}\}\subseteq\Fq^{2t\times 2t}$$ be an $\MRD{2t,q^{2t^2+2t}}$ code of minimum rank-distance $t$ and let $r \leq q^{2t^2 + 2t}$.
Consider the $(1,2)$-$\mathcal{N}_{4,r,5}$ network with message vectors $\vec{x}_1, \vec{x}_2, \vec{x}_3, \vec{x}_4 \in \Fq^t$.
The $i$-th middle node transmits:
\begin{equation*}
\begin{pmatrix}
\vec{y}^\top_{i_1} \\
\vec{y}^\top_{i_2}
\end{pmatrix}
=
\begin{pmatrix}
\I_{2t} & \Mat{C}_i
\end{pmatrix}
\cdot
\begin{pmatrix}
\vec{x}^\top_1\\
\vec{x}^\top_2\\
\vec{x}^\top_3\\
\vec{x}^\top_4
\end{pmatrix}
\in \Fq^{2t},
\quad
i=1,\dots,r.
\end{equation*}
The direct link from the source which ends in the same receiver as the links from two distinct middle nodes $i,j \in \{ 1,2,\ldots , r\}$ transmits the vector $$\vec{z}_{ij} = \Mat{P}_{ij} \cdot \begin{pmatrix}
\vec{x}_1,
\vec{x}_2,
\vec{x}_3,
\vec{x}_4
\end{pmatrix}^\top\in \Fq^t,$$
where the matrix $\vec{P}_{ij}\in \Fq^{t \times 4t}$ is chosen such that
\begin{equation}\label{eq:mat-constr1}
\rk\begin{pmatrix}
\I_{2t} & \Mat{C}_i \\
\I_{2t} & \Mat{C}_j \\
&\hspace{-4ex} \Mat{P}_{ij}
\end{pmatrix} = 4t.
\end{equation}
\end{construction}
Since $$\rk\left(\begin{matrix}
\I_{2t} & \Mat{C}_i \\
\I_{2t} & \Mat{C}_j \\
\end{matrix}\right)  =
\rk\left(\begin{matrix}
\I_{2t} & \Mat{C}_i \\
\0_{2t} & \Mat{C}_j-\Mat{C}_i \\
\end{matrix}\right)
\geq 3t,$$
it follows that the $t$ rows of $\Mat{P}_{ij}$ can be chosen such that the overall rank of the matrix from~\eqref{eq:mat-constr1} is $4t$.
\begin{theorem}
\label{thm:solution-extra-rate2}
Construction~\ref{constr:mrd-notfull-rank-comb-2} provides a vector solution to the $(1,2)$-$\mathcal{N}_{4,r,5}$ network with messages of length $t$ over $\Fq$ for any $ r\leq q^{2t(t+1)}$.
\end{theorem}
\begin{proof} %
Each receiver $R_{ij}$ (i.e., the receiver node connecting to the $i$-th and $j$-th middle node) obtains the vectors $\vec{y}_{i_1}, \vec{y}_{i_2},
\vec{y}_{j_1},\vec{y}_{j_2}\in\Fq^t$ and the vector ${\vec{z}_{ij}}$ from the direct link.
From these five vectors, the receiver wants to reconstruct the message vectors
$\vec{x}_1,\vec{x}_2,\vec{x}_3,\vec{x}_4$ by solving the following linear system of equations:
\begin{equation*}
  (\vec{y}_{i_1}, \vec{y}_{i_2},
\vec{y}_{j_1},\vec{y}_{j_2})^\top
=
\begin{pmatrix}
\I_{2t} & \Mat{C}_i\\
\I_{2t} & \Mat{C}_j\\
& \hspace{-4ex}\Mat{P}_{ij}
\end{pmatrix}
\cdot
(\vec{x}_1,\vec{x}_2,\vec{x}_3,\vec{x}_4)^{\top}
\end{equation*}
The choice of $\Mat{P}_{ij}$ from Construction~\ref{constr:mrd-notfull-rank-comb-2} guarantees that this linear system of equations
has a unique solution for $(\vec{x}_1, \vec{x}_2,\vec{x}_3,\vec{x}_4)$.
\end{proof}
Construction~\ref{constr:mrd-notfull-rank-comb-2} can be further generalized to any $(\eps,\ell)$-$\mathcal{N}_{2\ell,r,2\ell+\eps}$ networks with $\alpha=2$ and $\eps\geq \ell-1$. In particular, we include the construction for the $(\ell-1,\ell)$-$\mathcal{N}_{2\ell,r,3\ell-1}$ generalized combination network.
\begin{construction}[For $(\ell-1,\ell)$-$\mathcal{N}_{2\ell,r,3\ell-1}$ network]
\label{constr:mrd-notfull-rank-comb-2-gen2}
Let $$\mycode{C}=\{\Mat{C}_1,\Mat{C}_2,\dots,\Mat{C}_{q^{\ell\eps t^2 + \ell t}}\}$$ be an $\MRD{\ell t,q^{\ell(\ell-1)t^2+\ell t}}$ code of minimum rank-distance $t$ and let $r$ be any integer such that $r \leq q^{\ell(\ell-1) t^2 + \ell t}$.
Consider the $(\ell-1,\ell)$-$\mathcal{N}_{2\ell,r,3\ell-1}$ network with message vectors $\vec{x}_1, \dots, \vec{x}_{2\ell} \in \Fq^t$, where $\ell \geq 2$.
The $i$-th middle node transmits:
\begin{equation*}
(\vec{y}_{i_1}, \vec{y}_{i_2}
)^\top
=
\begin{pmatrix}
\I_{\ell t} & \Mat{C}_i
\end{pmatrix}
\cdot
(
\vec{x}_1,
\dots, \vec{x}_{2\ell}
)^\top
\in \Fq^{\ell t},
\quad
i=1,\dots,r.
\end{equation*}
The $\ell-1$ direct links from the source, which end at the same receiver as
the links from two distinct nodes $i,j \in \{ 1,2,\ldots , r\}$
of the middle layer, transmit the vectors
$$\vec{z}_{ijs} = \Mat{P}_{ijs} \cdot \begin{pmatrix}
\vec{x}_1,
\dots, \vec{x}_{2\ell}
\end{pmatrix}^\top\in \Fq^t,\ s=1,\dots,\ell-1$$
where the $t \times (2\ell t)$ matrices $\vec{P}_{ijs}$ are chosen such that
\begin{equation}\label{eq:matrix-nplus-network}
\rk\begin{pmatrix}
\I_{\ell t} & \Mat{C}_i \\
\I_{\ell t} & \Mat{C}_j \\
&\hspace{-4ex} \Mat{P}_{ij1}\\
&\hspace{-4ex} \vdots\\
&\hspace{-4ex} \Mat{P}_{ij(\ell-1)}
\end{pmatrix} = 2\ell t.
\end{equation}
\end{construction}
By the rank distance of $\mycode{C}$ we have that
$\rk\left(\begin{smallmatrix}
\I_{\ell t} & \Mat{C}_i \\
\I_{\ell t} & \Mat{C}_j \\
\end{smallmatrix}\right) \geq \ell t + t = (\ell+1)t$, and hence the $(\ell-1)t$
rows of the matrices~$\Mat{P}_{ijs}$ can be chosen such that the overall rank
of the matrix from \eqref{eq:matrix-nplus-network} is~$2\ell t$.

The following result is an immediate consequence of this construction.
\begin{corollary}
\label{thm:solution-extra-rate2-gen2}
Construction~\ref{constr:mrd-notfull-rank-comb-2-gen2} provides a vector solution of field size $q$
and dimension $t$ to the $(\ell-1,\ell)$-$\mathcal{N}_{2\ell,r,3\ell-1}$ network
for any $r \leq q^{\ell(\ell-1) t^2 + \ell t}$ with $2\ell$ messages for any $\ell \geq 2$.
\end{corollary}
\section{Error Control in (Random) Linear Network Coding}\label{sec:error-model}
In the following, the constructions are based on the notion of \emph{subspaces}. Therefore we first introduce the following
notations related to subspaces used throughout the remainder of this chapter.

Given a subspace $\myspace{V}\in\ProjspaceAny{\dimAmbSpace}$, its orthogonal subspace is defined as
\begin{align}\label{eq:defOrthSpace}
  \myspaceDual{\myspace{V}}\coloneqq \{\u\ :\ \v\cdot\u=\0\ \forall\ \v\in\myspace{V}\}
\end{align}

The \emph{subspace distance} between two subspaces $\myspace{U},\myspace{V}$ in $\ProjspaceAny{\dimAmbSpace}$ is defined as
\begin{align}\label{eq:subspaceDistance}
  \Subspacedist{\myspace{U},\myspace{V}}&\coloneqq\dim(\myspace{U}+\myspace{V})-\dim(\myspace{U}\cap \myspace{V}) \nonumber \\
    &\overset{\phantom{def}}{=}\dim(\myspace{U})+\dim(\myspace{V})-2\dim(\myspace{U}\cap \myspace{V})\ .
\end{align}

A different metric on $\ProjspaceAny{\dimAmbSpace}$ is the \emph{injection distance}~\citep{SilvaKschischang-MetricsErrorCorrectionNetworkCoding_2009,Silva_PhD_ErrorControlNetworkCoding} which is defined as
\begin{equation}\label{eq:injectionDistance}
 \Injectiondist{\myspace{U},\myspace{V}}\coloneqq\max\{\dim(\myspace{U}),\dim(\myspace{V})\}-\dim(\myspace{U}\cap\myspace{V}).
\end{equation}

We now give a more detailed description of the  channel models of (random) linear network coding.
\subsection{Matrix Channel}
Suppose $K$ packets $\vec{x}_i\in\Fq^{\dimAmbSpace}$ for $i\in\intervallincl{0}{K-1}$ of length $\dimAmbSpace$ are transmitted from a source to the receivers.
Let the matrix $\Mat{X}\in\Fq^{K\times\dimAmbSpace}$ contain the transmitted packets $\vec{x}_i\in\Fq^{\dimAmbSpace}$ for $i\in\intervallincl{0}{K-1}$ as rows.
The in-network linear combinations from the source to one sink can be modeled by
  \begin{equation*}
   \vec{y}_j = \sum_{i=0}^{K-1}a_{j,i}\vec{x}_i,
   \quad\forall j\in\intervallincl{0}{M-1}
   \quad\Longleftrightarrow\quad
   \Mat{Y}=\Mat{A}\Mat{X}
  \end{equation*}
where the network matrix $\Mat{A}$ depends on the network topology as well as the coefficients of the linear combinations performed in the network.
If $\Mat{A}$ is constant and known at the receiver we call the scenario \emph{coherent network coding}.
If $\Mat{A}$ is not known at the receiver we have \emph{noncoherent} (or ``\emph{channel oblivious}'') network coding~\citep{ChouWuJain-PracticalNetworkCoding_2003,koetter_kschischang}.
A scheme for noncoherent network coding was presented by~\citet{ChouWuJain-PracticalNetworkCoding_2003}.
The idea is to append an identity matrix to the packet matrix $\Mat{X}$ such that if $\Mat{A}$ has full rank $K$ the transmitted packets in $\Mat{X}$ can be recovered from the received matrix $\Mat{Y}$ by Gaussian elimination.

In real networks the network matrix $\Mat{A}$ can have smaller rank than $K$ due to an insufficient number of links from the source to a sink or erased packets due to link failures or an unfortunate choice of the coding coefficients.
Additionally, noisy links may corrupt symbols or entire packets.
A single malicious or lost packet in the received matrix $\Mat{Y}$ causes the scheme by~\citet{ChouWuJain-PracticalNetworkCoding_2003} to fail.

A \emph{matrix channel} model~\citep{koetter_kschischang,silva_rank_metric_approach,Silva_PhD_ErrorControlNetworkCoding} incorporating these types of errors is given by
  \begin{equation}\label{eq:networkCodingMatrixChannel}
   \vec{y}_j = \sum_{i=0}^{K-1}a_{j,i}\vec{x}_i+\sum_{t=0}^{T-1}d_{j,t}\vec{z}_t,
   \quad\forall j\in\intervallincl{0}{M-1}
   \quad\Longleftrightarrow\quad
   \Mat{Y}=\Mat{A}\Mat{X}+\Mat{D}\Mat{Z}
  \end{equation}
  where $\Mat{D}\in\Fq^{M\times T}$ and the matrix $\Mat{Z}\in\Fq^{T\times\dimAmbSpace}$ contains $T$ erroneous packets $\vec{z}_t\in\Fq^{\dimAmbSpace}$ for $t\in\intervallincl{0}{T-1}$ as rows.

The challenge of error control in network coding is that linear combinations with erroneous packets again result in an erroneous packet which makes errors propagate through the entire network.
The problem of error propagation is illustrated in Figure~\ref{fig:errorProp}.

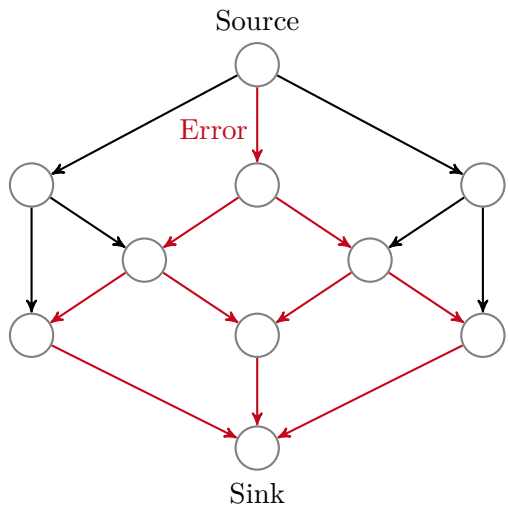
\begin{figure}[t!]
    \centering
    \tikzset{
  treenode/.style = {align=center, inner sep=0pt, text centered},
  corr/.style = {treenode, circle,  gray, draw= gray, text width=1.5em, thick},%
  err/.style = {treenode, circle, red, draw=TUMred, text width=1.5em, thick},%
  link/.style={->,draw=black, thick},
  error/.style={->,TUMred,draw=TUMred, thick}
}

\begin{tikzpicture}[->,>=stealth'] 
\node[corr] (source) at (0,0.1){};
\node[corr] (1a) at (-3,-1.5){};
\node[corr] (1b) at (0,-1.5){};
\node[corr] (1c) at (3,-1.5){};

\node[corr] (2a) at (-1.5,-2.5){};
\node[corr] (2b) at (1.5,-2.5){};

\node[corr] (3a) at (-3,-3.5){};
\node[corr] (3b) at (0,-3.5){};
\node[corr] (3c) at (3,-3.5){};

\node[corr] (sink) at (0,-5){};

\node[TUMred] at (-0.58,-0.75) {Error};
\node[black] at (0,0.7) {Source};
\node[black] at (0,-5.6) {Sink};

\path (source)  edge [link]   	(1a);
\path (source)  edge [error]   	(1b);
\path (source)  edge [link]   	(1c);

\path (1a)  edge [link]   	(2a);
\path (1a)  edge [link]   	(3a);
\path (1b)  edge [error]   	(2a);
\path (1b)  edge [error]   	(2b);
\path (1c)  edge [link]   	(2b);
\path (1c)  edge [link]   	(3c);

\path (2a)  edge [error]   	(3a);
\path (2a)  edge [error]   	(3b);
\path (2b)  edge [error]   	(3b);
\path (2b)  edge [error]   	(3c);

\path (3a)  edge [error]   	(sink);
\path (3b)  edge [error]   	(sink);
\path (3c)  edge [error]   	(sink);

\end{tikzpicture}
    \caption{Error propagation in RLNC. A single link error can corrupt the entire transmission.}\label{fig:errorProp}
\end{figure}

Error control for coherent network coding, i.e., if $\Mat{A}$ is known at the transmitter and the receiver, was considered by~\citet{NWCerrorCorrectionCai02,yeung2006errorCorrectionPart1,cai2006errorCorrectionPart2}.
These schemes select the network coding coefficients such that they can be used for error correction.

\subsection{Operator Channel}
In 2008, Kötter and Kschischang observed that in the error-free case for noncoherent RLNC (i.e., $\Mat{Z}=\Mat{0}$ and $\rk_q(\Mat{A})=K$) the row space of the transmitted packet matrix is preserved by the random and unknown $\Fq$-linear combinations in the network, i.e., we have that $\Rowspace{\Mat{A} \ \Mat{X}}=\Rowspace{\Mat{X}}$~\citep{koetter_kschischang}.
This motivated their idea to consider the row space of the transmitted and received matrices $\Mat{X}$ and $\Mat{Y}$.
The row space $\Rowspace{\Mat{Y}}$ of the received $\Mat{Y}$ in~\eqref{eq:networkCodingMatrixChannel} can be decomposed into a direct sum
  \begin{equation}\label{eq:networkCodingMatrixChannelRowSpace}
   \Rowspace{\Mat{Y}}=\tilde{\myspace{X}}\oplus\myspace{Z} ,
  \end{equation}
  where $\tilde{\myspace{X}}=\Rowspace{\Mat{X}}\cap\Rowspace{\Mat{Y}}$ is a subspace of $\Rowspace{\Mat{X}}$ and $\myspace{Z}$ is an error space that intersects trivially with $\Rowspace{\myspace{\Mat{X}}}$ (i.e., $\myspace{Z}\cap \Rowspace{\Mat{X}}=\0$).
  Motivated by the decomposition~\eqref{eq:networkCodingMatrixChannelRowSpace}, \citet{koetter_kschischang} proposed a channel model that abstracts the linear network coding channel on the packet level~\eqref{eq:networkCodingMatrixChannel} to the subspace level, i.e., the linear vector spaces spanned by the transmitted and received packets.
The \emph{operator channel} \citep[Definition~1]{koetter_kschischang} is a discrete memoryless channel that has input and output from an alphabet $\ProjspaceAny{\dimAmbSpace}$.
The output $\myspace{U}$ is related to the input $\myspace{V}$ with $\nTransmit\coloneqq\dim(\myspace{V})$ by
\begin{equation}\label{eq:operatorChannel}
 \myspace{U}=\delOp{\nTransmit-\deletions}(\myspace{V})\oplus \myspace{E} ,
\end{equation}
where $\delOp{\nTransmit-\deletions}(\myspace{V})$ returns a random $(\nTransmit-\deletions)$-dimensional subspace of $\myspace{V}$, and $\myspace{E}$ denotes an error space of dimension~$\insertions$ (see Figure~\ref{fig:opChannel}). We call $\deletions$ the number of \emph{deletions} and $\insertions$ the number of \emph{insertions}.
We consider the worst case that $\myspace{V}$ and $\myspace{E}$ intersect trivially, i.e., $\myspace{V}\cap\myspace{E}=\{\vec{0}\}$, since otherwise vectors that are contained in $\myspace{V}$ and $\myspace{E}$ but are not contained in $\delOp{\nTransmit-\deletions}(\myspace{V})$ might decrease the number of observed deletions at the channel output~\citep{koetter_kschischang}.
\begin{figure}[t!]
    \centering
    \begin{tikzpicture}[auto]
	\node (a) at (0,0) {};
	\node [draw, inner sep=15pt,thick] (channel) at (130pt,0) {$\delOp{\nTransmit-\deletions}(\myspace{V})\oplus\myspace{E}$};
	\node [coordinate,thick] (end) at (250pt,0){};
	\path[->,thick] (a) edge node {$\myspace{V}\in\ProjspaceAny{\dimAmbSpace}$} (channel);
	\path[->,thick] (channel) edge node {$\myspace{U}\in\ProjspaceAny{\dimAmbSpace}$} (end);
\end{tikzpicture}
    \caption{The operator channel~\citep{koetter_kschischang}.}\label{fig:opChannel}
\end{figure}
Thus the noncorrupted subspace is $\myspace{U}\cap\myspace{V}=\delOp{\nTransmit-\deletions}(\myspace{V})$ and we have
\begin{equation}\label{eq:constraintInsertionsDeletions}
  \dim(\myspace{U}\cap\myspace{V})\geq0
  \quad\Longleftrightarrow\quad
  \nTransmit\geq\deletions
  \quad\text{and}\quad
  \insertions=\dim(\myspace{E})\leq\dimAmbSpace-\nTransmit.
\end{equation}

The relation between the input and output of the operator channel is illustrated in Figure~\ref{fig:SubspaceChannel}.
\begin{figure}[ht!]
\centering
  \def\transSpace{(0,0) rectangle (3,1.5);}
\def\recSpace{(1,0) rectangle (4,1.5);}
\def\intSpace{(1,0) rectangle (3,1.5);}
\def\errSpace{(3,0) rectangle (4,1.5);}

\begin{tikzpicture}

\draw   [fill=TUMblue, blueskydeep] \transSpace;
\fill   [fill=TUMblue, blueskydeep, pattern=north west lines] \transSpace;
\draw   [fill=TUMgreen!50!white, TUMgreen!70!white] \intSpace;
\draw   [fill=lightcoral, lightcoral] \errSpace;

\node at (2, 0.75) {$\myspace{V} \cap \myspace{U}$};
\node[left] at (0, 0.75) {$\myspace{V}$};
\node[left] at (4.6, 0.75) {$\myspace{U}$};
\node[right] at (3.25, 0.75) {$\myspace{E}$};

\path [<->,thick,black] (0,-0.25) edge node[below] {$\deletions$} (1,-0.25);
\path [<->,thick,black] (1,-0.25) edge node[below] {$\nTransmit-\deletions$} (3,-0.25);
\path [<->,thick,black] (3,-0.25) edge node[below] {$\insertions$} (4,-0.25);

\end{tikzpicture}
  \caption{Illustration of the relation between the input and output of the operator channel~\eqref{eq:operatorChannel}.
  The transmitted space is $\myspace{V}$ (blue and green) and the received space is $\myspace{U}=(\myspace{V}\cap\myspace{U})\oplus\myspace{E}$ (green and red).
  The green intersection space is returned by $\delOp{\nTransmit-\deletions}(\myspace{V})$ and the red error space is $\myspace{E}$.
}
  \label{fig:SubspaceChannel}
\end{figure}

The distribution of $\delOp{\nTransmit-\deletions}(\myspace{V})$ does not affect the performance of the code and can be chosen to be uniform (see~\citep{koetter_kschischang}). %
Any type of errors occurring in the matrix channel~\eqref{eq:networkCodingMatrixChannel} can be modeled by the operator channel~\eqref{eq:operatorChannel} and vice versa.
The dimension of the output subspace $\myspace{U}$ is then %
\begin{equation*}
  \nReceive\coloneqq\dim(\myspace{U})=\dim\left(\delOp{\nTransmit-\deletions}(\myspace{V})\right)+\dim(\myspace{E})=\nTransmit-\deletions+\insertions.
\end{equation*}
The subspace distance between the input subspace $\myspace{V}$ and the output subspace $\myspace{U}$ is
\begin{align}\label{eq:subDistOpChannel}
 \Subspacedist{\myspace{V},\myspace{U}}&=\dim(\myspace{V})+\dim(\myspace{U})-2\dim(\myspace{V}\cap\myspace{U}) \nonumber  \\
    &=\nTransmit+\nReceive-2(\nTransmit-\deletions) \nonumber
    \\
    &=\insertions+\deletions.
\end{align}

We further relate the input $\myspace{V}$ and the output  $\myspace{U}$ with the following definition.
\begin{definition}[$(\insertions,\deletions)$-Reachability]
 We say that a subspace $\myspace{V}$ is \emph{$(\insertions,\deletions)$-reachable} from a subspace $\myspace{U}$ if there exists a realization of the operator channel~\eqref{eq:operatorChannel} with $\insertions$ insertions and $\deletions$ deletions that transforms the input $\myspace{V}$ to the output $\myspace{U}$.
\end{definition}
If a space $\myspace{V}$ is $(\insertions,\deletions)$-reachable from a space $\myspace{U}$, then we have that $\Subspacedist{\myspace{U},\myspace{V}}=\insertions+\deletions$.

The operator channel is of particular interest to evaluate the performance of decoders for \emph{subspace codes} (cf.~Section~\ref{sec:subspace-codes}).
Later we present decoding schemes that can correct insertions and deletions beyond the unique decoding region by allowing a very small decoding failure probability.
The decoding failure probability depends on the number of insertions $\insertions$  and deletions $\deletions$.
We use the operator channel to validate the upper bounds for the decoding failure probability for particular values of $\insertions$ and $\deletions$.

\section{Subspace Codes}
\label{sec:subspace-codes}
Subspace codes have been proposed for error control for noncoherent RLNC, e.g., when the network topology and the random in-network linear combinations are not known or used by the transmitter and the receiver~\citep{koetter_kschischang, silva_rank_metric_approach}.
This idea was motivated by \emph{Grassmann} codes in the field of complex numbers that are used for multiple antenna channels~\citep{Tse2002GrassmCodes}.
Constructions of subspace codes from Gabidulin codes were proposed in 2003 as linear authentication codes~\citep{wang2003linear}.
K\"otter and Kschischang revisited the Reed--Solomon like construction by~\citet{wang2003linear} in the context of error correction in RLNC and proposed a suitable metric for subspaces~\citep{koetter_kschischang}.

\begin{definition}[Subspace Code]\label{def:subspaceCode}
 A subspace code $\CSub$ is a non\-empty subset of $\ProjspaceAny{\dimAmbSpace}$.
\end{definition}

The minimum distance $\Subspacedist{\CSub}$ of a subspace code $\CSub$ is defined as
\begin{equation}
  \Subspacedist{\CSub}\coloneqq \min_{\myspace{V},\myspace{V}'\in\CSub,\myspace{V}\neq\myspace{V}'}\Subspacedist{\myspace{V},\myspace{V}'}.
\end{equation}

\subsection{Constant-Dimension Subspace Codes}

An important class of subspace codes are constant-dimension subspace codes, which are defined as follows.

\begin{definition}[Constant-Dimension Subspace Code]\label{def:CDsubspaceCode}
 A constant-dimension subspace code of dimension $\nTransmit$ is a nonempty subset of $\Grassm{\dimAmbSpace,\nTransmit}$.
\end{definition}
The relation between the subspace distance~\eqref{eq:subspaceDistance} and the injection distance~\eqref{eq:injectionDistance} of a constant-dimension subspace code $\CSub$ is (see~\citep{Silva_PhD_ErrorControlNetworkCoding})
\begin{equation}\label{eq:relationInjSubDist}
  \Subspacedist{\CSub}=2\Injectiondist{\CSub}.
\end{equation}
In the following we consider constant-dimension codes only and use the subspace distance as a metric.
All results can be expressed in terms of the injection distance by using~\eqref{eq:relationInjSubDist}.

Constructions of constant-dimension subspace codes were first considered by~\citet{wang2003linear} for linear authentication codes, later applied to RLNC by~\citet{koetter_kschischang, silva_rank_metric_approach, Silva_PhD_ErrorControlNetworkCoding}.
The code rate of a constant-dimension subspace code $\CSub\subseteq\Grassm{\dimAmbSpace,\nTransmit}$ is defined as
\begin{equation}\label{eq:rateSubspaceCode}
  R=\frac{\log_q(|\CSub|)}{\dimAmbSpace\nTransmit}.
\end{equation}
For a constant-dimension subspace code $\CSub\subseteq\Grassm{\dimAmbSpace,\nTransmit}$ with minimum distance $\Subspacedist{\CSub}$, the \emph{complementary} code is define as the set of all orthogonal subspaces (see~\eqref{eq:defOrthSpace})
\begin{equation}\label{eq:defComplementaryCode}
  \CSubComp\coloneqq\left\{\myspaceDual{V}:\myspace{V}\in\CSub\right\}.
\end{equation}
The complementary code is a constant-dimension code $\CSubComp\subseteq\Grassm{\dimAmbSpace,\dimAmbSpace-\nTransmit}$ of size $|\CSubComp|=|\CSub|$, minimum distance $\Subspacedist{\CSubComp}=\Subspacedist{\CSub}$ \citep{koetter_kschischang} and code rate
\begin{equation}\label{eq:rateComplementarySubspaceCode}
  R^{\perp}=\frac{\log_q(|\CSub|)}{\dimAmbSpace(\dimAmbSpace-\nTransmit)}=\frac{\nTransmit}{\dimAmbSpace-\nTransmit}R.
\end{equation}
Hence, we consider only codes with $\nTransmit\leq\dimAmbSpace/2$ since for each code with $\nTransmit>\dimAmbSpace/2$ there exists a complementary code with $\nTransmit<\dimAmbSpace/2$ that has the same minimum distance and a higher code rate.

\subsection{Lifted Rank-Metric Codes}
It was shown by~\citet{silva_rank_metric_approach} that constant-dimension subspace codes can be obtained by \emph{lifting} rank-metric codes.
Lifted MRD codes are ``near-optimal''~\citep[Theorem~4.24]{Silva_PhD_ErrorControlNetworkCoding}.
The lifting operation~\citep{silva_rank_metric_approach,Silva_PhD_ErrorControlNetworkCoding} appends the identity matrix to each rank-metric codeword and considers the row space of the resulting augmented matrix as a subspace codeword.

To describe a large variety of constant-dimension subspace codes that are constructed from rank-metric codes, we define a \emph{generalized lifting} operation.

\begin{definition}[Generalized Lifting]\label{def:generalizedLifting}
 Let $\Mat{A}\in\Fq^{\nTransmit\times\nTransmit}$ with $\rk_q(\Mat{A})=\nTransmit$ and let $\Mat{C}\in\Fq^{\nTransmit\times M}$.
 Define the map $\lifting:\Fqm^{\nTransmit\times M}\mapsto\Grassm{\nTransmit+M,\nTransmit}$
\begin{equation}
   \Mat{C}\mapsto\liftingMap{\Mat{A}}{\Mat{C}}=\Rowspace{\left(\Mat{A} \ \Mat{C}\right)}.
\end{equation}
The subspace $\liftingMap{\Mat{A}}{\Mat{C}}\in\Grassm{\nTransmit+M,\nTransmit}$ is called an $\Mat{A}$-lifting of $\Mat{C}$.
Given a matrix code $\mycode{C}\subseteq\Fq^{\nTransmit\times M}$ the corresponding $\Mat{A}$-lifted code $\liftingMap{\Mat{A}}{\mycode{C}}\subseteq\Grassm{\nTransmit+M,\nTransmit}$ is defined as
\begin{align*}
 \liftingMap{\Mat{A}}{\mycode{C}}\coloneqq \left\{\liftingMap{\Mat{A}}{\Mat{C}}:\Mat{C}\in\mycode{C}\right\}.
\end{align*}
\end{definition}

Let $\Gab{n_t,k}\subseteq\Fq^{\nTransmit\times \nTransmit}$ be a Gabidulin code (see Definition~\ref{def:GabCode}) with evaluation points $\vecalpha$.
There are two common choices for the lifting matrix $\Mat{A}$.

The interpolation-based decoding scheme~\citep{koetter_kschischang} uses $\Mat{A}=\extsmallfieldinput{\vecalpha}$ to construct the subspace code $\liftingMap{\Mat{A}}{\Gab{n,k}}$.
This construction is beneficial for interpolation-based decoding schemes since the basis vectors contain the code locators and thus can be used directly for decoding.
We call codes of this form \emph{locator-lifted} rank-metric codes.

In the syndrome-based approach~\citep{silva_rank_metric_approach} an identity matrix is used to construct subspace codes of the form $\liftingMap{\Mat{I}_{n_t}}{\Gab{n_t,k}}$.
For syndrome-based decoding schemes a canonical form of the received basis is required which can be obtained easily by Gaussian elimination if the code is lifted using an identity matrix.
Thus we call constructions of this form \emph{identity-lifted} rank-metric codes. %

For identity-lifted rank-metric codes the subspace distance of the lifted code is twice the rank distance of the rank-metric code~\citep{silva_rank_metric_approach,Silva_PhD_ErrorControlNetworkCoding}.
The following lemma extends this result to $\Mat{A}$-lifted codes.
\vspace*{5pt}
\begin{lemma}[Subspace Distance of $\Mat{A}$-Lifted Codes]\label{lem:distanceAliftedCodes}
 Let $\CRank\subset\Fq^{\nTransmit\times m}$ be a rank-metric code of length $\nTransmit$ and minimum distance $\Rankdist{\CRank}$ over the field $\Fq$.
 Let $\Mat{A}\in\Fq^{\nTransmit\times\nTransmit}$ be nonsingular.
 Then the $\Mat{A}$-lifted subspace code $\CSub=\liftingMap{\Mat{A}}{\CRank}$ has minimum subspace distance $\Subspacedist{\CSub}=2\Rankdist{\CRank}$.
\end{lemma}

\begin{proof}
 We have
 \begin{equation*}
  \Subspacedist{\CSub}=\liftingMap{\Mat{A}}{\CRank}=\liftingMap{\Mat{I}}{\Mat{A}^{-1}\CRank}.
 \end{equation*}
 Using Proposition~4 from \citet{silva_rank_metric_approach} we obtain
 \begin{equation*}
  \Subspacedist{\CSub}
  =\liftingMap{\Mat{I}}{\Mat{A}^{-1}\CRank}
  =2\Rankdist{\Mat{A}^{-1}\CRank}
  =2\Rankdist{\CRank}.
 \end{equation*}
\end{proof}
A well-known method to obtain a parity check matrix from a systematic generator matrix of a code~\citep[p.~55]{LinCostello-ErrorControlCoding_2004} can be used to construct the complementary code~\eqref{eq:defComplementaryCode} of an identity-lifted rank-metric code.

\begin{proposition}[Complementary Code of Identity-Lifted Code]
  Let $\CRank\subset\Fq^{\nTransmit\times m}$ be a rank-metric code.
  Let $\CSub=\liftingMap{\Mat{I}}{\CRank}$ be the corresponding identity-lifted rank-metric code in $\Grassm{\dimAmbSpace,\nTransmit}$ with $\dimAmbSpace=\nTransmit+m$.
  Then the complementary code $\CSubComp$ with dimension $\nTransmit^\perp=\dimAmbSpace-\nTransmit$ can be constructed from $\CRank$ by
  \begin{equation}
    \CSubComp=\left\{\Rowspace{\left(-\Mat{C}^\top \ \Mat{I}_{\nTransmit^{\perp}}\right)}:\Mat{C}\in\CRank\right\}.
  \end{equation}
\end{proposition}

\subsection{Interleaved Subspace Codes}\label{subsec:interleavedSubspace}

Identity-lifted \emph{interleaved} Gabidulin codes were considered by~\citet{silva_rank_metric_approach} to reduce the rate-loss due to the lifting and the computational complexity rather than increasing the decoding region.
Decoding schemes for identity-lifted interleaved Gabidulin codes with an improved decoding region were proposed by~\citet{SidBoss_InterlGabCodes_ISIT2010,Li2014transformDomainDCC}.
Locator-lifted interleaved Gabidulin codes - further called interleaved subspace codes - were first considered by~\citet{BartzWachterInterleavedSubspace2014}.

\begin{definition}[$\intOrder$-Interleaved Subspace Code]\label{def:interleaved_subspace}
 Let $\vecalpha=\left(\alpha_0 \ \alpha_1 \ \dots \alpha_{\nTransmit-1}\right)^\top$ with $\nTransmit\leq m$ be a vector containing $\Fq$-linearly independent code locators from $\Fqm$.
 For fixed integers $k^{(1)},\dots,k^{(\intOrder)} \leq \nTransmit$, an interleaved subspace code $\IntSub{\intOrder; \nTransmit, k^{(1)},\dots,k^{(\intOrder)}}$ of dimension $\nTransmit$ and interleaving order $\intOrder$ is defined as
 \begin{equation}\label{eq:DefIntSubspaceCode}
  \left\{
  \Rowspace{\left(\vecalpha \ f^{(1)}(\vecalpha) \ f^{(2)}(\vecalpha) \ \dots \ f^{(\intOrder)}(\vecalpha)\right)}:
  f^{(j)}(x)\in\Linpolyring_{<k^{(j)}}, \forall j\in\intervallincl{1}{\intOrder}
  \right\}.
 \end{equation}
\end{definition}
Let $\Mat{A}=\extsmallfieldinput{\vecalpha}$. $\intOrder$-interleaved subspace codes are $\Mat{A}$-lifted interleaved Gabidulin codes $\IntGabcode{\intOrder;n, k^{(1)},\dots,k^{(\intOrder)}}$ with code locators $\vecalpha$ (see Definition~\ref{def:IntGabCode}),
i.e., we have
\begin{equation*}
 \IntSub{\intOrder;\nTransmit,k^{(1)},\dots,k^{(\intOrder)}}
 =\liftingMap{\Mat{A}}{ \IntGabcode{\intOrder;\nTransmit,k^{(1)},\dots, k^{(\intOrder)} }^\top} ,
\end{equation*}
where $\IntGab^\top\subset \Fqm^{n_t\times L}$ is obtained by transposing all codewords of $\IntGab\subset \Fqm^{L\times n_t}$.

If $k^{(j)}=k,\forall j\in\intervallincl{1}{\intOrder}$ we call the code a \emph{homogeneous} interleaved subspace code and denote it by $\IntSub{\intOrder,\vecalpha; \nTransmit, k}$.
The basis vectors in~\eqref{eq:DefIntSubspaceCode} are of the form
\begin{equation*}
 \left(\alpha,\beta^{(1)},\dots, \beta^{(\intOrder)}\right)\text{ with }\alpha \in \Rowspace{\vecalpha}, \beta^{(1)},\dots,\beta^{(\intOrder)} \in \Fqm
\end{equation*}
and are expanded over the field $\Fq$ before transmission.
The ambient vector space is
 \begin{align*}
 \ambSpace&= \Rowspace{\vecFq{\vecalpha}} \times \underbrace{\Fq^{m}\times \dots \times \Fq^m}_{\intOrder\text{ times}}
 \end{align*}
 with dimension $\dim(\ambSpace)=\dimAmbSpace=\nTransmit+\intOrder m$.

The code rate of $\IntSub{\intOrder; \nTransmit, k^{(1)},\dots,k^{(\intOrder)}}$ is

\begin{equation}\label{eq:codeRateISub}
 R=\frac{\log_q\left(|\IntSub{\intOrder; \nTransmit, k^{(1)},\dots,k^{(\intOrder)}}|\right)}{\nTransmit\dimAmbSpace}
 =\frac{m\sum_{j=1}^\intOrder k^{(j)}}{\nTransmit(\nTransmit+\intOrder m)}.
\end{equation}

For increasing interleaving order $\intOrder$ the rate loss caused by the appended code locators decreases since $\nTransmit<<\intOrder m$.
For $\intOrder=1$ (no interleaving) the codes from Defintion~\ref{def:interleaved_subspace} are equivalent to Kötter-Kschischang subspace codes~\citep{koetter_kschischang}.

\begin{proposition}[Minimum Distance of $\intOrder$-Interleaved Subspace Codes]\label{lem:IntSubDistance}
 The minimum sub\-space distance of an $\intOrder$-interleaved subspace code $\IntSub{\intOrder,\vecalpha; \nTransmit, k^{(1)},\dots,k^{(\intOrder)}}$ as in Definition~\ref{def:interleaved_subspace} is
 \begin{equation*}
  \Subspacedist{\mathcal{IS}}
  = 2\left(\nTransmit-\max_{j\in\intervallincl{1}{\intOrder}}\left\{k^{(j)}\right\}+1\right).
 \end{equation*}
\end{proposition}

\subsection{Decoding of Interleaved Subspace Codes}

Decoding of interleaved subspace codes was considered in~\citep{bartz2017algebraic,bartz2018efficient}.
The algorithm consists of two steps: the \emph{interpolation step} computes $r\leq\intOrder$ non-zero and (left) $\Linpolyring$-independent vectors of linearized polynomials
\begin{equation*}
\vec{Q}^{(i)}\!=\![Q_0^{(i)},Q_1^{(i)},\dots,Q_\intOrder^{(i)}] \in \Linpolyring^{\intOrder+1} \setminus \{\0\}, \ \forall i=1,\dots,r
\end{equation*}
such that they fulfill certain degree and evaluation conditions with respect to the received space.
The \emph{root-finding step} finds all message polynomials $f_j$ of degrees $\deg f_j < k$ such that
\begin{equation*}
Q_0^{(i)} + \sum_{j=1}^{u} Q_j^{(i)} f_j = 0 \quad \forall \, i=1,\dots,r.
\end{equation*}
If the number of insertions $\insertions$ and deletions $\deletions$ satisfies $\insertions+\intOrder\deletions<\intOrder(\nTransmit-k+1)$, then at least one satisfactory interpolation vector $\vec{Q}^{(i)}$ exists, see~\citep{bartz2017algebraic,bartz2018efficient}. The output list contains the transmitted message polynomial vector.
The algorithm can be considered as a partial unique or list decoder~\citep{bartz2017algebraic,bartz2018efficient}.

\subsection{Folded Subspace Codes}\label{subsec:FoldedSubspace}

We present a family
of folded subspace codes that can be decoded from insertions and deletions beyond the unique decoding region for any code rate $R$.
This class of folded subspace codes is motivated by the constructions by~\citet{Mahdavifar2012Listdecoding,GuruswamiWang2013LinearAlgebraic} and was published by~\citet{BartzSidorenko_FoldedSubspace2015_ISIT}.
\begin{definition}[$\foldPar$-Folded Subspace Code]\label{def:foldedSubspace}
 Let $\pe$ be a primitive element of $\Fqm$.
 An $\foldPar$-folded subspace code \FSub{\foldPar,\pe;\nTransmit,k} of dimension $\nTransmit$, where $\foldPar\nTransmit\leq m$, is defined as the set of subspaces
 \begin{equation*}
  \left\{
    \RowspaceHuge{
  \begin{pmatrix}
   \pe^{0} & f(\pe^{0})  & \dots & f(\pe^{\foldPar-1})
   \\
   \pe^{\foldPar} & f(\pe^{\foldPar})  & \dots & f(\pe^{2\foldPar-1})
   \\
   \vdots & \vdots & \vdots & \vdots
   \\
   \pe^{\foldPar\nTransmit-\foldPar} & f(\pe^{\foldPar\nTransmit-\foldPar}) & \dots & f(\pe^{\foldPar\nTransmit-1})
 \end{pmatrix}
 }
  :f(x)\in\LinpolyringK
  \right\}.
 \end{equation*}

\end{definition}
Defining the column vector $\foldedVec=\left(\pe^0 \ \pe^{\foldPar} \ \dots \ \pe^{\foldPar\nTransmit-\foldPar}\right)^\top$ we can write each codeword of \FSub{\foldPar,\pe;\nTransmit,k} as
\begin{equation}\label{eq:defFSubShort}
  \Rowspace{\left(\foldedVec \ f(\foldedVec) \ f(\pe\foldedVec) \ \dots \ f(\pe^{\foldPar-1}\foldedVec)\right)}
\end{equation}
where $f(x)\in\LinpolyringK$.
The dimension of the ambient vector space
 \begin{equation*}
  \ambSpace = \Rowspace{\vecFq{\foldedVec}} \times \underbrace{\Fq^{m}\times \dots \times \Fq^{m}}_{\foldPar\:\text{times}}
 \end{equation*}
 is $\dimAmbSpace=\nTransmit+\foldPar m$, since the vectors in the space $\Rowspace{\vecFq{\foldedVec}}$ have nonzero components at the $\nTransmit$ known positions $0,\foldPar,2\foldPar,\dots,\foldPar\nTransmit-\foldPar$ only.
 The zeroes at the known positions do not need to be transmitted and can be inserted at the receiver.
 The size of $\FSub{\foldPar,\pe;\nTransmit,k}$ is $|\FSub{\foldPar,\pe;\nTransmit,k}|=q^{mk}$ and the code rate is
\begin{equation}\label{eq:codeRateFSub}
  R=\frac{\log_q\left(|\FSub{\foldPar,\pe;\nTransmit,k}|\right)}{\nTransmit\dimAmbSpace}
  =\frac{km}{\nTransmit(\nTransmit+\foldPar m)}.
\end{equation}
The $\foldPar$-folded subspace codes in Definition~\ref{def:foldedSubspace} are locator-lifted $\foldPar$-folded Gabidulin codes (see Definition~\ref{def:hFoldedGab}).
Combining Theorem~\ref{thm:minimumDistance} and Lemma~\ref{lem:distanceAliftedCodes} we obtain the minimum distance of $\foldPar$-folded subspace codes.
\begin{proposition}[Minimum Distance of $\foldPar$-Folded Subspace Codes]\label{lem:subDistFS}
The minimum subspace distance of an $\foldPar$-folded subspace code \FSub{\foldPar,\pe;\nTransmit,k} is
\begin{equation}
  \SubspacedistNoInput=2\left(\nTransmit-\left\lceil\frac{k}{\foldPar}\right\rceil+1\right).
\end{equation}
\end{proposition}

\subsection{Decoding of Folded Subspace Codes}

Based on the decoding schemes for folded Gabidulin codes~\citep{Mahdavifar2012Listdecoding}, an interpolation-based decoding scheme for folded subspace codes was considered in~\citep{bartz2017algebraic,BartzSidorenko_FoldedSubspace2015_ISIT}.
The algorithm consists of two steps: the \emph{interpolation step} computes $r\leq\intOrder$ non-zero and (left) $\Linpolyring$-independent vectors of linearized polynomials
\begin{equation*}
\vec{Q}^{(i)}\!=\![Q_0^{(i)},Q_1^{(i)},\dots,Q_\intDim^{(i)}] \in \Linpolyring^{\intDim+1} \setminus \{\0\}, \ \forall i=1,\dots,r
\end{equation*}
such that they fulfill certain degree and evaluation conditions with respect to the received subspace.
The \emph{root-finding step} finds all message polynomials $f$ of degrees $\deg f < k$ such that
\begin{equation*}
Q_0^{(i)} + \sum_{j=1}^{u} Q_j^{(i)} f(\pe^{j-1}x) = 0 \quad \forall \, i=1,\dots,r.
\end{equation*}
If the number of insertions $\insertions$ and deletions $\deletions$ satisfies $\insertions+\intDim\deletions<\intDim(\nTransmit-\tfrac{k-1}{\foldPar-\intDim+1})$, then at least one satisfactory interpolation vector $\vec{Q}^{(i)}$ exists, see~\citep{bartz2017algebraic,bartz2018efficient}. The output list contains the transmitted message polynomial vector.
The algorithm can be considered as a partial unique or list decoder~\citep{bartz2017algebraic,bartz2018efficient}.

\section{Upper Bounds on Subspace Codes}
\label{sec:upperbound-size}
In this section, we consider upper bounds on the cardinality and the average list size of constant-dimension subspace codes.
Since we are interested in constructions of codes of maximum size we focus on upper bounds. %
An extensive survey on lower and upper bounds on the size of subspace codes can be found in~\citep{KhaleghiSilvaKschischang-SubspaceCodes_2009}.
We compare the interleaved subspace codes from Section~\ref{subsec:interleavedSubspace} and the folded subspace codes from Section~\ref{subsec:FoldedSubspace} with the bounds and show that the bounds can be asymptotically achieved by the considered codes while keeping the field size low.

Consider a constant-dimension subspace code $\CSub\subset\Grassm{\dimAmbSpace,\nTransmit}$ and define the \emph{normalized weight} $\lambda$, the \emph{code rate} $R$ and the \emph{normalized distance} $\normSubspacedist$ %
as

\begin{equation}\label{eq:normCodeParameters}
  \begin{split}
  \lambda&=\frac{\nTransmit}{\dimAmbSpace}\ , \\
  R&=\frac{\log_q(|\CSub|)}{\dimAmbSpace\nTransmit}=\frac{\log_q(|\CSub|)}{\lambda\dimAmbSpace^2}\ , \\
  \normSubspacedist&=\frac{\Subspacedist{\CSub}}{2\nTransmit}=\frac{\Subspacedist{\CSub}}{2\lambda\dimAmbSpace}. %
\end{split}
\end{equation}

\subsection{Singleton-like Bound for Subspace Codes}\label{subsec:SingletonBound}

The Singleton-like bound for subspace codes~\citep{koetter_kschischang} upper bounds the size of a subspace code $\CSub\in\Grassm{\dimAmbSpace,\nTransmit}$ by
\begin{equation}\label{eq:SingletonSubspace}
  |\CSub|\leq\quadbinom{\dimAmbSpace-(\Subspacedist{\CSub}-2)/2}{\max\left\{\nTransmit,\dimAmbSpace-\nTransmit\right\}}_q.
\end{equation}
The Gaussian coefficient can be lower and upper bounded by~\citep{KK08}
\begin{equation}\label{eq:boundGaussianCoeff}
q^{\ell(N-\ell)}\leq  \quadbinom{N}{\ell}_q\leq 3.5q^{\ell(N-\ell)} .
\end{equation}
Using~\eqref{eq:boundGaussianCoeff} we can upper bound~\eqref{eq:SingletonSubspace} by
\begin{equation}\label{eq:SingletonSubspaceExp}
  |\CSub|\leq 3.5q^{\max\left\{\nTransmit, \dimAmbSpace-\nTransmit\right\}\left(\dimAmbSpace-\Subspacedist{\CSub}/2+1-\max\left\{\nTransmit, \dimAmbSpace-\nTransmit\right\}\right)}.
\end{equation}
For $\nTransmit\leq \dimAmbSpace/2$ we have
\begin{equation*}
 |\CSub|\leq 3.5q^{(\dimAmbSpace-\nTransmit)(\nTransmit-\Subspacedist{\CSub}/2+1)}
\end{equation*}
and can express the Singleton-like bound~\eqref{eq:SingletonSubspaceExp} in terms of normalized parameters as
\begin{align}\label{eq:SingletonSubspaceNormPar}
  \log_q(|\CSub|)&\leq (\dimAmbSpace-\nTransmit)(\nTransmit-\Subspacedist{\CSub}/2+1)+\log_q(3.5)  \nonumber
  \\ \Longleftrightarrow\quad
  R&\leq
  (1-\lambda)\left(1-\normSubspacedist + \frac{1}{\lambda\dimAmbSpace}\right)+\frac{\log_q(3.5)}{\lambda\dimAmbSpace^{2}}.
\end{align}

Notice that for fixed dimension $\nTransmit$ we have $\lambda\sim\frac{1}{\dimAmbSpace}$ for $\dimAmbSpace\gg\nTransmit$.
Thus the term $\log_q(3.5)/(\lambda\dimAmbSpace^{2})$ in~\eqref{eq:SingletonSubspaceNormPar} vanishes asymptotically for $\dimAmbSpace\rightarrow\infty$ with order $1/\dimAmbSpace$.

For a code $\CSub'$ with $\nTransmit'\geq\dimAmbSpace/2$ we have
\begin{equation}
  |\CSub'|\leq 3.5q^{\nTransmit'(\dimAmbSpace-\nTransmit'-\Subspacedist{\CSub'}/2+1)}
\end{equation}
which in terms of the normalized parameters $\lambda'=\nTransmit'/\dimAmbSpace$ and $\normSubspacedist'=\Subspacedist{\CSub'}/(2\nTransmit')$ of $\CSub'$ gives
\begin{equation}\label{eq:SingletonNormParNtBiggerNhalf}
  R'\leq 1-\lambda'-\lambda'\normSubspacedist'+\frac{1}{\dimAmbSpace}+\frac{\log_q(3.5)}{\lambda'\dimAmbSpace^2}.
\end{equation}

For the complementary code $\CSubComp$ (see~\eqref{eq:defComplementaryCode}) of $\CSub$ we have $\nTransmit'=\dimAmbSpace-\nTransmit\geq\dimAmbSpace/2$, $\lambda'=1-\lambda$ and $\normSubspacedist'=\frac{\lambda}{1-\lambda}\normSubspacedist$.
By substituting $\lambda'$ and $\normSubspacedist'$ in~\eqref{eq:SingletonNormParNtBiggerNhalf} we can
write the Singleton bound for $\CSubComp$ in terms of the normalized parameters of $\CSub$ as
\begin{align}\label{eq:SingletonNormParCompRelation}
  R^\perp&\leq \lambda\left(1-\normSubspacedist+\frac{1}{\lambda\dimAmbSpace}\right)+\frac{\log_q(3.5)}{(1-\lambda)\dimAmbSpace^2} \nonumber
  \\
  &=\frac{\lambda}{1-\lambda}\left(  (1-\lambda)\left(1-\normSubspacedist + \frac{1}{\lambda\dimAmbSpace}\right)+\frac{\log_q(3.5)}{\lambda\dimAmbSpace^{2}}\right).
\end{align}
Recall from~\eqref{eq:rateComplementarySubspaceCode} that the code rate $R$ of a subspace code $\CSub$ is related to the code rate $R^\perp$ of the dual code $\CSubComp$ by $R^\perp=\frac{\lambda}{1-\lambda} R$.
From~\eqref{eq:SingletonSubspace} and~\eqref{eq:SingletonNormParCompRelation} we see that the complementary code of a Singleton-bound-achieving code also achieves the Singleton bound.
This relation is analog to the dual codes of Singleton bound achieving codes in the Hamming metric.

\subsection{Anticode Bound}\label{subsec:AnticodeBound}

The \emph{Anticode} bound was proposed by Delsarte for arbitrary association schemes~\citep[p.~32]{DelsarteAnticode73}.
Any subset $\AntiCode{t}$ of $\Grassm{\dimAmbSpace,\nTransmit}$ with $\Subspacedist{\myspace{U},\myspace{V}}\leq 2t$ for all $\myspace{U},\myspace{V}\in\AntiCode{t}$ is called an \emph{Anticode} of diameter $t$.
Let $\CSub\subseteq\Grassm{\dimAmbSpace,\nTransmit}$ be a constant-dimension subspace code.
The Anticode bound implies that
\begin{equation}\label{eq:AntiCodeBoundGeneral}
 |\CSub|
 \leq \frac{|\Grassm{\dimAmbSpace,\nTransmit}|}{|\AntiCode{t-1}|}
 =\frac{|\Grassm{\dimAmbSpace,\nTransmit}|}{\left|\AntiCode{\frac{\Subspacedist{\CSub}-2}{2}}\right|}.
\end{equation}
The bound is tight for the largest Anticode of diameter $t$ in $\Grassm{\dimAmbSpace,\nTransmit}$ which has size $|\AntiCode{t}|=\quadbinom{\dimAmbSpace-\nTransmit+t}{t}_q$ for $\nTransmit\leq\dimAmbSpace/2$~\citep{FranklWilson86}.
Using this result in~\eqref{eq:AntiCodeBoundGeneral} we get
\begin{equation}\label{eq:AnticodeBoundQbinom}
 |\CSub|
 \leq \frac{\quadbinom{\dimAmbSpace}{\nTransmit}_q}{\quadbinom{\nTransmit}{\nTransmit-\frac{\Subspacedist{\CSub}}{2}+1}_q}
 =\frac{\quadbinom{\dimAmbSpace}{\nTransmit-\frac{\Subspacedist{\CSub}}{2}+1}_q}{ \quadbinom{\nTransmit}{\nTransmit-\frac{\Subspacedist{\CSub}}{2}+1}_q}.
\end{equation}
The Anticode bound~\eqref{eq:AnticodeBoundQbinom} was proposed in~\citep[Theorem~1]{Etzion2011ErrorCorrecting} and described in~\citep[Theorem~3]{KhaleghiSilvaKschischang-SubspaceCodes_2009}.
Using~\eqref{eq:boundGaussianCoeff} we can bound~\eqref{eq:AnticodeBoundQbinom} from above by
\begin{equation}\label{eq:AnticodeBoundKK}
 |\CSub|\leq 3.5q^{(\dimAmbSpace-\nTransmit)(\nTransmit-\Subspacedist{\CSub}/2+1)}
\end{equation}
which coincides with the Singleton-like bound~\eqref{eq:SingletonSubspaceExp}.
Thus~\eqref{eq:AnticodeBoundKK} can be expressed in terms of normalized parameters as~\eqref{eq:SingletonSubspaceNormPar}, i.e., we have
\begin{align}\label{eq:AnticodeBoundNormPar}
 R&\leq
 (1-\lambda)\left(1-\normSubspacedist+\frac{1}{\lambda\dimAmbSpace}\right)+\frac{\log_{q}(3.5)}{\dimAmbSpace^{2}\lambda}.
\end{align}

\subsection{Upper Bounds for Interleaved and Folded Subspace Codes}

We now evaluate the bounds from Section~\ref{subsec:SingletonBound} and Section~\ref{subsec:AnticodeBound} for the parameters of interleaved and folded subspace codes.

\subsubsection*{Evaluation of Bounds for Interleaved Subspace Codes}

Consider an interleaved subspace code $\CSub=\IntSub{\intOrder,\vecalpha;\nTransmit,k^{(1)},\dots,k^{(\intOrder)}}$ and define $k_\text{max}\coloneqq\max_{j\in\intervallincl{1}{\intOrder}}\{k^{(j)}\}$.
The normalized parameters~\eqref{eq:normCodeParameters} for $\CSub$ are
\begin{align*}
 \dimAmbSpace&=\nTransmit+\intOrder m\ , \\
 \lambda&=\frac{\nTransmit}{\nTransmit+\intOrder m}\ , \\
 R&=\frac{m\sum_{j=1}^{\intOrder}k^{(j)}}{\nTransmit(\nTransmit+\intOrder m)}\ ,\\
 \normSubspacedist&=\frac{\nTransmit-k_\text{max}+1}{\nTransmit}\ .
\end{align*}
For fixed $\nTransmit$ and $m$ the limit $\dimAmbSpace\rightarrow\infty$ corresponds to $\intOrder\rightarrow\infty$ and we get $\lim_{\intOrder\rightarrow\infty}\lambda=0$ and $\lim_{\intOrder\rightarrow\infty}(R)=\frac{k_\text{max}}{\nTransmit}$ which is the asymptotic code rate of the corresponding interleaved Gabidulin code $\IntGabcode{\intOrder;\nTransmit,k^{(1)},\dots,k^{(\intOrder)}}$ for $\intOrder\rightarrow\infty$.
This illustrates that the rate loss due to the lifting becomes negligible for large $\intOrder$.

Evaluating the Singleton-like bound~\eqref{eq:SingletonSubspace} for the code parameters of $\intOrder$-interleaved subspace codes $\IntSub{\intOrder,\vecalpha;\nTransmit,k^{(1)},\allowbreak\dots,k^{(\intOrder)}}$ gives
\begin{align}\label{eq:SingletonBoundEval_ISub}
  |\CSub|&\leq
           \quadbinom{\nTransmit+\intOrder m-(\nTransmit-k_\text{max}+1)+1}{\nTransmit+\intOrder m-\nTransmit}_q\\
  &=\quadbinom{\intOrder m+k_\text{max}}{\intOrder m}_q\\
  &\leq
  3.5\cdot q^{\intOrder mk_\text{max}}.
\end{align}
Equation~\eqref{eq:SingletonBoundEval_ISub} shows that the size of an interleaved subspace code $\IntSub{\intOrder,\vecalpha;\nTransmit,k^{(1)},\allowbreak\dots,k^{(\intOrder)}}$ (which is $q^{m\sum_{j=1}^{\intOrder}k^{(j)}}$) has the same asymptotic behavior as the Singleton bound if $k^{(j)}=k_\text{max}$ for all $j\in\intervallincl{1}{\intOrder}$, i.e., if the code is a \emph{homogeneous} interleaved subspace code.
A code that achieves the Singleton-like bound in subspace metric can have at most $3.5$ times more codewords than a homogeneous interleaved subspace code $\IntSub{\intOrder,\vecalpha;\nTransmit,k}$ of size $q^{\intOrder mk}$.
The code rate of a \emph{homogeneous} interleaved subspace code $\CSub=\IntSub{\intOrder,\vecalpha;\nTransmit,k}$ in terms of normalized parameters is
\begin{align}\label{eq:rateISubHom}
  \log_q(|\CSub|)= \intOrder mk \nonumber
  \quad&\Longleftrightarrow\quad\frac{\log_q(|\CSub|)}{\dimAmbSpace\nTransmit}= \frac{1}{\dimAmbSpace\nTransmit}(\dimAmbSpace-\nTransmit)(\nTransmit-\Subspacedist{\CSub}/2+1)  \nonumber
  \\
  \quad&\Longleftrightarrow\quad
  R=(1-\lambda)\left(1-\normSubspacedist+\frac{1}{\lambda\dimAmbSpace}\right)
\end{align}
which has the same asymptotic behavior as the Singleton-like bound~\eqref{eq:SingletonSubspaceNormPar} and the Anticode bound~\eqref{eq:AnticodeBoundNormPar} for $\intOrder\rightarrow\infty$ (i.e., $\dimAmbSpace\rightarrow\infty$).

\subsubsection*{Evaluation of Bounds for Folded Subspace Codes}\label{subsubsec:SingletonFSub}

The normalized parameters~\eqref{eq:normCodeParameters} for a folded subspace code $\CSub=\FSub{\foldPar,\alpha;\nTransmit,k}$ are
\begin{align*}
 \dimAmbSpace&=\nTransmit+\foldPar m, \quad
 \lambda=\frac{\nTransmit}{\nTransmit+\foldPar m}, \\\
 R&=\frac{mk}{\nTransmit(\nTransmit+\foldPar m)} \quad\text{and}\quad
 \normSubspacedist=\frac{\nTransmit-\left\lceil k / \foldPar\right\rceil+1}{\nTransmit}.
\end{align*}
Evaluating the Singleton-like bound~\eqref{eq:SingletonSubspace} for these parameters and using~\eqref{eq:boundGaussianCoeff} gives
\begin{equation}\label{eq:SingletonBoundEval_FSub}
  \begin{split}
  |\CSub|
  &\leq\quadbinom{\nTransmit+\foldPar m-(\nTransmit-\left\lceil k / \foldPar\right\rceil+1)+1}{\nTransmit+\foldPar m-\nTransmit}_q\\
  &=\quadbinom{\foldPar m+\left\lceil k / \foldPar\right\rceil}{\foldPar m}_q
  \leq
  3.5\cdot q^{\foldPar m \left\lceil k / \foldPar\right\rceil}.
  \end{split}
\end{equation}
Recall that the size of a folded subspace code is $|\FSub{\foldPar,\alpha;\nTransmit,k}|=q^{mk}$.
Thus the size of $\FSub{\foldPar,\alpha;\nTransmit,k}$ has the same asymptotic behavior as the Singleton-like bound~\eqref{eq:SingletonBoundEval_FSub} if and only if $\foldPar$ divides $k$.
In this case a Singleton bound achieving code can have at most $3.5$ times more codewords than $\FSub{\foldPar,\alpha;\nTransmit,k}$.

Notice that for fixed $\nTransmit$ the degree $m$ of the field $\Fqm$ increases in $\foldPar$ since we require $\foldPar\nTransmit\leq m$.
For $\foldPar\nTransmit = n=m$, where $n$ is the length of the unfolded code of $\FGab{\foldPar,\alpha;\nTransmit,k}$ (see Definition~\ref{def:hFoldedGab}) and $\foldPar>>\nTransmit$,
we have
\begin{equation*}
 R=\frac{k\foldPar\nTransmit}{\nTransmit(\nTransmit+\foldPar^2\nTransmit)}
 =\frac{k\foldPar}{\nTransmit(\foldPar^2+1)}
 \approx\frac{k}{n}
\end{equation*}
which is the code rate of $\FGab{\foldPar,\alpha;\nTransmit,k}$ (see~\eqref{eq:codeRateFGab}).
This shows that the code rate loss due to the lifting is negligible for large $\foldPar>>\nTransmit$.

For the case when $\foldPar$ divides $k$ the code rate of a folded subspace code $\FSub{\foldPar,\alpha;\nTransmit,k}$ in terms of normalized parameters is
\begin{align}\label{eq:rateFSubHom}
  \log_q(|\CSub|)= \foldPar mk \nonumber
  \quad &\Longleftrightarrow\quad
  \frac{\log_q(|\CSub|)}{\dimAmbSpace\nTransmit}= \frac{1}{\dimAmbSpace\nTransmit}(\dimAmbSpace-\nTransmit)(\nTransmit-\Subspacedist{\CSub}/2+1)  \nonumber
  \\
  \quad &\Longleftrightarrow\quad
  R=(1-\lambda)\left(1-\normSubspacedist+\frac{1}{\lambda\dimAmbSpace}\right)
\end{align}
which has the same asymptotic behavior as the Singleton-like bound~\eqref{eq:SingletonSubspaceNormPar} and the Anticode bound~\eqref{eq:AnticodeBoundNormPar} for $\dimAmbSpace\rightarrow\infty$.

Hence, with increasing interleaving order $\intOrder$ and folding parameter $\foldPar$ homogeneous interleaved subspace codes and folded subspace codes where $\foldPar$ divides $k$ show the same asymptotic behavior as the Singleton-like bound~\eqref{eq:SingletonSubspaceNormPar} and the Anticode bound~\eqref{eq:AnticodeBoundNormPar}.

\subsubsection{Comparison of Interleaved and Folded Subspace Codes with Upper Bounds}

We now compare the code rate $R$ of interleaved and folded subspace codes with the Singleton-like bound~\eqref{eq:SingletonSubspaceNormPar} and the Anticode bound~\eqref{eq:AnticodeBoundNormPar}.
Figure~\ref{fig:normDistm6nt6} shows the normalized distance $\normSubspacedist$ over the code rate $R$ for Kötter-Kschischang codes ($\intOrder=\foldPar=1$) and interleaved/folded subspace codes with $\intOrder=\foldPar=3,5,10$.
The Singleton-like bound for the corresponding $\dimAmbSpace$ and $\lambda$ is computed using~\eqref{eq:SingletonSubspaceNormPar}.
Recall that in terms of normalized parameters the Anticode bound~\eqref{eq:AnticodeBoundNormPar} coincides with the Singleton-like bound~\eqref{eq:SingletonSubspaceNormPar}.

With increasing interleaving order $\intOrder$ and folding parameter $\foldPar$ the interleaved and folded subspace codes approach the Singleton-like bound for subspace codes since the normalized weight $\lambda$ decreases for increasing interleaving order~$\intOrder$ or increasing folding parameter~$\foldPar$.

Notice that a noninterleaved subspace code in~\citep{koetter_kschischang} with $M=\intOrder m$ shows the same behavior.
This code has to be decoded in $\mathbb{F}_{q^{M}}$ whereas the interleaved code is decoded in the (sub-) field $\Fqm$ which is in general more efficient.
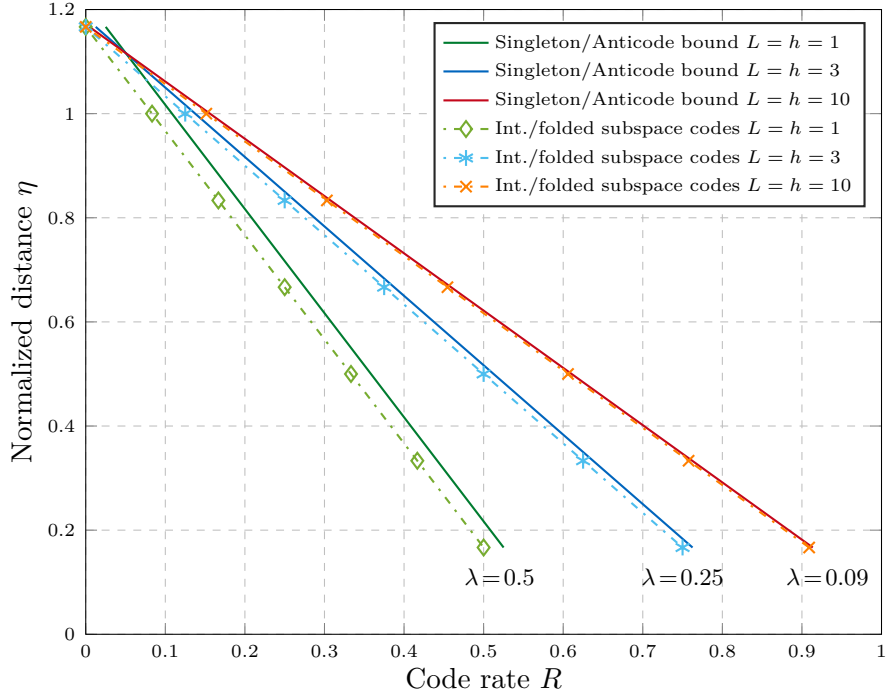
\begin{figure}[t!]
\centering
  \definecolor{mycolor1}{rgb}{0.00000,0.44700,0.74100}%
\definecolor{mycolor2}{rgb}{0.85000,0.32500,0.09800}%
\definecolor{mycolor3}{rgb}{0.92900,0.69400,0.12500}%
\definecolor{mycolor4}{rgb}{0.49400,0.18400,0.55600}%
\definecolor{mycolor5}{rgb}{0.46600,0.67400,0.18800}%
\definecolor{mycolor6}{rgb}{0.30100,0.74500,0.93300}%
\definecolor{mycolor7}{rgb}{0.63500,0.07800,0.18400}%
\begin{tikzpicture}

\begin{axis}[%
xmin=0,
xmax=1,
xlabel={Code rate $R$},
xmajorgrids,
ymin=0,
ymax=1.2,
ylabel={Normalized distance $\normSubspacedist$},
ymajorgrids,
title style={font=\bfseries},
legend style={legend cell align=left,align=left,draw=white!15!black},
mystyle
]

\addplot [color=TUMgreen,solid]
  table[row sep=crcr]{%
0.0251021516952445	1.16666666666667\\
0.108435485028578	1\\
0.191768818361911	0.833333333333333\\
0.275102151695245	0.666666666666667\\
0.358435485028578	0.5\\
0.441768818361911	0.333333333333333\\
0.525102151695245	0.166666666666667\\
};
\addlegendentry{Singleton/Anticode bound $\intOrder=\foldPar=1$};

\addplot [color=TUMblue,solid]
  table[row sep=crcr]{%
0.0125510758476222	1.16666666666667\\
0.137551075847622	1\\
0.262551075847622	0.833333333333333\\
0.387551075847622	0.666666666666667\\
0.512551075847622	0.5\\
0.637551075847622	0.333333333333333\\
0.762551075847622	0.166666666666667\\
};
\addlegendentry{Singleton/Anticode bound $\intOrder=\foldPar=3$};

\addplot [color=TUMred,solid]
  table[row sep=crcr]{%
0.00456402758095348	1.16666666666667\\
0.156079179096105	1\\
0.307594330611257	0.833333333333333\\
0.459109482126408	0.666666666666667\\
0.61062463364156	0.5\\
0.762139785156711	0.333333333333333\\
0.913654936671863	0.166666666666667\\
};
\addlegendentry{Singleton/Anticode bound $\intOrder=\foldPar=10$};

\addplot [color=mycolor5,dash pattern=on 1pt off 3pt on 3pt off 3pt,mark=diamond,mark options={solid}]
  table[row sep=crcr]{%
0	1.16666666666667\\
0.0833333333333333	1\\
0.166666666666667	0.833333333333333\\
0.25	0.666666666666667\\
0.333333333333333	0.5\\
0.416666666666667	0.333333333333333\\
0.5	0.166666666666667\\
};
\addlegendentry{Int./folded subspace codes $\intOrder=\foldPar=1$};

\addplot [color=mycolor6,dash pattern=on 1pt off 3pt on 3pt off 3pt,mark=asterisk,mark options={solid}]
  table[row sep=crcr]{%
0	1.16666666666667\\
0.125	1\\
0.25	0.833333333333333\\
0.375	0.666666666666667\\
0.5	0.5\\
0.625	0.333333333333333\\
0.75	0.166666666666667\\
};
\addlegendentry{Int./folded subspace codes $\intOrder=\foldPar=3$};

\addplot [color=TUMorange,dash pattern=on 1pt off 3pt on 3pt off 3pt,mark=x,mark options={solid}]
  table[row sep=crcr]{%
0	1.16666666666667\\
0.151515151515152	1\\
0.303030303030303	0.833333333333333\\
0.454545454545455	0.666666666666667\\
0.606060606060606	0.5\\
0.757575757575758	0.333333333333333\\
0.909090909090909	0.166666666666667\\
};
\addlegendentry{Int./folded subspace codes $\intOrder=\foldPar=10$};

\node (A) at (axis cs: 0.52, 0.11){\footnotesize{$\lambda\!=\!0.5$}};
\node (B) at (axis cs: 0.75, 0.11){\footnotesize{$\lambda\!=\!0.25$}};
\node (D) at (axis cs: 0.932, 0.11){\footnotesize{$\lambda\!=\!0.09$}};

\end{axis}
\end{tikzpicture}%
  \caption{Normalized distance $\normSubspacedist$ over the code rate $R$ for $q=2,m=6$ and $\nTransmit=6$ and interleaving orders / folding parameters $\intOrder=\foldPar=1,3,10$ with the corresponding normalized weight $\lambda$.}
  \label{fig:normDistm6nt6}
\end{figure}

\subsection{Upper Bound on the Average List Size of Subspace Codes}
\label{subsec:list-decoding-subspace-codes}
Given an $\nReceive$-dimensional received space $\myspace{U}\in\ProjspaceAny{\dimAmbSpace}$ the list decoding problem of a subspace code $\CSub$ is to find the list
\begin{equation}
  \List=\left\{\myspace{V}\in\CSub:\Subspacedist{\myspace{U},\myspace{V}}\leq r\right\}
\end{equation}
where $r$ is the decoding radius in subspace metric.
The challenge of list decoding subspace codes is to decrease the size of the list of candidate codewords, which is exponential in the dimension of the transmitted subspace~\citep{wa13a}.
List decodable variants of subspace codes have been proposed in~\citep{Mahdavifar2010Algebraic, GuruswamiInsertionsDeletions, TrautmannSilbersteinRosenthal-ListDecodingLiftedGabidulinCodes, wachter2014list} and allow to correct insertions and deletions beyond half the minimum subspace distance.
Most list decodable subspace codes are based on locator-lifted Gabidulin codes and control the list size by either restricting the message symbols or the code locators to belong to a subfield.
The list size for this decoder is further reduced in~\citep{GuruswamiWang13Explicit} by restricting the coefficients of the message polynomials to belong to the \emph{hierarchical subspace evasive sets}.
The output of this decoder is a \emph{basis} for the affine space of candidate solutions which in the \emph{worst case} results in a very large list of exponential size in the dimension of the transmitted subspace.
Bounds on the list-decodability of random subspace codes were given in~\citep{DingListRandomRMandSubspace}.

For constant-dimension subspace codes the receiver knows that all codewords have dimension $\nTransmit$.
Recall that the distance between the input and the output of the operator channel \eqref{eq:operatorChannel} with parameters $\insertions$ and $\deletions$ is $\insertions+\deletions$ (see~\eqref{eq:subDistOpChannel}).
We now give an upper bound on the \emph{average list size} of constant-dimension subspace codes, i.e., the average number of codewords that are in subspace distance within $r=\insertions+\deletions$ from an $\nReceive$-dimensional received subspace.
The bound uses ideas for the average list size of Reed-Solomon codes~\citep{McEliece_OntheaveragelistsizefortheGuruswami-Sudandecoder_2003} and Gabidulin codes~\citep{wachter2014list}.

The number of $\nTransmit$-dimensional subspaces in $\ProjspaceAny{\dimAmbSpace}$ at subspace distance at most $r$ from a fixed $\nReceive$-dimensional subspace $\myspace{U}$ in $\ProjspaceAny{\dimAmbSpace}$ is denoted by
 \begin{equation}
   \volSubBallVardy{\nReceive,\nTransmit,r}=|\{\myspace{V}\in\Grassm{\dimAmbSpace,\nTransmit}:\Subspacedist{\myspace{U},\myspace{V}}\leq r\}|.
 \end{equation}
 It is shown in~\citep[Lemma~7]{Etzion2011ErrorCorrecting} $\volSubBallVardy{\nReceive,\nTransmit,r}$ is independent of the center $\myspace{U}$ and given by
\begin{align}\label{eq:volSub}
\volSubBallVardy{\nReceive,\nTransmit,r}
&=\sum_{j=\lceil\frac{\nTransmit+\nReceive-r}{2}\rceil}^{\min\{\nTransmit,\nReceive\}}q^{(\nTransmit-j)(\nReceive-j)}\quadbinom{\nReceive}{j}_q\quadbinom{N-\nReceive}{\nTransmit-j}_q.
\end{align}
For $\nTransmit=\nReceive$ and $r=2t$ this definition coincides with the volume of a sphere in~\citep{koetter_kschischang}
\begin{equation}
  S(\myspace{V},\nReceive,t)=\{\myspace{U}\in\Grassm{\dimAmbSpace,\nReceive}:\Subspacedist{\myspace{U},\myspace{V}}\leq 2t\}
\end{equation}
and we get
\begin{equation*}
  |S(\myspace{V},\nReceive,t)|=\volSubBallVardy{\nReceive,\nReceive,2t}=\sum_{i=0}^{t} q^{i^{2}} \quadbinom{\ell}{i}_q\ \quadbinom{\dimAmbSpace-\ell}{i}_q.
\end{equation*}
For deriving the average list size we need the number of $\nTransmit$-dimensional subspaces within distance at most $\insertions+\deletions$ around a $(\nReceive=\nTransmit+\insertions-\deletions)$-dimensional subspace.
We now derive an upper bound for $\volSubBallVardy{\nTransmit+\insertions-\deletions,\nTransmit,\insertions+\deletions}\eqqcolon\volSubBall{\nTransmit,\insertions,\deletions}$.

\begin{lemma}[Volume of Balls in Subspace Metric]\label{lem:volSubBound}
 The number $\volSubBall{\nTransmit,\insertions,\deletions}$ of $\nTransmit$-dimensional subspaces in $\ProjspaceAny{\dimAmbSpace}$ at subspace distance at most $\insertions+\deletions$ from a fixed $(\nReceive=\nTransmit+\insertions-\deletions)$-dimensional subspace in $\ProjspaceAny{\dimAmbSpace}$ satisfies
 \begin{align}
   \volSubBall{\nTransmit,\insertions,\deletions}
   &=\sum_{i=0}^{\min\{\insertions,\deletions\}}q^{(\deletions-i)(\insertions-i)}
 \quadbinom{\nTransmit+\insertions-\deletions}{\insertions-i}_q
 \quadbinom{N-(\nTransmit+\insertions-\deletions)}{\deletions-i}_q
 \\
 &<16\cdot (\min\{\insertions,\deletions\}+1)\cdot q^{\insertions(\nTransmit-\deletions)+\deletions(\dimAmbSpace-\nTransmit)}.
 \end{align}
 \end{lemma}
The proof of Lemma~\ref{lem:volSubBound} can be found in~\citep[Appendix A.1.1]{bartz2017algebraic}. %

We now derive an upper bound on the average list size of constant-dimension subspace codes.

\begin{theorem}[Average List Size of Subspace Codes]\label{thm:avgListSize}
 Let $\CSub\subseteq\Grassm{\dimAmbSpace,\nTransmit}$ be a constant-dimension subspace code over $\Fq$.
 Let $\myspace{U}$ be an $\nReceive=\nTransmit+\insertions-\deletions$ dimensional subspace chosen uniformly at random from all subspaces in $\Grassm{\dimAmbSpace,\nReceive}$ that are within distance at most $\insertions+\deletions$ to a codeword of $\CSub$.
 The average list size $\avgListSizeCS{\insertions,\deletions}$, i.e., the average number of codewords from $\CSub$ at subspace distance at most $\insertions+\deletions$ from the fixed $(\nReceive=\nTransmit+\insertions-\deletions)$-dimensional subspace $\myspace{U}$, is upper bounded by
 \begin{align}\label{eq:upperBoundAvgListSizeSub}
   \avgListSizeCS{\insertions,\deletions}
   &<1+\frac{|\CSub|-1}{\quadbinom{N}{\nTransmit}_q}\cdot\volSubBall{\nTransmit,\insertions,\deletions}
   <1+ q^{\log_{q}{(|\CSub|)}-(\dimAmbSpace-\nTransmit-\insertions)(\nTransmit-\deletions)}.
 \end{align}
\end{theorem}

\begin{proof}
 By assumption $\myspace{U}$ is chosen uniformly at random from all subspaces in $\Grassm{\dimAmbSpace,\nReceive}$ that are within subspace distance at most $\insertions+\deletions$ to a codeword of $\CSub$.
 There are $|\CSub|-1$ noncausal\footnote{Here noncausal codewords refer to all codewords in a code except for the transmitted codeword.} codewords (subspaces) out of $\quadbinom{N}{\nTransmit}_q$ possible $\nTransmit$-dimensional subspaces.
 Thus there are on average
 \begin{align}\label{eq:avgNumSubSpaces}
  \avgListSizeCSstar{\insertions,\deletions}=\frac{|\CSub|-1}{\quadbinom{N}{\nTransmit}_q}\cdot\volSubBall{\nTransmit,\insertions,\deletions}
  <\frac{|\CSub|}{\quadbinom{N}{\nTransmit}_q}\cdot\volSubBall{\nTransmit,\insertions,\deletions}
 \end{align}
 noncausal codewords within subspace distance at most $\insertions+\deletions$ from the received subspace $\myspace{U}$.
 Using Lemma~\ref{lem:volSubBound} we can upper bound~\eqref{eq:avgNumSubSpaces} by
 \begin{align*}
 \avgListSizeCSstar{\insertions,\deletions}
 &<\frac{|\CSub|}{\quadbinom{N}{\nTransmit}_q}\cdot\volSubBall{\nTransmit,\insertions,\deletions}
 \\
 &<q^{\log_{q}{(|\CSub|)}-\nTransmit(N-\nTransmit)}\cdot 16\cdot(\min\{\insertions,\deletions\}+1)\cdot
 q^{\insertions(\nTransmit-\deletions)+\deletions(\dimAmbSpace-\nTransmit)}
 \\
 &=16\cdot\left(\min\{\insertions,\deletions\}+1\right)\cdot q^{\log_{q}{(|\CSub|)}-\nTransmit(N-\nTransmit)+\insertions(\nTransmit-\deletions)+\deletions(\dimAmbSpace-\nTransmit)}
 \\
 &=16\cdot\left(\min\{\insertions,\deletions\}+1\right)\cdot q^{\log_{q}{(|\CSub|)}-(\nTransmit-\deletions)(\dimAmbSpace-\nTransmit)+\insertions(\nTransmit-\deletions)}
  \\
 &=16\cdot\left(\min\{\insertions,\deletions\}+1\right)\cdot q^{\log_{q}{(|\CSub|)}-(\dimAmbSpace-\nTransmit-\insertions)(\nTransmit-\deletions)}.
 \end{align*}
 Including the causal codeword we get $\avgListSizeCS{\insertions,\deletions}=1+\avgListSizeCSstar{\insertions,\deletions}$.
 \end{proof}
In terms of normalized parameters~\eqref{eq:normCodeParameters} we can write~\eqref{eq:upperBoundAvgListSizeSub} as
\vspace*{-5pt}
\begin{align}\label{eq:avgListSizeNormPar}
  \avgListSizeCS{\insertions,\deletions}
   &<1+ q^{\log_{q}{(|\CSub|)}-((1-\lambda)\dimAmbSpace-\insertions)(\lambda\dimAmbSpace-\deletions)}
   \\
   &=1+ q^{\log_{q}{(|\CSub|)}-\frac{1-\lambda}{\lambda}\left(\lambda\dimAmbSpace-\frac{\lambda}{1-\lambda}\insertions\right)(\lambda\dimAmbSpace-\deletions)}.
   \vspace*{-5pt}
\end{align}
From~\eqref{eq:avgListSizeNormPar} we see that the influence of insertions and deletions on the average list size is \emph{asymmetric}.
The degree of this asymmetry depends on the normalized weight~$\lambda$.
If $\lambda\leq1/2$ (i.e., if $\nTransmit\leq\dimAmbSpace/2$) a \emph{deletion} affects the average list size $(1-\lambda)/\lambda$ times more than an insertions.
For $\lambda>1/2$ (i.e., if $\nTransmit>\dimAmbSpace/2$) an \emph{insertion} affects the average list size $\lambda/(1-\lambda)$ times more than a deletion.

Hence, a code with normalized weight $\lambda\leq\dimAmbSpace/2$ should be more robust against insertions whereas a code with $\lambda>\dimAmbSpace/2$ should be able to tolerate more deletions.

We now evaluate Theorem~\ref{thm:avgListSize} for the parameters of lifted rank-metric codes.
\begin{corollary}[Average List Size of Lifted Rank-Metric Codes]\label{cor:AvgListLiftedRankMetric}
 Let $\CRank\subset\Fq^{\nTransmit\times M}$ be a rank-metric code with code rate $R_r$ and let $\CSub=\liftingMap{\Mat{I}}{\CRank}$.
 Let the received space $\myspace{U}$ be chosen uniformly at random among all subspaces from $\Grassm{\dimAmbSpace,\nReceive}$ that contain a codeword.
 The average list size $\avgListSizeCS{\insertions,\deletions}$, i.e., the average number of codewords within subspace distance at most $\insertions+\deletions$ from an $(\nReceive=\nTransmit+\insertions-\deletions)$-dimensional received subspace $\myspace{U}$, is upper bounded by
 \begin{equation}\label{eq:avgListLiftedRankMetric}
   \abovedisplayskip=8pt
   \belowdisplayskip=6pt
    \avgListSizeCS{\insertions,\deletions}<1+16\cdot (\min\{\insertions,\deletions\}+1)\cdot
    q^{\nTransmit R_r M -(M-\insertions)(\nTransmit-\deletions)}.
 \end{equation}
\end{corollary}

\subsubsection{Average List Size of Interleaved Subspace Codes}

We now estimate the average number of codewords of an interleaved subspace code that are within subspace distance $\insertions+\deletions$ from the received subspace.
By evaluating Corollary~\ref{cor:AvgListLiftedRankMetric} for the parameters of interleaved subspace codes we obtain
 \begin{equation}\label{eq:avgListSizeIS}
   \abovedisplayskip=8pt
   \belowdisplayskip=6pt
    \avgListSizeIS{\insertions,\deletions}<1+16\cdot (\min\{\insertions,\deletions\}+1)\cdot q^{\intOrder(mk-(\nTransmit-\deletions)(m-\frac{\insertions}{\intOrder}))}.
 \end{equation}

Note that if we choose $\nTransmit\approx m$ in~\eqref{eq:avgListSizeIS} we observe an asymmetry between insertions and deletions of degree $\intOrder$, i.e., deletions affect the average list size of {the code} $\intOrder$ times more than insertions.

Consider a homogeneous interleaved subspace code $\CSub=\IntSub{\intOrder,\vecalpha;\nTransmit, k}$.
To decode interleaved subspace codes with a probabilistic unique decoder we require the average list size to be close to one. %
This is fulfilled if the exponent in~\eqref{eq:avgListSizeIS} becomes negative, i.e., if we have
\begin{align*}
  mk-(\nTransmit-\deletions)\left(m-\frac{\insertions}{\intOrder}\right) &< 0
  \\ \Longleftrightarrow\quad
  mk&<\nTransmit m - \nTransmit \frac{\insertions}{\intOrder} -\deletions m + \frac{\deletions\insertions}{\intOrder}
  \\ \Longleftrightarrow\quad
  k&<\nTransmit - \frac{\nTransmit}{m} \frac{\insertions}{\intOrder} -\deletions + \frac{\deletions\insertions}{\intOrder m}.
\end{align*}
For $\nTransmit\approx m$ we get
\begin{align}\label{eq:avgListISub}
  \frac{\insertions}{\intOrder} + \deletions<\nTransmit - k + \frac{\deletions\insertions}{\intOrder m}
  \quad&\Longleftrightarrow\quad
   \insertions + \intOrder\deletions<\intOrder\left(\nTransmit - k\right) + \frac{\deletions\insertions}{m} \nonumber
   \\ \quad&\Longleftrightarrow\quad
   \insertions + \deletions\left(\intOrder - \frac{\insertions}{m}\right) <\intOrder\left(\nTransmit - k\right) \nonumber
   \\ \quad&\Longleftrightarrow\quad
   \insertions + \deletions\left(\intOrder - \frac{\insertions}{m}\right) <\intOrder\left(\frac{\Subspacedist{\CSub}-2}{2}\right).
\end{align}
From~\eqref{eq:avgListISub} we see that a good list decoder for interleaved subspace codes should be able to tolerate approximately $\intOrder$ times more insertions $\insertions$ than deletions $\deletions$ and return on average a list of size close to one if $\insertions$ and $\deletions$ satisfy~\eqref{eq:avgListISub}.

The asymmetry between insertions and deletions is due to the influence of $\intOrder$ on the normalized weight $\lambda=\nTransmit/(\nTransmit+\intOrder m)$ (see~\eqref{eq:avgListSizeNormPar}).

\subsubsection{Average List Size of Folded Subspace Codes}

By Corollary~\ref{cor:AvgListLiftedRankMetric}, the average list size for folded subspace codes is upper bounded by
 \begin{equation}\label{eq:avgListSizeFSubCodes}
   \avgListSizeFS{\insertions,\deletions}<1+16\cdot(\min\{\insertions,\deletions\}+1)\cdot q^{mk-(\nTransmit-\deletions)(\foldPar m-\insertions)}.
 \end{equation}

Let $\CSub=\FSub{\foldPar,\pe;\nTransmit,k}$ and assume that $\foldPar$ divides $k$.
We get an average list size close to one if the exponent in~\eqref{eq:avgListSizeFSubCodes} becomes negative, i.e., if for $\nTransmit\foldPar\approx m$ we have
\begin{align}\label{eq:avgListFSub}
 &mk<(\nTransmit-\deletions)(\foldPar m-\insertions) \nonumber
 \\ \Longleftrightarrow\quad
&mk<\nTransmit\foldPar m - \nTransmit\insertions-\deletions\foldPar m + \deletions\insertions= \nTransmit\foldPar m - \frac{m\insertions}{\foldPar}-\deletions\foldPar m + \deletions\insertions \nonumber
 \\ \Longleftrightarrow\quad
 &\frac{\insertions}{h} + \deletions\foldPar < \nTransmit\foldPar -k + \frac{\deletions\insertions}{m} \nonumber
  \\ \Longleftrightarrow\quad
 &\frac{\insertions}{h} + \deletions\left(\foldPar-\frac{\insertions}{m}\right) < \foldPar\left(\nTransmit -\frac{k}{\foldPar}\right)=\foldPar\left(\frac{\Subspacedist{\CSub}-2}{2}\right).
\end{align}

Compared to interleaved codes~\eqref{eq:avgListISub} the average list size of folded subspace codes~\eqref{eq:avgListFSub} shows an even a higher degree of asymmetry between the tolerable insertions and deletions.
Thus, a good list decoder for $\foldPar$-folded subspace codes should tolerate approximately $\foldPar^2$-times more insertions $\insertions$ than deletions $\deletions$ and return a list of average size close to one if $\insertions$ and $\deletions$ satisfy~\eqref{eq:avgListFSub}.

\chapter{Conclusion}\label{chap:concl}

In this survey, we have presented rank-metric codes and some of their most important applications. Chapter~\ref{ch:introRankMetric} formally introduced several classes of and properties of rank-metric codes with a focus on Gabidulin codes and then summarized a selection of known results on their decoding. Chapter~\ref{chap:crypto} investigated the application of rank-metric codes in cryptography. We briefly recalled a number of rank-metric based cryptosystems and formally defined the respective hard problems as well as some known attacks on these systems. Chapter \ref{chap:storage} focused on applications in storage, first highlighting the role of MRD codes in the construction of codes with locality followed by a brief summary of a coded caching scheme based on Gabidulin codes. Finally, in Chapter~\ref{chap:network_coding} we explored error-correction schemes for network coding utilizing MRD codes.

In the following, we list some further results and applications of rank-metric codes that we have not discussed in detail in this survey.

\subsection*{Further Results on Rank-Metric Codes}

In \citep{ne18sys}, the structure of systematic generator matrices of Gabidulin codes was studied. It was shown that the non-systematic part of these matrices are $q$-analogs of Cauchy matrices, which can be seen as the rank-metric analog result of \citep{roth1985generator} on systematic generator matrices of generalized Reed--Solomon codes.

There is a very simple rank-metric code construction, introduced in \citep{gaborit2017identity} under the name ``simple codes''. These codes are able to correct probabilistic up to a radius that is larger than the one of Gabidulin codes of the same parameters. This comes at the cost of a high decoding complexity (but still polynomial), and the codes have minimum rank distance one, which means that decoding fails already for some (very few) errors of rank one.

Rank-metric codes can also be constructed over infinite fields. 
For instance, \citet{Roth_RankCodes_1991} gives bounds for such codes, and presents a simple construction for rank-metric codes over algebraically closed fields (here, we mean sets of matrices over an algebraically closed field, and the rank is taken w.r.t.\ this field).
In \citep{roth1996tensor,augot2018generalized}, constructions of Gabidulin codes over arbitrary Galois field extensions were given. Decoding of Gabidulin codes over Galois extensions was studied by \citet{robert2016quadratic,muelich2016alternative,augot2018generalized}. For exact computation domains (such as number fields), it is still an open problem to properly analyze the intermediate coefficient growth in these decoding algorithms, e.g., analog to \citep{sippel2019reed} for Reed-Solomon codes. In \citep{roth2017decoding}, decoding of rank-metric codes over algebraically closed fields in \citep{Roth_RankCodes_1991} was studied.

There is also a line of work that studies rank-metric codes over finite rings, mainly motivated by network coding applications. The first work was by \citet{kamche2019rank}, who studied Gabidulin codes over principal ideal rings, their decoding, and applications. \citet{puchinger2021efficient} studied the first decoding algorithm for Gabidulin codes over Galois rings that has a provable quadratic complexity in the code length.
The papers \citep{renner2020low,renner2020lowDCC,djomou2021generalization,kamche2021low} study low-rank parity-check codes over various finite rings.

In \citep{renner2019interleavedlrpc}, interleaved low-rank parity-check codes were defined (similar as vertically and heterogeneously interleaved Gabidulin codes), an efficient decoding algorithm for this code class was devised, and upper bounds on the decoding failure rate of this algorithm were derived.

There are also various papers that study weight distributions and MacWilliams identities for codes in the rank metric \citep{Delsarte_1978,Gabidulin_TheoryOfCodes_1985,gadouleau2008macwilliams,jurrius2015defining,ravagnani2016rank,blanco2018rank,byrne2020rank}. For more details, we refer to the survey by \citet{gorla2018codes}.

In \citep{neri2018genericity}, it was shown that for growing extension degree, $\Fqm$-linear MRD codes become dense in the set of rank-metric codes, i.e., their relative number converges to one. This agrees with the result on the density of $\Fqm$-linear MRD codes in~\citep{byrne2020partition}, which derives upper and lower bounds on the probability that a randomly chosen $\Fqm$-linear $[n,k]$ code has a given minimum $\Fq$-rank distance $d$. In contrast to MDS codes in the Hamming metric and $\Fqm$-linear MRD codes, $\Fq$-linear MRD codes are not dense. \citet{byrne2020partition} showed that the density of $\Fq$-linear MRD codes is asymptotically at most $\frac{1}{2}$, both as $q\to\infty$ and $m\to\infty$. \citet{gluesing-luerssen2020sparseness} showed that $\Fq$-linear $3\times 3$ MRD codes of minimum rank distance $3$ are \textit{sparse}, i.e., its proportion approaches $0$ as $q\to\infty$. In a recent work by~\citet{gruica2020common}, the asymptotic density of $\Fq$-linear $n\times m$ MRD of minimum rank distance $d$ is $O(q^{-(d-1)(n-d+1)+1})$ as $q\to\infty$, which means that $\Fq$-linear MRD codes are also sparse, unless $d=1$ or $n=d=2$.

This survey has discussed block codes in the rank metric. There is also a line of work on convolutional codes in the rank metric \citep{wachter2011partial,wachter2012rank,wachter2015convolutional,napp2017mrd,napp2017column,napp2018faster,almeida2021new}. These codes are considered in the related sum-rank metric, which can be seen as a mix of the Hamming and rank metric.
The sum-rank metric has attracted a lot of attention recently due to promising applications in network coding, distributed data storage, and space-time coding.
For more details on sum-rank-metric codes, we refer to the recent survey by \citet{martines2021survey}.

\subsection*{Further Applications of Rank-Metric Codes}

There are also several further applications of the codes that we have not discussed.
For instance, rank-metric codes can be used in combinatorics~\citep{brewster2020rook}, to construct linear authentication codes~\citep{wang2003linear} and space-time codes \citep{Gabidulin2000Spacetime,Lusina2003Maximum,Lu2004Generalized,AugotLoidreauRobert-GabidulinCharacteristicZero_2013,puchinger2016space,kamche2019rank} (see \citep{martines2021survey} for a recent survey). They have also been used for digital image watermarking \citep{lefevre2019application}, low-rank matrix recovery \citep{forbes2012identity,mu17} (using rank-metric codes over infinite fields \citep{Roth_RankCodes_1991,roth1996tensor,augot2018generalized}), and private information retrieval over networks \citep{TajedineHollantiWZ-2018}.

\begin{acknowledgements}
  The work of L.~Holzbaur and A.~Wachter-Zeh was supported by the German Research Foundation (Deutsche Forschungsgemeinschaft, DFG) under Grant No. WA3907/1-1. This work has also been supported by the European Research Council (ERC) under the European Union’s Horizon 2020 research and innovation programme (grant agreement no.~801434).\\
  H.~Liu has been supported by a German Israeli Project Cooperation (DIP) grant under grant no.~PE2398/1-1 and KR3517/9-1.\\
  This work was done while S.~Puchinger was with the Department of Applied Mathematics and Computer Science, Technical University of Denmark (DTU), Lyngby, Denmark and the Department of Electrical and Computer Engineering, Technical University of Munich, Munich, Germany. He was supported by the European Union’s Horizon 2020 research and innovation program under the Marie Sklodowska-Curie grant agreement no. 713683 and by the European Research Council (ERC) under the European Union’s Horizon 2020 research and innovation programme (grant agreement no. 801434).
\end{acknowledgements}

\backmatter
\printbibliography

\end{document}